\documentclass[11pt]{article}
\usepackage{afterpage}
\usepackage{amsmath}
\usepackage{amsthm}
\usepackage{amssymb}
\usepackage{adjustbox} 
\usepackage{algorithmic}
\usepackage{bbm}
\usepackage{color}
\usepackage{caption,subcaption}
\usepackage{dsfont} 
\usepackage{graphicx}
\usepackage{hyperref}
\usepackage{mathrsfs}
\usepackage[margin=1in]{geometry}
\usepackage{mathtools}
\usepackage{rotating}
\usepackage{tikz}
\usepackage{wrapfig}
\usepackage{setspace}
\usepackage{booktabs} 
\usepackage{accents} 
\usepackage{soul} 
\usepackage{relsize} 
\usepackage{lscape} 
\usepackage{xcolor} 
\setstcolor{red}
\usepackage{multirow}
\usepackage{booktabs,siunitx}

\newcommand{\indep}{\mathbin{\rotatebox[origin=c]{90}{$\models$}}}

\newcommand{\dx}[1]{\ \text{d} #1}
\newcommand{\E}{\mathbb{E}}

\newcommand{\indicator}[1]{\mathds{1}\left\{ #1 \right\}}

\newcommand{\X}{\mathbf{X}}

\newcommand{\bt}{\mathbf{t}}

\newcommand{\Y}{\mathbf{Y}}
\newcommand{\x}{\mathbf{x}}

\newcommand{\bL}{\mathbf{L}}
\newcommand{\bl}{\mathbf{l}}
\newcommand{\bLi}{\mathbf{L}_{(i)}}
\newcommand{\bli}{\mathbf{l}_{(i)}}

\newcommand{\bT}{\mathbf{T}}
\newcommand{\DE}{DE}
\newcommand{\IDE}{IDE}
\newcommand{\AR}{AR}
\newcommand{\VE}{VE}

\newcommand{\CE}{CE}
\newcommand{\SE}{SE}
\newcommand{\IE}{IE}

\newcommand{\Xminusi}{\mathbf{X}_{(i)}}
\newcommand{\xminusi}{\mathbf{x}_{(i)}}
\newcommand{\xminusj}{\mathbf{x}_{(j)}}
\newcommand{\xik}{\mathbf{x}_{(i)}^\mathbf{k}}

\DeclareMathAlphabet\mathbfcal{OMS}{cmsy}{b}{n}
\newcommand{\bH}{\mathbfcal{H}}
\newcommand{\bh}{\mathbf{h}}
\newcommand{\bHi}{\mathbfcal{H}_{(i)}}
\newcommand{\bhi}{\mathbf{h}_{(i)}}

\newcommand{\bhik}{\mathbf{h}_{(i)}^\mathbf{k}}

\newcommand{\redsout}{\bgroup\markoverwith{\textcolor{red}{\rule[0.5ex]{2pt}{0.4pt}}}\ULon}

\newtheorem{thm}{Theorem}
\newtheorem{lem}{Lemma}
\newtheorem{defn}{Definition}
\newtheorem*{defn*}{Definition}
\newtheorem{prop}{Proposition}
\newtheorem{cor}{Corollary}

\newtheorem{assumption}{Assumption}
\newtheoremstyle{break}
  {\topsep}{\topsep}%
  {\itshape}{}%
  {\bfseries}{}%
  {\newline}{}%
\theoremstyle{break}

\usepackage[numbers,sort&compress]{natbib}

\title{Causal identification of infectious disease intervention effects \\ in a clustered population}

\author{Xiaoxuan Cai$^1$, Eben Kenah$^2$, and Forrest W. Crawford$^{1,3,4,5}$ \\[1em]
 \normalsize 1. Department of Biostatistics, Yale School of Public Health \\
 \normalsize 2. Division of Biostatistics, College of Public Health, The Ohio State University \\
 \normalsize 3. Department of Statistics \& Data Science, Yale University \\
 \normalsize 4. Department of Ecology and Evolutionary Biology, Yale University \\
 \normalsize 5. Yale School of Management}
\date{\today}

\begin{document}
\maketitle


\begin{abstract}
\noindent Causal identification of treatment effects for infectious disease outcomes in interconnected populations is challenging because infection outcomes may be transmissible to others, and treatment given to one individual may affect others' outcomes. Contagion, or transmissibility of outcomes, complicates standard conceptions of treatment interference in which an intervention delivered to one individual can affect outcomes of others.    
  For example, a vaccine given to an individual may affect their risk of infection given exposure to an infectious source as well as their infectiousness if they become infected.  
  Several statistical frameworks have been proposed to measure causal treatment effects in this setting, including structural transmission models, mediation-based partnership models, and randomized trial designs. 
  However, existing estimands for infectious disease intervention effects are of limited conceptual usefulness: 
  Some are parameters in a structural model whose causal interpretation is unclear, others are causal effects defined only in a restricted two-person setting, and still others are nonparametric estimands that arise naturally in the context of a randomized trial but may not measure any biologically meaningful effect.  
  In this paper, we describe a unifying formalism for defining nonparametric structural causal estimands and an identification strategy for learning about infectious disease intervention effects in clusters of interacting individuals when infection times are observed. 
  The estimands generalize existing quantities and provide a framework for causal identification in randomized and observational studies, including situations where only binary infection outcomes are observed.  
  A semiparametric class of pairwise Cox-type transmission hazard models is used to facilitate statistical inference in finite samples.  
  A comprehensive simulation study compares existing and proposed estimands under a variety of randomized and observational vaccine trial designs. \\[1em]
\end{abstract}



\section{Introduction}

Infectious diseases represent a leading cause of illness, disability, and death worldwide~\citep{world2009global,cdc2011}.  
Major successes in the control and prevention of infectious diseases have been achieved by vaccination programs, and evaluation of vaccine efficacy is of major interest for epidemiologists and statisticians~\citep{halloran2010design}.  
However, when trial participants come from the same population, the infection outcomes of study participants may not be independent, posing problems for definition and identification of causal vaccine effects.   
This dependence arises from two distinct sources.  
First, infection itself---with or without vaccination---may be transmissible between subjects. 
Second, vaccination of an individual may alter both their susceptibility to infection and their infectiousness if infected.
The complex interaction of vaccination patterns with individual and group-level risks of infection makes causal inference for vaccine effects a uniquely challenging problem.

Statisticians and infectious disease epidemiologists have devoted a great deal of effort to studying this problem.  
Causal dependence in infectious disease outcomes arises in several different but overlapping forms, including herd immunity~\citep{longini1988statistical,struchiner1990behaviour,halloran1991direct,o2014estimating}, spillover \citep{sobel2006randomized,vanderweele2011bounding}, dependent happenings~\citep{struchiner1990behaviour,halloran1991study,halloran1991direct,halloran1995causal,eisenberg2003bias,hudgens2008toward,vanderweele2011bounding,o2014estimating}, interference~\citep{cox1958,halloran1991study,sobel2006randomized,rosenbaum2007interference,hudgens2008toward,vanderweele2010direct,vanderweele2011bounding,tchetgen2012causal,o2014estimating}, and indirect effects~\citep{halloran1991study,halloran1991direct,halloran1995causal,rosenbaum2007interference,pitzer2012linking,vanderweele2011bounding,vanderweele2012components,ogburn2017vaccines}.  
Specifically, treatment of one individual may affect the distribution of their infection time when they are exposed to an infectious source.
If this individual becomes infected, their infection may subsequently increase the risk of infection in others, and this change in risk may depend on treatment in both the infectious and susceptible individuals.
This dependence between infection outcomes creates identification problems in causal inference for intervention effects. 
Using principles of causal reasoning \citep{halloran1995causal,sobel2006randomized,vanderweele2011bounding,o2014estimating,eck2019randomization,cai2021identification}, evidence from simulation \citep{struchiner1990behaviour,koopman1991assessing,halloran1994exposure,eisenberg2003bias,staples2015incorporating,morozova2018risk}, and the behavior of dynamic transmission models \citep{halloran1994exposure,halloran1997study,longini1982estimating,koopman1991assessing,eisenberg2003bias,koopman2004modeling,pitzer2012linking,rhodes1996counting,morozova2018risk}, researchers have shown that ignoring this dependence may result in biased estimates of causal vaccine effects. 

Some researchers have attempted to circumvent this problem by randomization, proposing design-based inferential frameworks for estimating vaccine effects. 
Halloran and colleagues proposed a two-stage block randomization design and formalized the ``direct'' and ``indirect'' effects of treatment~\citep{struchiner1990behaviour,halloran1991direct,halloran1991study,halloran1995causal,hudgens2006causal,hudgens2008toward}. 
Following ideas from mediation analysis, \citet{vanderweele2011bounding,vanderweele2012components,ogburn2017vaccines} and \citet{halloran2012causal} restricted the identification problem to an asymmetric partnership setting (i.e., pairs where infection can be transmitted in only one direction), and they defined a causal direct effect, contagion effect, and infectiousness effect. 
\citet{hudgens2006causal} and \citet{halloran2012causal} applied a principal stratification strategy to define a causal infectiousness effect, by contrasting the secondary attack risks of individuals exposed to treated versus untreated infected neighbors, providing a non-parametric bound for this infectiousness effect. 
These methods do not require complete knowledge of the causal structure between outcomes, but they also ignore individual characteristics and the dynamics of of infectious disease transmission. 
\citet{cai2021identification} study causal identification in the symmetric partnership setting, where infection can be transmitted in both directions. 
Their continuous-time strategy clarifies the causal structure of contagion and vaccine effects in groups of two, when both individuals have covariates, treatments, and can transmit infections to each other, but its usefulness is limited to the partnership scenario.  The identification problem remains open for clusters of size greater than two.

Statistical measures of association have long been used to summarize vaccine effects. 
\citet{halloran1991direct,halloran1994exposure,halloran1995causal,rhodes1996counting,halloran1997study,o2014estimating} described summary statistics based on attack risks, secondary attack risks, and per-contact-per-time conditional (secondary) attack rates of different treatment groups. 
However, the causal interpretation of these quantities depends on assumptions of homogeneous mixing \cite{halloran1991direct,halloran1997study,rhodes1996counting}, equal contact rates \cite{halloran1997study,rhodes1996counting}, constant prevalence of infection \cite{halloran1997study}, comparable exposure to infection between treated and untreated \cite{halloran1991direct,halloran1995causal,rhodes1996counting}, or detailed information about ``who-infected-whom'' \cite{longini1982estimating,rhodes1996counting,halloran1995causal}, and it remains unclear what mechanistic or causal features of the infection process are revealed by these measures. 

Mechanistic models serve as another major tool to represent transmission dynamics of an infectious disease, but it is unclear whether parameters in these models are interpretable as causal effects. 
\citet{longini1982estimating,longini1988statistical} and \citet{longini1982household} proposed the chain binomial model for longitudinal data and characterized disease transmission using two terms: exogenous risk from outside the cluster and endogenous risk per infective-susceptible contact within the cluster. 
However, the chain binomial model only applies to longitudinal data of fixed time intervals, and it can be difficult for these models to deal with individual heterogeneity. 
\citet{halloran1994exposure,halloran1997study} and \citet{rhodes1996counting} further consider models for infection hazards on continuous timescale, also assuming homogeneous mixing~\citep{halloran1991direct,rhodes1996counting,halloran1997study}, equal contact rates \cite{rhodes1996counting,halloran1997study}, constant prevalence \cite{halloran1997study}, or detailed information about ``who-infected-whom''\cite{longini1982estimating,rhodes1996counting}.
\citet{kenah2011contact,kenah2015semiparametric,kenah2013non} and \citet{kenah2008generation} proposed pairwise hazard models for transmission, which account for individual covariates and transmission dynamics over time.  However, it remains unclear whether parameters (i.e. a coefficient on the vaccination status of an individual) in mechanistic transmission models may have a causal interpretation. 

In this paper, we describe a nonparametric causal framework for defining and identifying intervention effects for transmissible infectious disease outcomes in clusters. 
These methods incorporate individual treatments and covariates that may affect both susceptibility and infectiousness. 
We propose new causal estimands---contagion, susceptibility, and infectiousness effects---using contrasts of potential infection outcomes under different treatment conditions.  
Two key insights distinguish this approach from others: 
First, we define potential infection outcomes holding constant infection times of other individuals, thereby blocking feedback from the infection status of a given individual to that of others who might infect them. 
Second, we define marginal effects by integrating over the infection times of others using a variant of Robins's G-computation formula~\citep{robins1986new} to obtain average potential infection outcomes under the infection risk an individual would experience \emph{if they were unable to infect others}. 
We describe a semiparametric hazard framework developed by~\citet{kenah2008generation} and \citet{kenah2011contact,kenah2013non,kenah2015semiparametric} for statistical estimation of parameters and causal effects.  The performance of the proposed causal estimands is illustrated and compared with several other methods.  We conclude with practical recommendations for targeting meaningful causal intervention effects for infectious disease outcomes. 


\section{Setting and notation}

Consider a population of independent clusters of individuals susceptible to infection by an infectious disease. 
Following the conventional partial interference assumption~\citep{sobel2006randomized,hudgens2008toward}, we assume there is no transmission between clusters. 
We develop notation and definitions for a generic cluster of $n$ individuals.  Let $T_i$ be the random time at which $i$ becomes infected, and let $\bT=(T_1, \ldots, T_n)$ be the vector of infection times of the cluster. 
Let $Y_i(t)= \indicator{T_i \leq t}$ be the binary indicator of prior infection of $i$ at time~$t$, and let $\Y(t) = \big(Y_1(t), \ldots, Y_n(t)\big)$ be the vector of infection indicators at time $t$ of the cluster. 
Let $X_{i}$ be the binary treatment (e.g. vaccination) status of $i$, and let $\X = (X_{1},\ldots,X_{n})$ be the treatment vector for the cluster. 
Similarly, let $\bL_i$ be the vector of measured covariates of $i$, and let $\bL = (\bL_1,\ldots,\bL_n)$ be the covariates for the cluster.  
Denote the vector $\bT$ excluding $T_i$ as $\bT_{(i)}=(T_1, \ldots, T_i, T_{i+1}, \ldots, T_n)$, and define the vectors $\Y_{(i)}(t)$, $\Xminusi$, and $\bL_{(i)}$ similarly. 

It is useful to represent the history of infections in the cluster up to a particular time~$t$. 
For a time $t > 0$, let $H_i(t) = \{Y_i(s):0 \le s<t\}$ be the infection history of individual~$i$ up to time $t$, with a realization denoted by $h_i(t)$, and let $\bH(t) = \{Y_i(s): 0 \le s < t, i = 1,\ldots,n\}$ be the binary infection history of the cluster up to time $t$, with a realization denoted by $\bh(t)$. Specifically, $\bH(t)$ represents the pattern of infections that have occurred in the cluster, up to time $t$. 
Let $\bHi(t) = \{Y_j(s): 0 \le s < t, j \neq i\}$ be the infection history of all individuals except~$i$ up to time $t$, with its realization $\bhi(t)$. 
To specify a complete infection history for all times $0\le t<\infty$, we will write $H_i = H_i(\infty)$, $\bH = \bH(\infty)$, and $\bHi = \bHi(\infty)$. 
We can partition the joint infection history in the cluster into $\bH(t)  = \big(H_i(t), \bHi(t)\big)$.

Informally, every individual's potential infection time depends on their own treatment as well as the treatments and infection times of others in the same cluster.
To express these simultaneous relationships, we define the potential infection outcome $Y_i(t)$ (or $T_i$) under the joint treatment $\X=\x$ and others' infection histories $\bHi = \bhi$. 
While $Y_i(t)$ and $T_i$ may depend on $\bHi$, the history $\bHi$ also depends on $Y_i(t)$ (or equivalently, $T_i$ or $H_i$) because infection of $i$ prior to $t$ might have caused infection of other individuals, affecting $\bHi(t)$. 
For this reason, we will define the potential infection outcomes $T_i\left(\x, \bhi\right)$ and $Y_i(t;\x,\bhi)=\indicator{T_i(\x,\bhi) \le t}$ under joint treatments~$\x$ and a \textit{deterministic} infection history~$\bhi$, where $\bhi$ is chosen in advance and cannot be affected by $T_i$ or $Y_i(t)$. 
By fixing the infection histories $\bhi$ of individuals other than $i$, we block feedback from the infection time of $i$ to the infection histories of others.

Following the strategy developed for two-person partnerships by~\citet{cai2021identification}, we decompose the potential infection time $T_i(\x, \bhi)$ of a given individual $i$ into a series of latent infection waiting times following the infections of individuals other than $i$.  
To illustrate, let $(t_{(i)}^{1},t_{(i)}^{2},\ldots,t_{(i)}^{n-1})$ be the ordered history of others' infection times specified by $\bhi$, where $t_{(i)}^k$ is the $k^\text{th}$ infection time among individuals other than $i$, ordered from earliest to latest. 
Denote the individual whose infection time is $t_{(i)}^k$ by $\varphi^k_i$, indexing the $k^{\text{th}}$ infected cluster member for $i$.  
Similarly, we rearrange the treatments for individuals except $i$ by the order of their infection times as $(x_{(i)}^{1}, \ldots, x_{(i)}^{n - 1})$, so that $x_{(i)}^{k} = x_{\varphi^k_i}$. 
For ease of notation, we denote $t_{(i)}^0=0$ and $t_{(i)}^{n}=\infty$, for all $i$.
The necessary reordering is contained in the history $\bhi$.

Because individuals are assumed to be uninfected at time $0$, let $I_i^0(\x,\bhi)$ be the potential time to infection of $i$ when no other cluster members are infected.  
If $i$ is infected before $t_{(i)}^{1}$, then $i$ is the first cluster member to be infected, and $i$ realizes the infection time $T_i(\x,\bhi)=I_i^0(\x,\bhi)$. 
Otherwise, the initial infection time $I_i^0(\x,\bhi)$ is censored by the infection of individual $\varphi^1_i$ at time $t_{(i)}^{1}$.  
If $I_i^0(\x,\bhi)$ is censored at $t_{(i)}^1$, let $I_i^1(\x,\bhi)$ be the additional time to infection of $i$ beyond $t_{(i)}^1$. 
If $i$ is subsequently infected before the next infection among its fellow cluster members, which occurs in individual $\varphi^2_i$ at time $t_{(i)}^2$, then $i$ realizes the infection time $T_i(\x,\bhi)=t_{(i)}^1 + I_i^1(\x,\bhi)$. 
Otherwise, the waiting time $I_i^1(\x,\bhi)$ is censored by infection of $\varphi^2_i$ at $t_{(i)}^{2}$.  
Following the same reasoning, let $I_i^k(\x,\bhi)$ be the additional time to infection of $i$ after $t_{(i)}^{k}$. 
If $i$ is infected before the $(k+1)^{\text{th}}$ infection in fellow cluster members at $t_{(i)}^{k+1}$, then $i$ realizes the infection time $T_i(\x,\bhi)=t_{(i)}^k + I_i^k(\x,\bhi)$, and otherwise $I_i^k(\x,\bhi)$ is censored. 
Finally, if $i$ is last to be infected, then $I_i^{n-1}(\x,\bhi)$ is the remaining time to infection beyond $t_{(i)}^{n-1}$.  
We can rewrite $T_i(\x,\bhi)$ in terms of its relationship to others' specified infection times in $\bhi$: 
\begin{multline}
    T_i(\x,\bhi) =  
      \begin{cases}
        I^0_i(\x,\bhi) & \textrm{if } I^0_i(\x,\bhi) < t_{(i)}^{1} \\
        t_{(i)}^{1} + I^1_i(\x,\bhi) & \textrm{if } I^0_i(\x,\bhi) \ge t_{(i)}^{1}, I^1_i(\x,\bhi) < t_{(i)}^{2}-t_{(i)}^{1}\\
        \vdots  & \vdots \\
        t_{(i)}^{n-1}+ I^{n-1}_i(\x,\bhi) & \textrm{if }I^0_i(\x,\bhi) \ge t_{(i)}^{1},\ \ldots,\ I^{n-2}_i(\x,\bhi) \ge t_{(i)}^{n - 1}-t_{(i)}^{n - 2}    
      \end{cases} \\[0.7em] 
      = \sum_{j=0}^{n-1} \left[\left( t_{(i)}^{j}+I_i^j(\x,\bhi) \right) \indicator{I_i^{j}(\x,\bhi) < t_{(i)}^{j+1}-t_{(i)}^{j}} \prod_{k=0}^{j-1} \indicator{ I_i^k(\x,\bhi) \ge t_{(i)}^{k+1}-t_{(i)}^{k}} \right]
    \label{eq:Tdecomposition}
\end{multline}
We emphasize that the decomposition in \eqref{eq:Tdecomposition} is purely notational and places no restrictions on the joint distribution of infection times in the cluster. 
Rather, the decomposition permits specification of causal assumptions for the latent times $I_i^k(\x,\bhi)$ that permit identification of intervention effects on potential infection outcomes.  
Because $\bhi$ is a deterministic infection history, stochasticity in $T_i(\x,\bhi)$ is entirely due to the random waiting times $I_i^k(\bhi,\x)$. 
Figure \ref{fig:Tdefinition} illustrates the decomposition of $T_i(\x,\bhi)$ in terms of the waiting times $I_i^k(\x,\bhi)$ in \eqref{eq:Tdecomposition}.

\begin{figure}
  \centering
  \includegraphics[width=\textwidth]{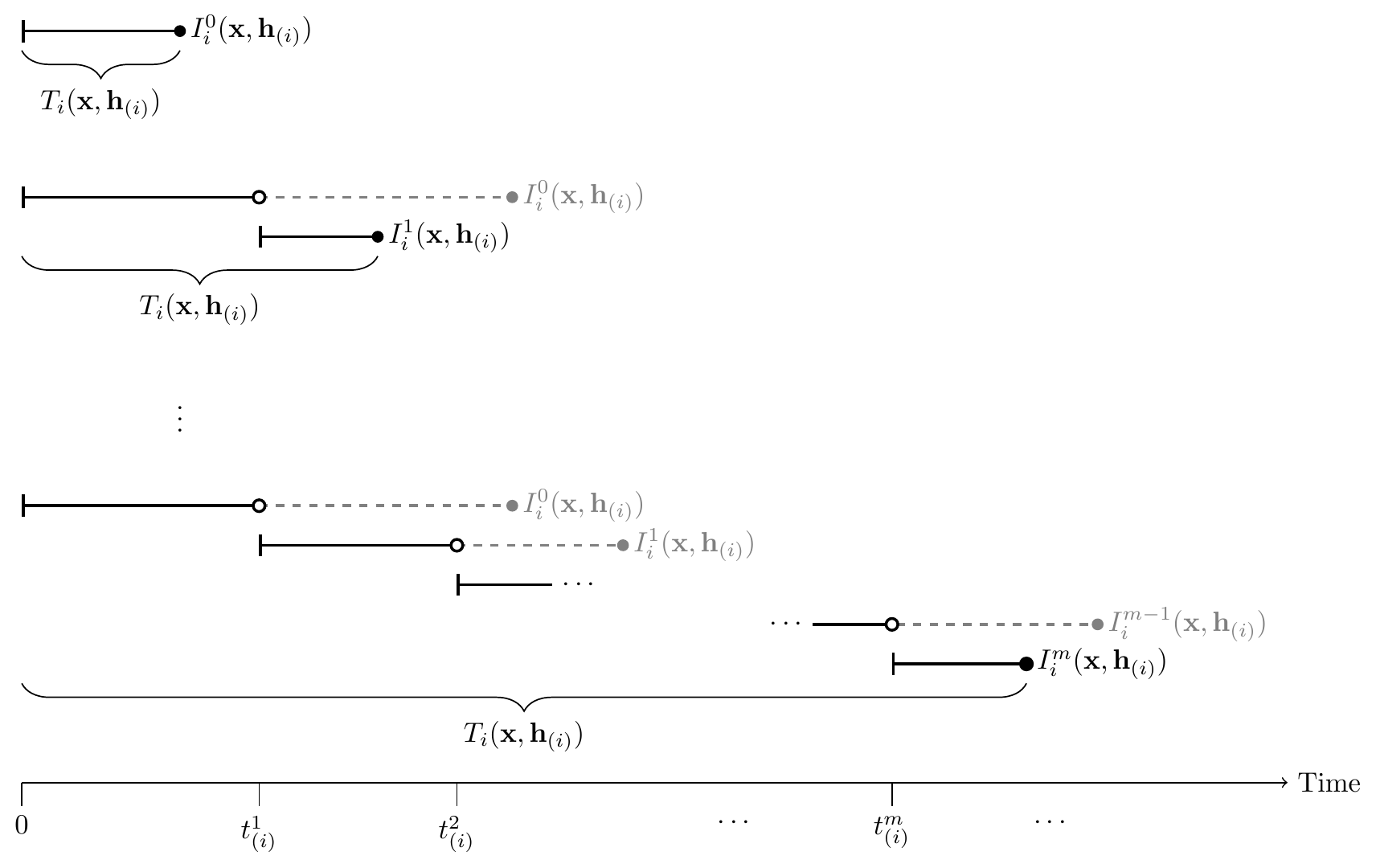}
  \caption{Illustration of the decomposition of the infection time $T_i(\x,\bhi)$ of $i$ into a sequence of waiting times according to \eqref{eq:Tdecomposition} given infection history $\bhi$.  
    If $I_i^0(\x,\bhi)$ occurs before $t_{(i)}^1$, then $T_i(\x,\bhi)=I_i^0(\x,\bhi)$ is observed. 
    If $I_i^0(\x,\bhi) \ge t_{(i)}^1$, then $I_i^0(\x,\bhi)$ is censored, and the waiting time $I_i^1(\x,\bhi)$ elapses following $t_{(i)}^1$. 
    If $I_i^1(\x,\bhi)<t_{(i)}^2$, then $T_i(\x,\bhi)=t_{(i)}^1+I_i^1(\x,\bhi)$ is observed.
 	If $I_i^1(\x,\bhi) \ge t_{(i)}^2$, then $I_i^1(\x,\bhi)$ is censored, and the waiting time $I_i^2(\x,\bhi)$ elapses following $t_{(i)}^2$.
    Likewise, if $I_i^k(\x,\bhi)$ is censored for every $k<n-1$, then $T_i(\x,\bhi)= t_{(i)}^{n-1} + I_i^{n-1}(\x,\bhi)$ is observed following $t_{(i)}^{n-1}$.}
  \label{fig:Tdefinition}
\end{figure}

If our interest is in evaluating potential infection outcomes as a function of treatments alone, we may view the history $\bhi$ as a nuisance variable, to be averaged out.
If we treat $\bhi$ as a realization of the random infection times $\bHi$ of others, what distribution should it have?  
In non-interventional settings, the distribution of $\bHi$ depends on $H_i$ (equivalently, $T_i$), because $i$ might infect others following their own infection. 
Therefore the marginal distribution of $\bHi$ is not of interest as a marginalizing distribution for the potential infection outcome of $i$, because it contains information about the outcome of $i$.  
Instead, we block this feedback by constructing an alternative marginalizing variable, denoted $\bHi^*$, whose distribution is invariant to the infection time of $i$.

Let $\bHi^*(t)$ be the random history of infection times in individuals other than $i$ up to time $t$ under the condition that $i$ cannot transmit infection to any other cluster member. 
Let $\bHi^*(t;\x_{(i)})$ be the corresponding potential infection history under this circumstance. 
The alternative history $\bHi^*(t;\x_{(i)})$ has a special connection to the observational infection history $\bH_{(i)}(\x)$: Before $i$ is infected, $i$ cannot transmit infection to others, so $\bHi(\x)$ is equivalent to $\bHi^*(\xminusi)$ in distribution. 
Once $i$ is infected and can transmit infection to others, this alters the distribution of $\bHi(\x)$ but not $\bHi^*(\xminusi)$, so they are no longer equal in distribution. 
Since we are evaluating the potential infection outcome of $i$, the difference in the evolution of the cluster-level infection after the infection of $i$ is inconsequential.
Therefore, $\bHi^*(\xminusi)$, the infection history in an otherwise identical group where $i$ cannot transmit infection, is equal in distribution to the observed $\bHi(\x)$ before the infection of $i$, and it imitates the observable exposure to infection for $i$ due to others' infections. 
We formalize this intuition in identification results presented in the next Section. 


\section{Identification of average potential infection outcomes}

\subsection{Assumptions}

By the end of the observation at time $t$, we observe $\big(\X, \bL, \bH(t)\big)$.  
We state the following assumptions for a generic individual $i$.
\begin{assumption}[Exclusion restriction for future infection times] 
  For $k = 1, \ldots, n - 1$, we have
  $I_i^0(\x,\bhi)=I_i^0 (x_i)$ and 
  $I_i^k(\x,\bhi)=I_i^k \big(x_i, \xminusi, \bhi(t_{(i)}^k)\big)$.
\label{as:exclusion}
\end{assumption}
Assumption~\ref{as:exclusion} states that the latent time to infection of $i$ following the $k$th infection time $t_{(i)}^k$ depends only on the treatments $\x$ assigned at baseline and others' infection history up to $t_{(i)}^k$, $\bhi(t_{(i)}^k)$.
It does not depend on the infection history after $t_{(i)}^k$.
\begin{assumption}[Independence of latent times to infection]
  For a given history $\bhi(t_{(i)}^k)$ up to the $k$th infection and each $j = \varphi_i^{k + 1}, \ldots, \varphi_i^{n - 1}$ infected after time $t_{(i)}^k$,
  \begin{equation*}
    I_i^k(x_i, \xminusi, \bhi (t_{(i)}^k)) 
    \indep I_j^{k}(x_j, \xminusj, \bhi (t_{(i)}^k)) \mid \bL.
  \end{equation*}
\label{as:independence}
\end{assumption}
Assumption~\ref{as:independence} states that under a given infection history up to the $k^{\text{th}}$ infection, the latent time to infection for individual $i$ is conditionally independent of the latent time to infection for any other uninfected individual $j$ given $\bL$.  In other words, the competing (simultaneously elapsing) infection waiting times $I_i^k(x_i, \xminusi, \bhi (t_{(i)}^k))$ and $I_j^{k}(x_j, \xminusj, \bhi (t_{(i)}^k))$ are conditionally independent. 
\begin{assumption}[Previous exposure exchangeability]
  For all $t > 0$ and each $k = 0, \ldots, n - 1$,
\begin{equation*}
    I_i^k(x_i,\xminusi,\bhi(t_{(i)}^k)) 
    \indep \bHi^{*}(t_{(i)}^k; \xminusi) \mid \bL.
\end{equation*}

\label{as:prevexchangeability}
\end{assumption}
Assumption~\ref{as:prevexchangeability} states that the latent time to infection $I_i^k\big(x_i,\xminusi,\bhi(t_{(i)}^k)\big)$ under the deterministic infection history $\bhi(t_{(i)}^k)$ up to $t_{(i)}^k$ is conditionally independent of the random infection history $\bHi^{*}(t_{(i)}^k; \xminusi)$ given $\bL$. 
If we interpret $\bHi^{*}(t;\xminusi)$ as the exposure of $i$ to infection up to time $t$, then this assumption corresponds to no unmeasured confounding for the effect of prior exposure to infection on the latent time to infection.
In this way, Assumption \ref{as:prevexchangeability} is a conventional treatment exchangeability assumption, where the ``treatment'' is prior exposure to infection. 
\begin{assumption}[Treatment exchangeability]
  For $k = 0, 1, \ldots, n - 1$,
  \begin{equation*}
    I_i^k(x_i, \xminusi, \bhi(t_{(i)}^k) \big) \indep \X \mid \bL.
  \end{equation*}
  \label{as:txexchangeability}
\end{assumption}
Assumption~\ref{as:txexchangeability} states that the potential latent time to infection $I_i^k\big(x_i,\xminusi,\bhi(t_{(i)}^k)\big)$ is conditionally independent of the joint treatment assignment $\X$ given $\bL$.
This corresponds to no unmeasured confounding of the effect of treatment on the latent time to infection.
\begin{assumption}[Consistency]
  When $X_i = x_i$, $\Xminusi = \xminusi$, $\bHi(t_{(i)}^k) = \bhi(t_{(i)}^k)$ and $T_i > t_{(i)}^k$,
  \begin{equation*}  
    I_i^k=I_i^k(x_i, \xminusi,\bhi(t_{(i)}^k)).
  \end{equation*}
\label{as:consistency}
\end{assumption}
Assumption~\ref{as:consistency} states that $I_i^k(x_i, \xminusi,\bhi(t_{(i)}^k))$ is the realized latent time to infection when the observed treatments and infection history up to $t_{(i)}^k$ match those specified and $i$ remains susceptible at $t_{(i)}^k$. 
\begin{assumption}[Positivity]
  $\Pr\big(I_i^k \le t \,\big|\, \X = \x, \bHi(t_{(i)}^k) = \bhi(t_{(i)}^k), \bL  = \bl\big)$ is positive and strictly increasing in $t$ and 
  \begin{equation*}
    \Pr\big(\X = \x, \bHi(t_{(i)}^k) = \bhi(t_{(i)}^k) \,\big|\, \bL = \bl\big) 
    \in (0, 1)
  \end{equation*}
  for $k = 0, \ldots, n - 1$ and all $\x$, $\bhi$, and $\bl$.
 
  \label{as:positivity}
\end{assumption}
Assumption~\ref{as:positivity} states that it is possible for the latent potential infection times to take any positive value and that any treatment allocation is possible.

A final assumption, which is only used to identify certain exposure-marginalized estimands, is given below. 

\begin{assumption}[Cross-world independence] 
  For $\xminusi \neq \xminusi'$ for $k=0,\ldots,n-1$,
  \begin{equation*}
    I_i^k(x_i, \xminusi, \bhi(t_{(i)}^k)) 
    \indep \bHi^{*}(t_{(i)}^k; \xminusi') \mid \bL
  \end{equation*}

  \label{as:crossworld}
\end{assumption}
Assumption \ref{as:crossworld} states that, given $\bL = \bl$, the potential latent time to infection $I_i^k(x_i, \xminusi, \bhi(t_{(i)}^k))$ is conditionally independent of the the infection history among individuals other than $i$ under treatments different from $\xminusi$.


\subsection{Identification of average potential infection outcomes}

We first seek to identify the distribution of the infection time $I_i^k(\x,\bhi)$ of $i$ following the $k$th infection among other cluster members. Let 
\begin{equation*}
\begin{split}
  F_{I_i^k}(s \,|\, \x,\bhi,\bl) 
  & = \Pr\big(I_i^k(x_i, \xminusi, \bhi) < s \,\big|\, 
  \bL = \bl\big) \\
  & = \Pr\big(I_i^k(x_i, \xminusi, \bhi(t_{(i)}^k)) < s \,\big|\, \bL = \bl\big)
  \end{split}
\end{equation*}
be the cumulative distribution function of $I_i^k(\x,\bhi)$ given $\bL=\bl$ for $k=1, \ldots, n-1$, and let 
\begin{equation*}
\begin{split}
  F_{I_i^0}(s|\x,\bhi,\bl) & = \Pr\big(I_i^0(\x,\bhi )<s \,\big|\, \bL_i = \bl_i) \\
  & = \Pr\big(I_i^0(x_i)<s \,\big|\, \bL_i = \bl_i)
\end{split}
\end{equation*}
be the cumulative distribution of $I_i^0(\x,\bhi)$ under $\bL=\bl$. 
The following result shows that these distributions are identified via observation of $\big(\X, \bL, \bH(t)\big)$.

\begin{lem}[Identification of latent time to infection under competing risks] 
  Suppose Assumptions \ref{as:exclusion}-\ref{as:positivity} hold and $\bhi$ is a deterministic infection history with ordered infection times $t_{(i)}^1,\ldots,t_{(i)}^{n - 1}$.  
  Then the distribution of $I_i^k(\x,\bhi)$ given $\bL=\bl$ is identified by 
  \begin{equation}
    F_{I_i^k}(s|\x,\bhi,\bl) = 1-\exp\left[ -\int_{t_{(i)}^k}^{t_{(i)}^k+s} \frac{f_i^k(u|\x,\bhi,\bl)}{S_i^k(u|\x,\bhi,\bl)} du \right]  \text{ for $k=0,\ldots,n-1$}\\
    \label{eq:lem2}
  \end{equation}
  where
  \begin{align*}
    S_i^0(u \,|\, \x,\bhi,\bl) 
    &= \Pr\big(T_i > u,T_{\varphi_i^{1}} > u, \ldots, T_{\varphi_i^{n-1}} > u 
      \,\big|\, \X = \x,\bL = \bl\big) \\
    S_i^k(u \,|\, \x, \bhi, \bl\big) 
    &= \Pr\big(T_i > u, T_{\varphi_i^{k+1}} > u, \ldots, T_{\varphi_i^{n-1}} > u \,\big|\, T_{\varphi_i^1}=t_{(i)}^1, \ldots, T_{\varphi_i^{k}}=t_{(i)}^{k},\X=\x,\bL=\bl\big)
  \end{align*}
  and
  \begin{align*}
    f_i^0(u \,|\, \x, \bhi, \bl) 
    &= p\big(T_i=u,T_{\varphi_i^{1}}>u,\ldots,T_{\varphi_i^{n-1}}>u
      \,\big|\, \X=\x,\bL=\bl\big) \\
    f_i^k(u \,|\, \x, \bhi,\bl) 
    &= p\big(T_i = u, T_{\varphi_i^{k+1}} > u, \ldots, T_{\varphi_i^{n-1}} > u
      \,\big|\, T_{\varphi_i^1} = t_{(i)}^1, \ldots,T_{\varphi_i^{k}} = t_{(i)}^{k}, \X=\x, \bL=\bl) 
  \end{align*}
  for any fixed values of $\x$, $\bhi$ and $\bl$.
  \label{Iidentification}
\end{lem}
For $k = n - 1$, \eqref{eq:lem2} simplifies to
\begin{equation*}
  F_{I_i^{n-1}}(s \,\big|\, \x,\bhi,\bl) 
  = \Pr\big(T_i \le t_{(i)}^{n-1} + s \,\big|\, T_{\varphi_i^1}=t_{(i)}^1, \ldots, T_{\varphi_i^{n-1}} = t_{(i)}^{n-1}, T_i > t_{(i)}^{n-1}, \X = \x,\bL = \bl\big). 
\end{equation*}
Lemma \ref{Iidentification} can be understood as a standard competing risks identification result: after the $i^\text{th}$ infection at $t_{(i)}^k$, the infection of $i$ and the infections of the other still-uninfected individuals serve as competing events; in order to observe $I_i^k$ and estimate $F_{I_i^k}$, $i$ must be infected first among the individuals in the cluster who remain uninfected. 

In the following main result, we show that the average potential outcome $\E[Y_i(t;\bhi,\x)]$ under a particular deterministic exposure history $\bhi$ and treatment allocation $\x$ can be identified nonparametrically. 

\begin{thm}[Identification of exposure-and-treatment-controlled average potential infection outcomes]
  Suppose Assumptions \ref{as:exclusion}--\ref{as:positivity} hold. 
  Then 
  \[
    \E\big[Y_i(t; \bhi, \x) \,\big|\, \bL = \bl\big] = \\
    \sum_{j=0}^{n-1} \bigg[F_{I_i^j}(\min\{t,t_{(i)}^{j+1}\} - t_{(i)}^{j} \,|\, \x, \bhi, \bl\big) 
    \prod_{k=0}^{j-1} \big(1 - F_{I_i^k}(t_{(i)}^{k+1} - t_{(i)}^{k} \,|\, \x,\bhi, \bl)\big) \bigg]\\
  \]
  where $F_{I^j_i}(s\,|\, \x, \bhi,\bl)$ is given by Lemma \ref{Iidentification}, $F_{I^j_i}(s \,|\, \x, \bhi, \bl) = 0$ when $s \le 0$, and the product term is one when $j=0$.
  \label{mainidentification}
\end{thm}
A simplified version of Theorem \ref{mainidentification}, derived under a two-person partnership model, is given in Theorem 1 of \citet{cai2021identification}, which we restate here.  
\begin{cor}[Identification of exposure-and-treatment-controlled average potential infection outcomes in a two-person partnership]
Let $n=2$ and label the individuals $i$ and $j$. Let $\bh_{(i)} = \{ t_{(i)}^1=t_j \}$ where $t_j$ is the infection time of $j$ in Theorem \ref{mainidentification}, and suppose Assumptions \ref{as:exclusion}--\ref{as:positivity} hold.  Then for individual $i$, if $t > t_j$, we have
\[ \E[Y_i(t;\bh_{(i)},\x)|\bL=\bl] = F_{I_i^0}(t_j|x_i,\bl_i) + \big( 1- F_{I_i^0}(t_j |x_i,\bl_i) \big) \Pr [T_i < t|T_i \ge t_j, T_j=t_j,\X=\x,\bL=\bl]  \]
and if $t \le t_j$, we have $\E[Y_i(t;\bh_{(i)},\x)|\bL=\bl] =F_{I_i^0}(t|x_i,\bl_i)$. 
\label{cor:twoperson}
\end{cor}

\subsection{Identification of average exposure-marginalized potential infection outcomes}

\label{maringalizedoutcome}

Theorem \ref{mainidentification} identifies average potential infection outcomes under a particular deterministic infection history $\bhi$. 
But this history $\bhi$ may not be of specific interest; rather, we often wish to learn about potential infection outcomes as a function of treatments alone, by marginalizing over an appropriate distribution on infection histories.  
To this end, we show that the distribution of the potential infection history $\bHi^*(\xminusi)$ of individuals other than $i$ is identified.

Define the \emph{potential exposure distribution} $G_{(i)}^*(\bhi|\x_{(i)},\bl_{(i)})$ as the distribution of the potential infection history $\bHi^*(\xminusi)$ given $\bL_{(i)} = \bl_{(i)}$. 
For a particular deterministic infection history $\bh_{(i)}$, we denote $\bh_{(i,j)}$ as the elements of this history without the history of $j$. 
To identify the potential exposure distribution, we will need to represent the history of the group, absent $j$, where $i$ remains uninfected. 
Therefore, for a deterministic history $\bh_{(i)}$, let $\bh_{(j)}^i=(\bh_{(i,j)},h_i)$ where $h_i$ corresponds to $y_i(t)=0$ for all $t>0$, so that $i$ is uninfected.  As before, we order individuals by their infection times so that $t_{\varphi_i^k}=t_{(i)}^k$ for $k=1,\ldots,n-1$. 
Then individual $\varphi_i^1$ is infected first (via some exogenous source of infection) at time $t_{(i)}^1$, which means $I_{\varphi_i^1}^0=t_{(i)}^1$ and $I_{j}^0>t_{(i)}^1$ for $j=\varphi_i^2,\ldots,\varphi_i^{n-1}$. 
After the first infection, $\varphi_i^2$ wins the competition to be infected at time $t_{(i)}^2$, which means $I_{\varphi_i^2}^1 = t_{(i)}^2-t_{(i)}^1$ and $I_{\varphi_j}^1 > t_{(i)}^2 - t_{(i)}^1$ for $j  = \varphi_i^3, \ldots, \varphi_i^{n - 1}$. 
This continues until the last individual $\varphi_i^{n-1}$ gets infected at $t_{(i)}^{n-1}$ with the additional time to infection after the previous infection as  $I_{\varphi_i^{n-1}}^{n-2}=t_{(i)}^{n-1}-t_{(i)}^{n-2}$.  Let $f_{I_{\varphi_i^j}^{j-1}}(\cdot | \x,\bh^i_{(\varphi_i^j)},\bl)$ be the density of $I_{\varphi_i^j}^{j-1}$, which specifies the density of additional time to infection for the $j^\text{th}$ infected subject $\varphi_i^j$ after the previous $j-1$ infections, and let $S_{I_{\varphi_i^k}^{j-1}}(\cdot |\x,\bh^j_{(\varphi_i^k)},\bl)$ be the corresponding survival function of $I_{\varphi_i^j}^{j-1}$, identified by Lemma \ref{Iidentification}.  By Lemma \ref{Iidentification}, we can identify the distribution of $\bHi^*(\xminusi)$.

\begin{thm}[Identification of the potential exposure distribution] The density function $\mathrm{d}G_{(i)}^*(\bhi \,|\, \xminusi,\bl_{(i)})$ is identified by
\begin{equation*}
  \begin{split}
    \text{d}G_{(i)}^*(\bhi \,|\, \x_{(i)}, \bli) 
    &= \prod_{j=1}^{n-1} \bigg[f_{I_{\varphi_i^j}^{j-1}}\big( t_{(i)}^j-t_{(i)}^{j-1} \,\big|\, \x,\bh^i_{(\varphi_i^j)},\bl \big) \prod_{k=j+1}^{n-1} S_{I_{\varphi_i^k}^{j-1}}\big( t_{(i)}^j-t_{(i)}^{j-1} \,\big|\, \x,\bh^j_{(\varphi_i^k)},\bl\big)\bigg] \\
    \end{split}
  \end{equation*}
  where the ordered infection times $t_{(i)}^1,\ldots,t_{(i)}^{n-1}$ and the histories $\bh^i_{(\varphi_i^j)}$ are specified by the infection history $\bhi$. 
  \label{thm:exposure}
\end{thm}

Using the expected potential infection outcome identified by Theorem \ref{thm:exposure}, we can describe how to integrate, marginalize, or standardize these conditional estimands so that they do not depend on any particular value of exposure history $\bh_{(i)}$. 

\begin{prop}[Exposure-marginalized average potential infection outcomes]
  Under Assumptions \ref{as:consistency}--\ref{as:positivity}, the exposure-marginalized average potential infection outcome under $\X=\x$ is
  \begin{equation}
    \E\big[Y_i\big{(}t;x_i,\xminusi,\bHi^*(\xminusi)\big{)} \,\big|\, \bL=\bl\big] 
    = \int \E[Y_i(t;x_i,\x_{(i)},\bhi) \,|\, \bL=\bl]\,\mathrm{d}G_{(i)}^*(\bhi \,|\, \xminusi, \bl_{(i)}) 
    \label{eq:yitx} 
  \end{equation}
  where $\E[Y_i(t;x_i,\x_{(i)},\bhi) \,|\, \bL=\bl]$ is given by Theorem \ref{mainidentification} and $\mathrm{d}G_{(i)}^*(\bhi|\xminusi, \bl_{(i)})$ is given by Theorem \ref{thm:exposure}. 
  If Assumption \ref{as:crossworld} also holds and $\xminusi \neq \xminusi'$ then the exposure-marginalized average potential infection outcome under $\bHi^*$, hold to be realized under $\x_{(i)}' \neq \x_{(i)}$, and $\X=\x$ is 
  \begin{equation}
    \E\big[Y_i\big(t;x_i, \x_{(i)}, \bHi^*(\x_{(i)}')\big) \,\big|\, \bL=\bl] 
    = \int \E[Y_i(t; x_i, \x_{(i)}, \bhi) \,|\, \bL=\bl]\ \mathrm{d}G_{(i)}^*(\bhi \,|\, \x_{(i)}', \bl_{(i)}) 
    \label{eq:yitx2} 
  \end{equation}
  \label{prop:overexposure}
\end{prop}
This result follows directly from Assumption \ref{as:positivity} and Theorems \ref{mainidentification} and \ref{thm:exposure}.  Note that when $\x_{(i)}\neq \x_{(i)}'$, \eqref{eq:yitx2} is a ``cross-world'' potential outcome because the integrand is evaluated under the treatment allocation $\X=(x_i,\xminusi)$, while the marginalization is with respect to infection histories under a different treatment allocation $\X_{(i)}=\xminusi'$.  

Theorems \ref{mainidentification} and \ref{thm:exposure} come together in Proposition \ref{prop:overexposure}, which can be interpreted as a special case of Robins' G-formula \citep{robins1986new}: it decomposes the average potential outcome into an integral over a fully conditional expectation, with respect to a modified distribution over prior outcomes, treatments, and covariates that avoids post-treatment confounding.  By integrating with respect to the modified history $\bHi^*(\xminusi)$ in the marginalization, we circumvent the dependence problem between the infection outcome of $i$ and the infection history of others. 
In fact, a stronger result shows that the exposure-marginalized average potential infection outcome, given by the integral in \eqref{eq:yitx}, is equal to the observed conditional expectation.

\begin{cor}[Equivalence of observed and potential infection outcomes]
 By Assumptions \ref{as:exclusion}, \ref{as:prevexchangeability}, \ref{as:txexchangeability} and \ref{as:consistency}, 
  \[ \E[Y_{i}(t) \,|\, \X=\x,\bL=\bl]= \E[Y_i \big(t;x_i,\xminusi,\bHi^*(\xminusi) \big) \,|\, \bL=\bl] \]
  \label{cor:obs}
\end{cor}
Proposition \ref{prop:overexposure} and Corollary \ref{cor:obs} provide two ways of identifying potential infection outcomes.  The key insight is that we can decompose the observational quantity $\E[Y_{i}(t) \,|\, \X=\x,\bL=\bl]$ into a integral of potential infection outcomes with respect to an infection history $\bHi^*(\xminusi)$ that omits the role of $i$ in infecting others.  The explicit decomposition in Proposition \ref{prop:overexposure} is especially useful in statistical settings when data are sparse: by estimating the integrand and marginalizing distribution, we can evaluate average potential infection outcomes under joint treatments or covariates that were not observed in the data.


\subsection{Treatment- and covariate-marginalized estimands}

Proposition \ref{prop:overexposure} identifies potential infection outcomes that are a function of the joint treatment vector $\X=\x$, conditional on covariates $\bL=\bl$. To represent potential outcomes as a function only of treatment assigned to the cluster, or of treatment assigned to a single cluster member, it is necessary to marginalize (integrate) with respect to $\X_{(i)}$ and/or $\bL$.  

To define exposure-and-treatment-marginalized average potential outcomes, we integrate the expected exposure-marginalized potential infection outcomes over a distribution for others' treatments conditional on covariates $\bL=\bl$. An additional assumption ensures that that the cluster-level covariates $\bL$ suffice to render $X_i$ and $\X_{(i)}$ conditionally independent.  

\begin{assumption}[Conditional independence of treatments]
  For $i = 1, \ldots, n$, $X_i \indep \X_{(i)} \mid \bL$. 
  \label{as:sufficientL}
\end{assumption}

In other words, Assumption \ref{as:sufficientL} makes $\Pr(\X_{(i)}=\x_{(i)}|\bL=\bl)$ invariant to the value of $X_i=x_i$.  Let $p(\x_{(i)}\,|\, \bL = \bl)$ be a distribution over others' treatments, conditional on cluster-level covariates.  When this is the distribution of others' treatments imposed by the data-generating process, then $p(\x_{(i)}\,|\, \bL = \bl) = \Pr(\X_{(i)}=\x_{(i)}|\bL=\bl)$.  Importantly, $p(\x_{(i)}\,|\, \bL = \bl)$ must be invariant to the treatment $X_i=x_i$ of $i$ \citep{vanderweele2011effect,crawford2019interpretation,savje2021average}.  If in addition Assumption \ref{as:sufficientL} holds, then we can identify the potential infection outcome of $i$ under treatment $x_i$, integrated with respect to the marginal distribution of treatment to others, given covariates.  Let $\mathcal{X}^{n-1}=\{0,1\}^{n-1}$ be the set of all binary vectors of length $n-1$. 

\begin{prop}[Exposure-and-treatment-marginalized average potential infection outcomes]
Under Assumption \ref{as:sufficientL}, 
  \begin{align*}
    \overline{Y}_i^p(t; x_i \,|\, \bL = \bl) 
    &= \sum_{x_{(i)}\in\mathcal{X}^{n-1}} \E\big[Y_i\big(t;x_i,\x_{(i)},\bHi^*(\x_{(i)})\big)\mid \bL=\bl\big] p(\x_{(i)}\,|\, \bL = \bl) \\
    \overline{Y}_i^p(t; x_i, \x_{(i)}' \,|\, \bL = \bl) 
    &= \sum_{x_{(i)}\in\mathcal{X}^{n-1}} \E\big[Y_i\big(t;x_i, \x_{(i)}, \bHi^*(\x_{(i)}')\big)\mid \bL=\bl\big] p(\x_{(i)}\,|\, \bL = \bl)
  \end{align*}
  \label{prop:overexposuretx}
\end{prop}

In Proposition \ref{prop:overexposuretx} we have specified an arbitrary distribution $p(\x_{(i)}\,|\, \bL = \bl)$ which may differ from the distribution by which treatment was assigned.  We define covariate-marginalized potential outcomes over the marginal distribution of covariates $\bL$ in the cluster, denoted by $Q(\bl)$.  

\begin{defn}[Covariate adjustment]
  Under Assumptions~\ref{as:exclusion}--\ref{as:positivity},
  \begin{align*}
    \E[Y_i(t;x_i,\x_{(i)},\bhi)] 
    &= \int \E[Y_i(t; x_i, \x_{(i)}, \bhi) \,|\, \bL = \bl] \,\mathrm{d}Q(\bl) \\
    \E\big[Y_i(t; x_i, \x_{(i)}, \bHi^*(\x_{(i)})\big)\big] 
    &= \int \E\big[Y_i(t; x_i, \x_{(i)}, \bHi^*(\x_{(i)})\big) \mid \bL = \bl\big] \,\mathrm{d}Q(\bl) \\
    \overline{Y}_i^p(t;x_i) 
    &= \int \overline{Y}_i^p(t;x_i \,|\, \bL=\bl) \,\mathrm{d}Q(\bl)
  \end{align*}
  Under Assumptions~\ref{as:exclusion}--\ref{as:crossworld} when $\x_{(i)} \neq \x_{(i)}'$, 
  \begin{align*}
    \E\big[Y_i(t; x_i,\x_{(i)},\bHi^*(\x_{(i)}')\big)\big] 
    &= \int \E\big[Y_i(t;x_i,\x_{(i)},\bHi^*(\x_{(i)}')\big) \,\big|\, \bL = \bl\big] \,\mathrm{d}Q(\bl) \\
    \overline{Y}_i^p(t;x_i,\x_{(i)}') 
    &= \int \overline{Y}_i^p(t; x_i, \x_{(i)}' \,|\, \bL=\bl) \,\mathrm{d}Q(\bl)
  \end{align*}
  where $\E[Y_i(t; x_i, \x_{(i)}, \bhi)) \,|\, \bL = \bl]$ is given by Theorem~\ref{mainidentification}, $\E\big[Y_i\big(t, x_i, \x_{(i)}, \bHi^*(\x_{(i)})\big)  \,\big|\, \bL = \bl\big]$ and $\E\big[Y_i\big(t,x_i,\x_{(i)},\bHi(\x_{(i)}')\big) \,\big|\, \bL = \bl\big]$ are given by Proposition~\ref{prop:overexposure}, and the definitions of $\overline{Y}_i^p(t, x_i \,|\, \bL = \bl)$ and $\overline{Y}_i^p(t, x_i, \x_{(i)}' \,|\, \bL=\bl)$ are given by Proposition \ref{prop:overexposuretx}.
\label{prop:overcovariate}
\end{defn}



\section{Causal effects based on risk differences}

With these identification results in hand, we present biologically meaningful causal estimands summarizing the transmissibility of disease, as well as the effects of treatment on infection and transmision.  
We define causal effects in terms of infection risk differences, but similar effects may be defined in terms of risk ratios or odds ratios.
We focus on the potential infection outcomes of a particular individual $i$. 

First, we define what it means for the infection outcome to be contagious, or transmissible. 

\begin{defn}[Contagion effect] 
  The exposure-controlled contagion effect is
  \begin{equation*}
    \CE_i(t,\x,\bhi,\bhi') = \E[Y_i(t;\x, \bhi) - Y_i(t;\x, \bhi')].
  \end{equation*}
  The exposure-marginalized contagion effect is
  \begin{equation*}
    \CE_i(t,x_i,\xminusi,\x'_{(i)}) = \E\big[Y_i\big(t;x_i,\x_{(i)},\bHi^*(\x_{(i)})\big) - Y_i\big(t;x_i,\x_{(i)},\bHi^*(\x_{(i)}')\big)\big].
  \end{equation*}
  \label{defn:contagion}
\end{defn}
Informally, the infection outcome is ``contagious'' if it depends on the infection history of others, holding treatments constant.  To make this idea more concrete, for $\bhi=(t_i^{(1)},\ldots,t_i^{(n-1)})$ and $\bhi'=(t_i^{'(1)},\ldots,t_i^{'(n-1)})$, we denote $\bhi \prec \bhi'$ if $t_i^{(k)} \le t_i^{'(k)}$ element-wise for $k=1,\ldots, n-1$.  We say the infection outcome is \emph{positively contagious} absent treatment if earlier infection of cluster members increases the average potential infection risk, so $CE(t,\mathbf{0},\bhi(t),\bhi'(t)) > 0$ whenever $\bhi \prec \bhi'$, when no cluster members receive treatment~\citep{cai2021identification}. 

Susceptibility effects compare potential infection outcomes of $i$ under treatment versus no treatment of $i$, while holding infection history and treatment to others constant. 

\begin{defn}[Susceptibility effect] 
  The exposure-controlled susceptibility effect is
  \begin{equation*}
    \SE_i(t,\x_{(i)}, \bhi) = \E[Y_i(t;1,\x_{(i)},\bhi) - Y_i(t;0,\x_{(i)},\bhi)], 
  \end{equation*} 
  the exposure-marginalized susceptibility effect is 
  \begin{equation*}
    \SE_i(t,\x_{(i)}) = \E\big[Y_i\big(t;1,\x_{(i)},\bHi^*(\x_{(i)})\big) - Y_i\big(t;0,\x_{(i)},\bHi^*(\x_{(i)})\big)\big], 
  \end{equation*} 
  and the exposure-and-treatment-marginalized susceptibility effect is
  \begin{equation*}
    \SE_i^p(t) = \overline{Y}_i^p(t;1) - \overline{Y}_i^p(t;0). 
  \end{equation*}
  \label{defn:susceptibility}
\end{defn}

This susceptibility effect summarizes the effect of changing only the treatment of $i$, holding all else constant; if it is negative, then the treatment reduces the chance of infection.  

Infectiousness effects compare the potential infection outcome of $i$ under different treatments to others while holding constant treatment and exposure to infection of $i$.
\begin{defn}[Infectiousness effect] 
  The exposure-controlled infectiousness effect is 
  \begin{equation*}
    \IE_{i}(t,x_i,\x_{(i)},\x_{(i)}',\bhi) = \E[Y_i(t;x_i,\x_{(i)},\bhi) - Y_i(t;x_i,\x_{(i)}', \bhi)], 
  \end{equation*} 
  and the exposure-marginalized infectiousness effect is 
  \begin{equation*}
    \IE_i(t,x_i,\x_{(i)},\x_{(i)}') = \E\big[Y_i\big(t;x_i,\x_{(i)}, \bHi^*(\xminusi)\big) - Y_i\big(t;x_i,\x_{(i)}', \bHi^*(\xminusi)\big)\big] 
  \end{equation*}
  where we have held the marginalizing distribution $\bHi^*(\xminusi)$ constant under treatment $\xminusi$. 
  The treatment-and-exposure-marginalized infectiousness effect is:
  \begin{equation*}
    \IE^p(t) =\sum_{\x \in \mathcal{X}^{n}}  \IE_i(t,x_i,\xminusi) p(\x). 
  \end{equation*}
  \label{defn:infectiousness}
\end{defn}
Informally, the infectiousness effect summarizes the effect of treatment on an infected individual's ability to infect others.  Note that the distribution of $\bHi^*(\xminusi)$ is the same in both parts of the contrast, and we are only changing the integrand, not the marginalizing distribution.


\section{Statistical identification using pairwise hazard models}

Non-parametric identification of exposure-controlled potential infection outcomes via Lemma \ref{Iidentification} and Theorem \ref{mainidentification} can be difficult in finite samples because realized infection histories may be different from the history of interest, or we may not have enough realizations to adequately capture the natural distribution over infection histories. 
Fortunately, flexible statistical modeling assumptions can help parameterize distribution functions $F_{I_i^k}(s|\x,\bhi,\bl)$ so that these can be efficiently estimated, facilitating computation of expected potential infection outcomes. 
\citet{kenah2013non,kenah2015semiparametric} developed a semiparametric class of infection transmission models that provides the tools needed to dramatically simplify inference of identified estimands while accommodating individual-level treatment and covariates. 
The approach is based on extensive prior work modeling time-to-infection data \citep{halloran1994exposure,rhodes1996counting,kenah2008generation,kenah2011contact,kenah2013non,eck2019randomization}. 

\subsection{Hazards for time-to-transmission}

We seek to parameterize the distribution of the potential waiting time $I_i^k(\x,\bhi)$ from the $k$th infection to the infection of $i$.  
The key idea is that previously infected individuals, along with a possibly exogenous source of infection, impose competing risks of disease transmission to remaining uninfected individuals.  
We view $I_i^k(x_i,\xik,\bhik)$ as the minimum of these $k + 1$ competing transmission times, starting at the time of the $k$th  infection $t_{(i)}^k$. 

To formalize this intuition, let $\tau_{ji}$ be the time for $j$ to infect $i$, following infection of $j$, and suppose all $\tau_{ji}$ are conditionally independent, given relevant treatments and covariates.  Let $\tau_{ji}(x_j,x_i)$ be the potential transmission time as a function of treatments to $j$ and $i$, and let $\lambda_{ji}(\cdot|x_j,x_i,\bl_j,\bl_i)$ be the hazard function of $\tau_{ji}(x_j,x_i)$, where we interpret $\tau_{0i}(x_i)$ as the potential transmission time to $i$ from an outside (exogenous) source of transmission, and let $\lambda_{0i}(t|x_i,\bl_i)$ be the corresponding transmission hazard. 

\begin{assumption}[Independence of competing transmission times]
For all $i$ and all $j \neq i$, $k \neq l$
  \begin{equation}
  \begin{split}
      \tau_{ji}(x_j,x_i) &\indep \tau_{kl}(x_k,x_l) \mid \bL, \\
    \lambda_{0i}(t \,|\, \x, \bhi, \bl) 
      &= \lambda_{0i}(t \,|\, x_i,\bl_i),\\
    \lambda_{ji}(t \,|\, \x, \bhi, \bl) 
      &= \lambda_{ji}(t -t_j \,|\, x_i, x_j, \bl_i, \bl_j).	
  \end{split}	
  \end{equation}
\label{as:independentcompeting}
\end{assumption} 
Assumption~\ref{as:independentcompeting} states that the latent waiting times to infection from different sources to $i$ are conditionally independent given treatments and covariates and that their hazard functions are determined by treatments and covariates in the target individual and, if needed, the infected source individual.  

This assumption yields a convenient additive hazard model for the infection time of $i$. At time $t > t_{(i)}^k$, the infection waiting time $I_i^k$ has hazard 
\begin{equation}
  \lambda_{0i}(t \,|\, x_i,\bl_i) + \sum_{j=1}^k \lambda_{\varphi_i^j,i}(t - t_{(i)}^k \,|\, x_{\varphi_i^j},x_i,\bl_j,\bl_i)
\label{eq:sumhazard_forI}
\end{equation}
which is the sum of the transmission hazards from existing sources of infection.  Therefore the infection hazard of $i$ at an arbitrary time $t$, derived from \eqref{eq:sumhazard_forI}, can be written as
\begin{equation}
  \lambda_i(t \,|\, \x,\bhi, \bl)
  = \lambda_{0i}(t \,|\, x_i,\bl_i) 
    + \sum_{j \neq i} y_j(t) \lambda_{ji}(t-t_j \,|\, x_i,x_j,\bl_i,\bl_j) .
\label{eq:sumhazard}
\end{equation}
Assumption \ref{as:independentcompeting} dramatically simplifies representation and estimation of causal estimands in terms of time integrals of these hazard functions.  Figure~\ref{fig:changeofhazard} shows how $\lambda_i(t|\x,\bhi,\bl)$ changes under a deterministic history $\bhi=(t_i^{(1)},\ldots,t_i^{(n-1)})$.

\begin{figure}
  \centering
  \includegraphics[width=0.7\textwidth]{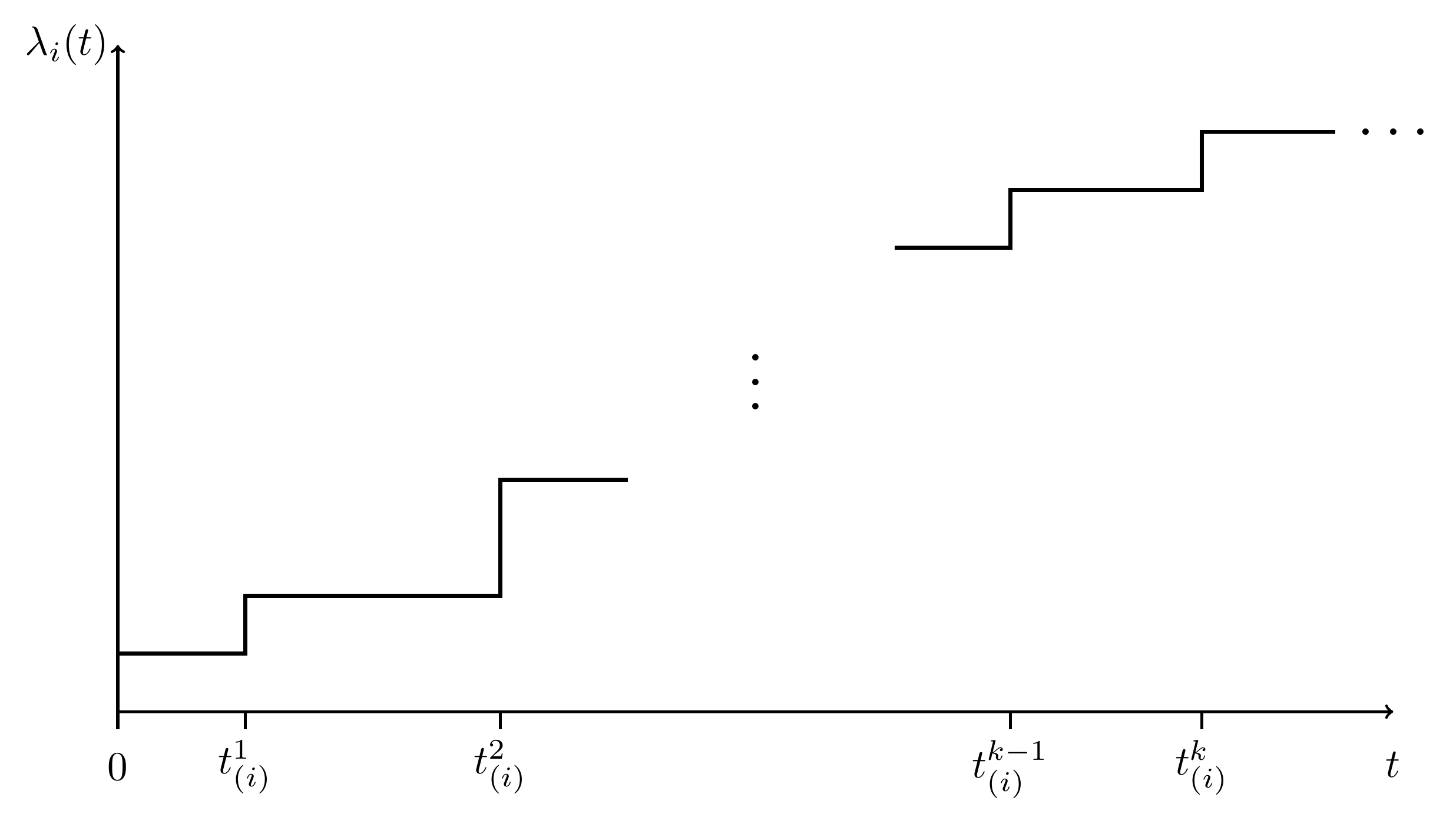}
  \caption{Illustration of the infection hazard experienced by $i$ under the hazard model \eqref{eq:sumhazard} and a deterministic infection history $\bhi=(t_i^{(1)},\ldots,t_i^{(n-1)})$. For simplicity, we assume $\lambda_{ji}(t|\x,\bhi,\bl)$ is constant for $t>0$. Before the first infection at $t_{(i)}^1$, $i$ experiences only the exogenous force of infection $\lambda_{0i}(t \,|\, x_i,\bl_i)$.  Following the first infection at $t_{(i)}^1$, $i$ experiences $\lambda_{0i}(t \,|\, x_i,\bl_i) + \lambda_{\varphi_i^1,i}(t-t_{(i)}^1|x_i,x_{\varphi_i^1},\bl_i,\bl_{\varphi_i^1})$. This process continues until all $n-1$ other cluster members are infected. }
\label{fig:changeofhazard}
\end{figure}

The hazard representation \eqref{eq:sumhazard} has an important consequence fo identification of expected potential infection outcomes.  

\begin{thm}[Identification of exposure-controlled potential outcomes under the pairwise hazard model]
  Under Assumption~\ref{as:independentcompeting}, Theorem~\ref{mainidentification} reduces to 
  \begin{multline*}
   \E[Y_i(t,\x,\bh_{(i)}) \,|\, \bL=\bl] 
= 1 - \exp \left[- \int_0^t \left( \lambda_{0i}(s|x_i,\bl_i) + \sum_{j \neq i} y_j(s) \lambda_{ji}(s-t_j|x_i,x_j,\bl_i,\bl_j) \right) \mathrm{d}s \right]
    \label{thm:Tidistribution}
  \end{multline*}
  where $t_j$ is the infection time of $j$ in the history $\bh_{(i)}$ and $y_j(s)$ is the indicator that $s>t_{(i)}^j$.
\label{thm:haz_mainresult}
\end{thm}
This result is a consequence of the hazard representation \eqref{eq:sumhazard}, and makes clear that the expected exposure-controlled potential infection outcome is simply the probability of infection (failure) prior to time $t$ under the specified treatments and infection history of others.  

Finally, exposure-marginalized potential infection outcomes can be computed using the identified distribution of the potential exposure distribution.

\begin{thm}[Identification of exposure-marginalized potential outcomes under the pairwise hazard model]
  Under Assumption~\ref{as:independentcompeting}, Theorem~\ref{thm:exposure} reduces to
  \begin{equation}
    \text{d}G_{(i)}^*(\bhi \,|\, \x_{(i)}, \bli) 
    = \prod_{j=1}^{n-1} \bigg[\lambda_{\varphi_i^j}(t_{(i)}^j \,|\, \x, \bh^i_{(\varphi_i^j)},  \bl) \prod_{k = j}^{n-1} \exp\bigg(-\int_{t_{(i)}^{j-1}}^{t_{(i)}^j} \lambda_{\varphi_i^k}(s \,|\, \x, \bh^i_{(\varphi_i^k)}, \bl) \,\mathrm{d}s\bigg)\bigg] \\
    \label{eq:thm4}
  \end{equation}
For the terms inside the square bracket, we have used the fact that the probability density for the infection time of $j$ equals the product of its hazard and survival functions. For the product term inside the square brackets, each individual $k$ is included for $j = 1, \ldots, k-1$.
  Thus, \eqref{eq:thm4} simplifies to
  \begin{equation}
    \text{d}G_{(i)}^*(\bhi \,|\, \x_{(i)}, \bli) 
    = \prod_{j=1}^{n-1} \lambda_{\varphi_i^j}(t_{(i)}^j \,|\, \x, \bh^i_{(\varphi_i^j)},  \bl) \times \prod_{k = 1}^{n-1} \exp\bigg(-\int_0^{t_{(i)}^k} \lambda_{\varphi_i^k}(s \,|\, \x, \bh^i_{(\varphi_i^k)}, \bl) \,\mathrm{d}s\bigg)
  \end{equation}
\label{thm:haz_Hresult}
\end{thm}

\subsection{Statistical infection hazard model}

Following~\citet{kenah2015semiparametric}, we consider a Cox-type hazard model for the pairwise infection times: 
\begin{equation}
\begin{split}
  \lambda_{0i}(t \,|\, x_i, \bl_i) 
  &= \alpha(t) \exp[\beta_1 x_i + \theta_1 \bl_i] \\
  \lambda_{ji}(t \,|\, x_i, x_j, \bl_i) 
  &= \gamma(t-t_j) \exp[\beta_1 x_i + \beta_2 x_j + \theta_1 \bl_i + \theta_2 \bl_j]	
\end{split}
\label{eq:haz}
\end{equation}
which has non-parametric baseline hazards $\alpha(t)$ for the (time-varying) risk of infection from outside the cluster, and $\gamma(t - t_j)$ for the (time-varying) infection risks from already-infected cluster members $j$, with the effects of treatments $\X=x$ and covariates $\bL$ accommodated in a log-linear model with coefficients $\beta_1$, $\beta_2$, $\theta_1$ and $\theta_2$.   This infection hazard model has been used by many authors to represent transmission dynamics in groups \citep{rhodes1996counting,cauchemez2004bayesian,kenah2013non,kenah2015semiparametric,morozova2018risk,eck2019randomization,crawford2020transmission}. 

The coefficients of the infection hazard model \eqref{eq:haz} have inutitive interpretations. If $\beta_1 < 0$, vaccination decreases the hazard of infection for treated individuals against all sources of infection by a factor of $e^{\beta_1}$.  If $\beta_2 < 0$, vaccination of other individuals decreases their hazard of transmission to $i$ by a factor of $e^{\beta2}$ once they become infectious.  Plugging in \eqref{eq:haz}, the hazard model~\eqref{eq:sumhazard} takes a simple form, 
\begin{equation}
  \lambda_i(t \,|\, \x,\bhi, \bl)
  =\exp[\beta_1 x_i + \theta_1 \bl_i]  \Big( \alpha(t) 
    + \sum_{j \neq i} y_j(t) \gamma(t-t_j) \exp[\beta_2 x_j + \theta_2 \bl_j] \Big) .
\end{equation}


\subsection{Hazard ratio estimands}

\label{sec:hazardestimands}

The parameters of the hazard model \eqref{eq:sumhazard} are closely related to causal effects defined above.  For example, the contagion effect is positive whenever $\gamma(t) > 0$, and the sign of the susceptibility effects and infectiousness effects depend only on $\beta_1$ and $\beta_2$, respectively (details are given in the Appendix). These facts suggest that the parameters of this model may have a causal interpretation in terms of their relationship to the causal contagion, susceptibility, and infectiousness effects defined above.  We construct hazard ratio estimands by contrasting potential infection hazard under different exposure history and treatment. 

\begin{defn}[Controlled susceptibility hazard ratio]
The controlled susceptibility hazard ratio is defined as:
\begin{equation*}
  HSE^C(t, \x_{(i)}, \bhi, \bl) = \frac{\lambda_{i}(t \,|\, x_i = 1, \x_{(i)}, \bhi, \bl)}{\lambda_{i}(t \,|\, x_i = 0, \x_{(i)}, \bhi, \bl)}
  \label{eq:de}
\end{equation*}
\end{defn}
The controlled hazard susceptibility ratio is the ratio of counterfactual hazards of $i$ when $i$ is treated versus untreated, while holding the treatments, infection histories, and covariates of others constant. Plugging \eqref{eq:haz} into the controlled hazard susceptibility effect, we find that
\begin{equation}
  HSE^C(t,\x_{(i)},\bhi,\bl)  =e^{\beta_1}
  \label{eq:controlledSE}
\end{equation}
regardless of the values of $t$, $\x_{(j)}$, $\bhi$, and cluster size.
If $\beta_1 < 0$, vaccination is beneficial because it decreases the hazard of infection from all sources of infection by a factor of $e^{\beta_1}<1$.

\begin{defn}[Controlled infectiousness hazard ratio]
The controlled infectiousness hazard ratio is defined as:
  \begin{equation*}
    HIE^C(t,h_j,h'_j,\x_{(j)}, \bh_{(i,j)},\bl) = \frac{\lambda_{i}(t|x_{j}=1, \x_{(j)},h'_j, \bh_{(i,j)},\bl) - \lambda_{i}(t|x_{j}=1, \x_{(j)},h_j, \bh_{(i,j)},\bl) }{\lambda_{i}(t|x_{j}=0, \x_{(j)},h'_j, \bh_{(i,j)},\bl) - \lambda_{i}(t|x_{j}=0, \x_{(j)},h_j, \bh_{(i,j)},\bl)}
  \end{equation*}
  where $y_j(t)=1$ in $h'_j$ and $y_j(t)=0$ in $h_j$. 
\end{defn}

The controlled infectiousness hazard ratio is the ratio of the hazard of $i$ being infected by a given infectious neighbor $j$, when $j$ is vaccinated versus unvaccinated, holding $\x_{(j)}$, exposure to infection from sources other than $j$, and all covariates constant. 
Applying the hazard model~\eqref{eq:haz} to the controlled infectiousness effect, we have
\begin{equation}
  HIE^C(t, h_j, h'_j, \x_{(j)}, \bh_{(i,j)},\bl) = e^{\beta_2}
  \label{eq:ie2}
\end{equation}
regardless of the values of $t$, $\x_{(j)}$, $\bh_{(i,j)}$, and cluster size. 
If $\beta_2 < 0$, the infectiousness effect is beneficial because it decreases the instantaneous infectiousness of $j$ by a factor of $e^{\beta_2}<1$.

To determine whether the disease is contagious, we focus on one transmission contact between $i$ and an infectious individual $j$, and evaluate if coming into contact with an infectious individual increases the infection risk. 
Since treatments are irrelevant to the intrinsic transmissibility of a disease, we define the contagion effect under null treatments $\X=(0,\ldots,0)=\mathbf{0}$.

\begin{defn}[Controlled contagion cumulative hazard ratio]
  The controlled contagion cumulative hazard ratio is
  \begin{equation*}
    HCE^C(t,h^{''}_j,h'_j,\bh_{(i,j)},\bl) = \frac{\int_0^t \big{[} \lambda_i(u; \mathbf{x}=\mathbf{0},h'_j(t),\bh_{(i,j)},\bl)-\lambda_i(u; \mathbf{x}=\mathbf{0},h_j(t),\bh_{(i,j)},\bl) \big{]} du}{\int_0^t \big{[}\lambda_i(u; \mathbf{x}=\mathbf{0},h^{''}_j(t),\bh_{(i,j)},\bl)-\lambda_i(u; \mathbf{x}=\mathbf{0},h_j(t),\bh_{(i,j)},\bl) \big{]} du}
    \label{de:controlledcontation}
  \end{equation*}
  where the corresponding infection times $t'_j$ and $t''_j$ contained in $h'_j$ and $h''_j$ has the relationship of $t'_j < t^{''}_j < t$, and $y_j(t) = 0$ in $h_j(t)$.
\end{defn}
Applying the hazard model \eqref{eq:haz} to the controlled contagion effect, we have:
\begin{equation}
HCE^C(t,h^{''}_j,h'_j,\mathbf{h}_{(i,j)})=\frac{\int_{t_j'}^t \gamma(u) du}{\int_{t_j''}^t \gamma(u) du}
\label{eq:contagion2}
\end{equation}
Therefore, $\gamma(t)>0$ for any $t>0$ implies positive contagion, since earlier exposure to infection increases the infection risk.


\subsection{Statistical estimation}

The hazard model \eqref{eq:haz} can be estimated parametrically \cite{kenah2011contact} or semiparametrically \cite{kenah2013non,kenah2015semiparametric}.  We illustrate identification and estimation in the parametric case when the functional forms for $\alpha(t)$ and $\gamma(t)$ are known to obtain maximum likelihood estimates (MLEs) for model parameters. We then compute causal estimands.  Specifying particular functional forms for $\alpha(t)$ and $\gamma(t)$ is a relatively strong assumption, and if we do not want to specify the functional form of the baseline hazards, we can also estimate them non-parametrically by martingale techniques in counting process. Nevertheless, regardless of known or unknown baseline hazards, the meanings of $e^{\beta_1}$ and $e^{\beta_2}$ as causal estimands for the susceptibility and infectiousness effects are preserved, as long as the hazard model \eqref{eq:haz} holds.

With independent (or administrative) censoring at $C_i$ for individual $i$, we denote $\Delta_i=\indicator{T_i <C_i}$, with $\Delta_i=1$ indicating $i$'s infection time being observed, and $\Delta_i=0$ indicating $i$'s  infection time being censored. Denote $T^*_i=\min(T_i,C_i)$. For each cluster, we observe $(T^*_i,\Delta_i,X_i,\mathbf{L}_i)$ for all $i$. Given the hazard model for $\lambda_{i}$ specified in \eqref{eq:haz}, we can derive the full likelihood of infection times of all clusters as follows: 
\begin{equation}
\begin{split}
L(\beta_1,\beta_2,\theta_1,\theta_2) &= \prod_{i=1}^n \left[ \lambda_{i}(t^*_{i}) \right]^{\Delta_i} \exp\left[ -\int_0^{t^*_i} \lambda_{i}(u) \dx{u} \right]  \\
& = \prod_{i=1}^n \left[ e^{\beta_1 x_i + \theta_1 \bl_i} \Big{(} \alpha(t^*_i) + \sum_{j \neq i} y_j(t^*_i) \gamma(t^*_i-t_j) e^{\beta_2 x_j + \theta_2 \bl_j} \Big{)} \right]^{\Delta_i}
\\
& \quad \times \exp \left( -\int_0^{t^*_i} e^{\beta_1 x_i +\theta_1 \bl_i} \Big{(} \alpha(u) + \sum_{j \neq i} y_j(u) \gamma(u-t_j) e^{\beta_2 x_j + \theta_2 \bl_j} \Big{)} \dx{u} \right)
\end{split}
\label{eq:lik}
\end{equation}


\section{Simulations and comparisons to other estimands}

Infectious disease treatment effects have been evaluated using both randomized trials and observational studies. For infectious diseases that are transmissible within pairs or partnerships (e.g., HIV), researchers have designed partnership randomization studies based on sexual partners \cite{kamenga1991evidence,de1994longitudinal,el2005long,gilbert2010couple} or mother-child pairs \cite{mcsherry1999effects,brogly2010birth}. In outbreak or endemic settings, researchers have proposed trial designs for clusters when individuals are randomized to receive treatment and placebo \cite{halloran1997study,halloran1999design,hudgens2008toward,halloran2010design,perez2014assessing,halloran2016dependent}, and defined a series of direct, indirect, susceptibility and infectiousness vaccine effects by different summary measurements \cite{halloran1991direct,halloran1997study,halloran1999design}. Generalized estimation equation (GEE) regression approaches are popular for confounder adjustment when outcomes are correlated. These methods have been applied to assess treatment effects in cluster randomized trials \citep{ali2005herd} and network-randomized trials \citep{perez2014assessing,buchanan2018assessing}, and model the spread of contagious outcomes across network ties \cite{christakis2007spread,christakis2008collective,cacioppo2009alone}.

\subsection{Simulation of seasonal influenza vaccine trials}

We illustrate the performance of the controlled and exposure-marginalized causal estimands under an extensive collection of simulations that mimic a hypothetical influenza trial. We assume all individuals are uninfected at baseline, and each individual is susceptible to infection from a source outside the cluster and from any of their infected fellow cluster members. Infections can be transmitted between individuals in either direction and treatments are randomized under multiple strategies or correlated as in observational studies. We assume individuals remain infectious after infection, with their transmissibility governed by the hazard $\gamma(t)>0$.  

We consider an influenza vaccine that is protective in its susceptibility effect ($\beta_1<0$) and infectiousness effect ($\beta_2<0$), in clusters of 2, 4 and 8 individuals. We consider simulations under varied transmission dynamics, including: (1) constant external and internal baseline hazards of $\alpha(t)=\alpha$ and $\gamma(t)=\gamma$, which reduces to a Markov susceptible-infective process with an external force of infection \citep[e.g][]{morozova2018risk,eck2019randomization}, (2) a constant external force of infection $\alpha(t)=\alpha$ and decreasing internal baseline hazards of $\gamma(t) = b \exp[-\omega(t-t_j)]$, and (3) a seasonally varying external baseline hazard of $\alpha(t)=a(1+\sin(2 \pi t+\phi))$ and decreasing internal baseline hazards of $\gamma(t)=b \exp[-\omega(t-t_j)]$.

In randomized trial simulations, vaccine allocation is randomized in accordance with a specified distribution -- Bernoulli, block, or cluster randomization. For Bernoulli randomization, each individual within the cluster is independently randomized with $\Pr(\X_i=1)=\frac{1}{2}$; for block randomization, exactly half of the individuals of the cluster are randomized to the treatment and the remaining half are randomized to the placebo with $\sum_{i=1}^n X_i=\frac{n}{2}$; for cluster randomization, the whole cluster is either randomized to all receiving treatments or all receiving placebo with $\Pr(\X=\mathbf{0})=\Pr(\X=\mathbf{1}) = \frac{1}{2}$. In the observational study simulation, the traits $\bL=\bl$ together determine the joint distribution of vaccine in the cluster as 
\[ \Pr(X_i=1|L_i=l_i) = \frac{1}{1+e^{l_i}} \]
where 
\begin{equation}
\begin{pmatrix}L_1 \\ L_2\\ ...\\L_{n}\end{pmatrix} \sim \text{Normal}\left(\begin{pmatrix} 0 \\ 0 \\...\\0 \end{pmatrix}, v\begin{pmatrix} 1 & \rho & ... & \rho \\ \rho & 1 &...& \rho \\... \\\rho & \rho & ...& 1 \end{pmatrix}\right) 
\label{eq:rho}
\end{equation}
with $v>0$. 

Tables \ref{tab:estimation1}--\ref{tab:estimation4} show estimates of the hazard model coefficients $\beta_1$ and $\beta_2$, along with exposure-marginalized effects $\CE$, $\SE$, and $\IE$, under different cluster sizes, transmission dynamics, and vaccine allocation distributions. We set $e^{\beta_1}=0.5$ so that influenza vaccination reduces the hazard of infection for vaccinated individuals by $50\%$ \cite{cdcfluvaccine} in Tables \ref{tab:estimation1}--\ref{tab:estimation3}. There is little information about how much vaccines reduces the infectiousness of influenza, so we vary this effect in simulations: we simulate $e^{\beta_2}=0.5$ in Tables \ref{tab:estimation1}--\ref{tab:estimation2} and $e^{\beta_2}=0.9$ in Table \ref{tab:estimation3}, so that vaccine reduces the transmission by $50\%$ and $10\%$. The covariates are set to be negatively correlated with $\rho=-0.1$ in Tables \ref{tab:estimation1} and \ref{tab:estimation3} and positively correlated with $\rho=0.1$ in Tables \ref{tab:estimation2} and \ref{tab:estimation4}.
In Table \ref{tab:estimation4} we show a special scenario when the infectiousness effect of a vaccine is much stronger than the susceptibility effect with $e^{\beta_1}=0.9$ and $e^{\beta_2}=0.1$ so that vaccination reduces the rate of infection in a treated individual by $10\%$ but reduces transmission to others by $90\%$. The transmission dynamics for the baseline hazards are: 1) $\alpha(t)=0.3$ and $\gamma(t)=3$ for the constant hazard scenario, 2) $\alpha(t)=0.3$ and $\gamma(t)=6 e^{-0.5(t-t_j)}$ for the constant external hazard and decreasing internal hazard scenario, and 3) $\alpha(t)=0.2 \big{(} 1+\sin(2 \pi t+ \frac{\pi}{2}) \big{)}$ and $\gamma(t)=6 e^{-0.5(t-t_j)}$ for the seasonal external hazard and decreasing internal hazard scenario,
 so that the average infection rate in the population is $5\%$-$20\%$ \cite{cdcfluvaccine}. The observational duration for clusters of $2$, $4$ and $8$ is approximately $4$ months, $2$ months, and $1$ month respectively from the beginning of a flu season, assuming that a flu season typically lasts for 6 months.

From Tables \ref{tab:estimation1}--\ref{tab:estimation4}, we find that the signs of the estimated exposure-marginalized susceptibility effect $\widehat{\SE}$ and infectiousness effect $\widehat{\IE}$ match the signs of $\beta_1$ and $\beta_2$, and $\widehat{\CE}>0$ confirms that infection is contagious. The causal estimands $\CE$, $\SE$, and $\IE$ are difficult to identify non-parametrically by Theorem \ref{mainidentification}--\ref{thm:exposure} and Proposition \ref{prop:overexposure}. Additionally, these estimands are unidentifiable under block and cluster randomization due to lack of outcomes under certain treatments hence a violation of the positivity assumption. Although the positivity assumption is violated (not all treatment assignments are possible), we show how to estimate these estimands using the hazard model, and recover reasonable estimates for block and cluster randomizations.
We also see that $\widehat{\CE}$, $\widehat{\SE}$, and $\widehat{\IE}$ are summary measurements for certain transmission process at a particular time and change with the cluster size $n$, baseline hazard and randomization strategy, as shown in Figure \ref{fig:compareDEIDE}. In contrast, $\beta_1$ and $\beta_2$ are mechanistic features of the vaccine effect that, under the specified hazard model, do not depend on the randomization design, or other features of transmission or the observation time. Thus, their estimates are relatively stable across all scenarios. 

\begin{table}
\small
\centering
\begin{tabular}{ccccccccc}
\toprule
\multirow{2}{*}{Cluster}& \multirow{2}{*}{Treatment} & \multicolumn{2}{c}{Hazard estimands} & \multicolumn{5}{c}{Probability estimands} \\
\cmidrule(lr){3-4} \cmidrule(lr){5-9}
 & & $\hat{\beta}_1$ & $\hat{\beta}_2$ & $\hat{CE}(t,0,\mathbf{0},\mathbf{1})$ & $\hat{SE}(t,\mathbf{0})$ & $\hat{IE}(t,0,\mathbf{0})$ & $DE(t)$ & $IDE(t)$ \\ 
\hline \\[-7pt]
\multicolumn{5}{l}{Constant external and internal hazards} \\[3pt]
 2 & Obs. & -0.726 & -0.565 &  0.024 & -0.077 & -0.015 & -0.067 & -0.006 \\ 
   & Bernoulli & -0.722 & -0.713 &  0.023 & -0.075 & -0.018 & -0.068 & -0.032 \\ 
   & Block & -0.742 & -0.832 &  0.024 & -0.076 & -0.021 & -0.045 &     - \\
   & Cluster & -0.755 & -0.732 &  0.024 & -0.077 & -0.020 & -0.097 &     - \\[3pt]
 4 & Obs. & -0.656 & -0.708 &  0.050 & -0.069 & -0.035 & -0.061 & -0.032 \\ 
   & Bernoulli & -0.682 & -0.724 &  0.052 & -0.073 & -0.034 & -0.062 & -0.051 \\ 
   & Block & -0.703 & -0.779 &  0.054 & -0.074 & -0.036 & -0.044 &     - \\
   & Cluster & -0.677 & -0.808 &  0.053 & -0.073 & -0.039 & -0.127 &     - \\[3pt]
 8 & Obs. & -0.683 & -0.764 &  0.100 & -0.078 & -0.060 & -0.061 &  0.016 \\ 
   & Bernoulli & -0.678 & -0.710 &  0.098 & -0.074 & -0.053 & -0.061 & -0.074 \\ 
   & Block & -0.691 & -0.680 &  0.099 & -0.078 & -0.052 & -0.049 &     - \\
   & Cluster & -0.648 & -0.833 &  0.095 & -0.071 & -0.061 & -0.167 &     - \\[3pt]
\multicolumn{5}{l}{Constant external hazard and time-varying internal hazard} \\ [3pt]  
 2 & Obs. & -0.650 & -0.688 &  0.029 & -0.072 & -0.022 & -0.068 & -0.047 \\ 
   & Bernoulli & -0.714 & -0.665 &  0.032 & -0.077 & -0.022 & -0.066 & -0.047 \\ 
   & Block & -0.575 & -0.685 &  0.025 & -0.063 & -0.021 & -0.026 &     - \\
   & Cluster & -0.721 & -0.630 &  0.032 & -0.078 & -0.021 & -0.104 &     - \\[3pt]
 4 & Obs. & -0.704 & -0.662 &  0.082 & -0.080 & -0.039 & -0.070 & -0.058 \\ 
   & Bernoulli & -0.694 & -0.685 &  0.084 & -0.082 & -0.044 & -0.071 & -0.065 \\ 
   & Block & -0.721 & -0.622 &  0.088 & -0.086 & -0.042 & -0.044 &     - \\
   & Cluster & -0.700 & -0.706 &  0.086 & -0.084 & -0.046 & -0.161 &     - \\[3pt]
 8 & Obs. & -0.669 & -0.722 &  0.136 & -0.065 & -0.046 & -0.070 & -0.233 \\ 
   & Bernoulli & -0.681 & -0.633 &  0.133 & -0.065 & -0.039 & -0.068 & -0.098 \\ 
   & Block & -0.703 & -0.706 &  0.138 & -0.067 & -0.043 & -0.044 &     - \\
   & Cluster & -0.731 & -0.611 &  0.141 & -0.071 & -0.038 & -0.223 &     - \\[3pt]
\multicolumn{5}{l}{Time-varying external and internal hazards} \\ [3pt] 
 2 & Obs. & -0.650 & -0.688 &  0.029 & -0.072 & -0.022 & -0.068 & -0.047 \\ 
   & Bernoulli & -0.714 & -0.665 &  0.032 & -0.077 & -0.022 & -0.066 & -0.047 \\ 
   & Block & -0.575 & -0.685 &  0.025 & -0.063 & -0.021 & -0.026 &     - \\
   & Cluster & -0.721 & -0.630 &  0.032 & -0.078 & -0.021 & -0.104 &     - \\[3pt]
 4 & Obs. & -0.704 & -0.662 &  0.082 & -0.080 & -0.039 & -0.070 & -0.058 \\ 
   & Bernoulli & -0.694 & -0.685 &  0.084 & -0.082 & -0.044 & -0.071 & -0.065 \\ 
   & Block & -0.721 & -0.622 &  0.088 & -0.086 & -0.042 & -0.044 &     - \\
   & Cluster & -0.700 & -0.706 &  0.086 & -0.084 & -0.046 & -0.161 &     - \\[3pt]
 8 & Obs. & -0.669 & -0.722 &  0.136 & -0.065 & -0.046 & -0.070 & -0.233 \\ 
   & Bernoulli & -0.681 & -0.633 &  0.133 & -0.065 & -0.039 & -0.068 & -0.098 \\ 
   & Block & -0.703 & -0.706 &  0.138 & -0.067 & -0.043 & -0.044 &     - \\
   & Cluster & -0.731 & -0.611 &  0.141 & -0.071 & -0.038 & -0.223 &     - \\
\bottomrule
\end{tabular}
\caption{Simulation results showing estimates of hazard coefficients, exposure-marginalized causal estimands, and direct and indirect effect under $e^{\beta_1}=0.5$ and $e^{\beta_2}=0.5$. Estimands are evaluated under three transmission scenarios -- constant external and internal hazards of $\alpha(t)=0.3$ and $\gamma(t)=3$, constant external hazard and time-varying internal hazard of $\alpha(t)=0.3$ and $\gamma(t)=6 e^{-0.5(t-t_j)}$, and time-varying internal and external hazard of $\alpha(t)=0.2 (1+\sin (2 \pi t+\frac{\pi}{2}))$ and $\gamma(t)=6 e^{-0.5(t-t_j)}$. The individual covariates are correlated with $\rho=-0.1$ and have coefficients $e^{\theta_1}=e^{\theta_2}=0.9$. Clusters of $2$, $4$, and $8$ are observed at $0.4$ year, $0.3$ year, and $0.2$ year, respectively. }
\label{tab:estimation1}
\end{table}

\begin{table}
\small
\centering
\begin{tabular}{ccccccccc}
\toprule
 \multirow{2}{*}{Cluster}& \multirow{2}{*}{Treatment} & \multicolumn{2}{c}{Hazard estimands} & \multicolumn{5}{c}{Probability estimands} \\
\cmidrule(lr){3-4} \cmidrule(lr){5-9}
& & $\hat{\beta}_1$ & $\hat{\beta}_2$ & $\hat{CE}(t,0,\mathbf{0},\mathbf{1})$ & $\hat{SE}(t,\mathbf{0})$ & $\hat{IE}(t,0,\mathbf{0})$ & $DE(t)$ & $IDE(t)$ \\ 
\hline \\[-7pt]
\multicolumn{5}{l}{Constant external and internal hazards} \\[3pt]
 2 & Obs. & -0.719 & -0.772 &  0.024 & -0.074 & -0.021 & -0.069 & -0.050 \\ 
   & Bernoulli & -0.722 & -0.713 &  0.023 & -0.075 & -0.018 & -0.068 & -0.032 \\ 
   & Block & -0.742 & -0.832 &  0.024 & -0.076 & -0.021 & -0.045 &     - \\ 
   & Cluster & -0.755 & -0.732 &  0.024 & -0.077 & -0.020 & -0.097 &     - \\[3pt] 
 4 & Obs. & -0.715 & -0.730 &  0.055 & -0.074 & -0.034 & -0.067 & -0.054 \\ 
   & Bernoulli & -0.682 & -0.724 &  0.052 & -0.073 & -0.034 & -0.062 & -0.051 \\ 
   & Block & -0.703 & -0.779 &  0.054 & -0.074 & -0.036 & -0.044 &     - \\ 
   & Cluster & -0.677 & -0.808 &  0.053 & -0.073 & -0.039 & -0.127 &     - \\[3pt]  
 8 & Obs. & -0.703 & -0.710 &  0.101 & -0.082 & -0.056 & -0.062 & -0.165 \\ 
   & Bernoulli & -0.678 & -0.710 &  0.098 & -0.074 & -0.053 & -0.061 & -0.074 \\ 
   & Block & -0.691 & -0.680 &  0.099 & -0.078 & -0.052 & -0.049 &     - \\ 
   & Cluster & -0.648 & -0.833 &  0.095 & -0.071 & -0.061 & -0.167 &     - \\[3pt] 
\multicolumn{5}{l}{Constant external hazard and time-varying internal hazard} \\ [3pt]  
 2 & Obs. & -0.596 & -0.552 &  0.025 & -0.064 & -0.017 & -0.055 & -0.014 \\ 
   & Bernoulli & -0.714 & -0.665 &  0.032 & -0.077 & -0.022 & -0.066 & -0.047 \\ 
   & Block & -0.575 & -0.685 &  0.025 & -0.063 & -0.021 & -0.026 &     - \\ 
   & Cluster & -0.721 & -0.630 &  0.032 & -0.078 & -0.021 & -0.104 &     - \\[3pt]  
 4 & Obs. & -0.706 & -0.739 &  0.087 & -0.084 & -0.046 & -0.072 & -0.061 \\ 
   & Bernoulli & -0.694 & -0.685 &  0.084 & -0.082 & -0.044 & -0.071 & -0.065 \\ 
   & Block & -0.721 & -0.622 &  0.088 & -0.086 & -0.042 & -0.044 &     - \\ 
   & Cluster & -0.700 & -0.706 &  0.086 & -0.084 & -0.046 & -0.161 &     - \\[3pt]  
 8 & Obs. & -0.686 & -0.715 &  0.138 & -0.063 & -0.043 & -0.065 & -0.153 \\ 
   & Bernoulli & -0.681 & -0.633 &  0.133 & -0.065 & -0.039 & -0.068 & -0.098 \\ 
   & Block & -0.703 & -0.706 &  0.138 & -0.067 & -0.043 & -0.044 &     - \\ 
   & Cluster & -0.731 & -0.611 &  0.141 & -0.071 & -0.038 & -0.223 &     - \\[3pt] 
\multicolumn{5}{l}{Time-varying external and internal hazards} \\ [3pt] 
 2 & Obs. & -0.596 & -0.552 &  0.025 & -0.064 & -0.017 & -0.055 & -0.014 \\ 
   & Bernoulli & -0.714 & -0.665 &  0.032 & -0.077 & -0.022 & -0.066 & -0.047 \\ 
   & Block & -0.575 & -0.685 &  0.025 & -0.063 & -0.021 & -0.026 &     - \\ 
   & Cluster & -0.721 & -0.630 &  0.032 & -0.078 & -0.021 & -0.104 &     - \\[3pt]  
 4 & Obs. & -0.706 & -0.739 &  0.087 & -0.084 & -0.046 & -0.072 & -0.061 \\ 
   & Bernoulli & -0.694 & -0.685 &  0.084 & -0.082 & -0.044 & -0.071 & -0.065 \\ 
   & Block & -0.721 & -0.622 &  0.088 & -0.086 & -0.042 & -0.044 &     - \\ 
   & Cluster & -0.700 & -0.706 &  0.086 & -0.084 & -0.046 & -0.161 &     - \\[3pt]  
 8 & Obs. & -0.686 & -0.715 &  0.138 & -0.063 & -0.043 & -0.065 & -0.153 \\ 
   & Bernoulli & -0.681 & -0.633 &  0.133 & -0.065 & -0.039 & -0.068 & -0.098 \\ 
   & Block & -0.703 & -0.706 &  0.138 & -0.067 & -0.043 & -0.044 &     - \\ 
   & Cluster & -0.731 & -0.611 &  0.141 & -0.071 & -0.038 & -0.223 &     - \\ 
\bottomrule
\end{tabular}
\caption{Simulation results showing estimates of hazard coefficients, exposure-marginalized causal estimands, and direct and indirect effect under $e^{\beta_1}=0.5$ and $e^{\beta_2}=0.5$. Estimands are evaluated under three transmission scenarios -- constant external and internal hazards of $\alpha(t)=0.3$ and $\gamma(t)=3$, constant external hazard and time-varying internal hazard of $\alpha(t)=0.3$ and $\gamma(t)=6 e^{-0.5(t-t_j)}$, and time-varying internal and external hazard of $\alpha(t)=0.2 (1+\sin (2 \pi t+\frac{\pi}{2}))$ and $\gamma(t)=6 e^{-0.5(t-t_j)}$. The individual covariates are correlated with $\rho=0.1$ and have coefficients $e^{\theta_1}=e^{\theta_2}=0.9$. Clusters of $2$, $4$, and $8$ are observed at $0.4$ year, $0.3$ year, and $0.2$ year, respectively. }
\label{tab:estimation2}
\end{table}

\begin{table}
\small
\centering
\begin{tabular}{ccccccccc}
\toprule
 \multirow{2}{*}{Cluster}& \multirow{2}{*}{Treatment} & \multicolumn{2}{c}{Hazard estimands} & \multicolumn{5}{c}{Probability estimands} \\
\cmidrule(lr){3-4} \cmidrule(lr){5-9}
& & $\hat{\beta}_1$ & $\hat{\beta}_2$ & $\hat{CE}(t,0,\mathbf{0},\mathbf{1})$ & $\hat{SE}(t,\mathbf{0})$ & $\hat{IE}(t,0,\mathbf{0})$ & $DE(t)$ & $IDE(t)$ \\ 
\hline \\[-7pt]
\multicolumn{5}{l}{Constant external and internal hazards} \\[3pt]
 2 & Obs. & -0.721 &  0.029 &  0.024 & -0.077 &  0.001 & -0.068 &  -0.002 \\ 
    & Bernoulli & -0.715 & -0.105 &  0.022 & -0.074 & -0.005 & -0.068 & -0.025 \\ 
    & Block & -0.742 & -0.203 &  0.024 & -0.076 & -0.008 & -0.052 &     - \\
    & Cluster & -0.755 & -0.028 &  0.024 & -0.077 &  0.000 & -0.091 &     - \\[3pt] 
  4 & Obs. & -0.671 & -0.039 &  0.048 & -0.070 & -0.002 & -0.066 & -0.016 \\ 
    & Bernoulli & -0.686 & -0.102 &  0.050 & -0.073 & -0.006 & -0.065 & -0.039 \\ 
    & Block & -0.700 & -0.159 &  0.052 & -0.074 & -0.008 & -0.051 &     - \\
    & Cluster & -0.677 & -0.195 &  0.051 & -0.073 & -0.011 & -0.115 &     - \\[3pt]
  8 & Obs. & -0.691 & -0.161 &  0.093 & -0.080 & -0.014 & -0.066 &  0.038 \\ 
    & Bernoulli & -0.688 & -0.150 &  0.092 & -0.076 & -0.011 & -0.065 & -0.055 \\ 
    & Block & -0.690 & -0.098 &  0.091 & -0.078 & -0.007 & -0.056 &     - \\
    & Cluster & -0.648 & -0.216 &  0.088 & -0.071 & -0.015 & -0.149 &     - \\[3pt] 
\multicolumn{5}{l}{Constant external hazard and time-varying internal hazard} \\ [3pt]  
  2 & Obs. & -0.648 & -0.090 &  0.028 & -0.071 & -0.002 & -0.069 & -0.037 \\ 
    & Bernoulli & -0.722 & -0.083 &  0.032 & -0.078 & -0.003 & -0.067 & -0.037 \\ 
    & Block & -0.576 & -0.089 &  0.025 & -0.063 & -0.003 & -0.036 &     - \\
    & Cluster & -0.721 & -0.028 &  0.032 & -0.078 & -0.001 & -0.097 &   - \\[3pt] 
  4 & Obs. & -0.703 & -0.066 &  0.078 & -0.080 & -0.004 & -0.071 & -0.043 \\ 
    & Bernoulli & -0.686 & -0.123 &  0.081 & -0.081 & -0.008 & -0.071 & -0.051 \\ 
    & Block & -0.720 & -0.040 &  0.085 & -0.087 & -0.003 & -0.050 &     - \\
    & Cluster & -0.700 & -0.086 &  0.082 & -0.083 & -0.006 & -0.141 &     - \\[3pt]
  8 & Obs. & -0.681 & -0.144 &  0.127 & -0.066 & -0.008 & -0.071 & -0.210 \\ 
    & Bernoulli & -0.681 & -0.059 &  0.123 & -0.065 & -0.003 & -0.067 & -0.078 \\ 
    & Block & -0.695 & -0.121 &  0.127 & -0.067 & -0.006 & -0.047 &     - \\
    & Cluster & -0.731 & -0.035 &  0.131 & -0.072 & -0.002 & -0.190 &     - \\[3pt]
\multicolumn{5}{l}{Time-varying external and internal hazards} \\ [3pt]
  2 & Obs. & -0.648 & -0.090 &  0.028 & -0.071 & -0.002 & -0.069 & -0.037 \\ 
    & Bernoulli & -0.722 & -0.083 &  0.032 & -0.078 & -0.003 & -0.067 & -0.037 \\ 
    & Block & -0.576 & -0.089 &  0.025 & -0.063 & -0.003 & -0.036 &     - \\
    & Cluster & -0.721 & -0.028 &  0.032 & -0.078 & -0.001 & -0.097 &     - \\[3pt]
  4 & Obs. & -0.703 & -0.066 &  0.078 & -0.080 & -0.004 & -0.071 & -0.043 \\ 
    & Bernoulli & -0.686 & -0.123 &  0.081 & -0.081 & -0.008 & -0.071 & -0.051 \\ 
    & Block & -0.720 & -0.040 &  0.085 & -0.087 & -0.003 & -0.050 &     - \\
    & Cluster & -0.700 & -0.086 &  0.082 & -0.083 & -0.006 & -0.141 &     - \\[3pt]
  8 & Obs. & -0.681 & -0.144 &  0.127 & -0.066 & -0.008 & -0.071 & -0.210 \\ 
    & Bernoulli & -0.681 & -0.059 &  0.123 & -0.065 & -0.003 & -0.067 & -0.078 \\ 
    & Block & -0.695 & -0.121 &  0.127 & -0.067 & -0.006 & -0.047 &     - \\
    & Cluster & -0.731 & -0.035 &  0.131 & -0.072 & -0.002 & -0.190 &     - \\
\bottomrule
\end{tabular}
\caption{Simulation results showing estimates of hazard coefficients, exposure-marginalized causal estimands, and direct and indirect effect under $e^{\beta_1}=0.5$ and $e^{\beta_2}=0.9$. Estimands are evaluated under three transmission scenarios -- constant external and internal hazards of $\alpha(t)=0.3$ and $\gamma(t)=3$, constant external hazard and time-varying internal hazard of $\alpha(t)=0.3$ and $\gamma(t)=6 e^{-0.5(t-t_j)}$, and time-varying internal and external hazard of $\alpha(t)=0.2 (1+\sin (2 \pi t+\frac{\pi}{2}))$ and $\gamma(t)=6 e^{-0.5(t-t_j)}$. The individual covariates are correlated with $\rho=-0.1$ and have coefficients $e^{\theta_1}=e^{\theta_2}=0.9$. Clusters of $2$, $4$, and $8$ are observed at $0.4$ year, $0.3$ year, and $0.2$ year, respectively. }
\label{tab:estimation3}
\end{table}

\begin{table}
\small
\centering
\begin{tabular}{ccccccccc}
\toprule
 \multirow{2}{*}{Cluster}& \multirow{2}{*}{Treatment} & \multicolumn{2}{c}{Hazard estimands} & \multicolumn{5}{c}{Probability estimands} \\
\cmidrule(lr){3-4} \cmidrule(lr){5-9}
& & $\hat{\beta}_1$ & $\hat{\beta}_2$ & $\hat{CE}(t,0,\mathbf{0},\mathbf{1})$ & $\hat{SE}(t,\mathbf{0})$ & $\hat{IE}(t,0,\mathbf{0})$ & $DE(t)$ & $IDE(t)$ \\ 
\hline \\[-7pt]
\multicolumn{5}{l}{Constant external and internal hazards} \\[3pt]
 2 & Obs. & -0.122 & -2.481 &  0.006 & -0.016 & -0.038 & -0.016 & -0.056 \\ 
   & Bernoulli & -0.115 & -2.276 &  0.005 & -0.015 & -0.035 & -0.012 & -0.039 \\ 
   & Block & -0.113 & -2.265 &  0.005 & -0.015 & -0.036 &  0.021 &     - \\ 
   & Cluster & -0.139 & -2.477 &  0.006 & -0.016 & -0.038 & -0.054 &     - \\[3pt]
 4 & Obs. & -0.111 & -2.412 &  0.023 & -0.013 & -0.083 & -0.011 & -0.067 \\ 
   & Bernoulli & -0.077 & -2.287 &  0.022 & -0.010 & -0.078 & -0.007 & -0.063 \\ 
   & Block & -0.084 & -2.326 &  0.023 & -0.011 & -0.080 &  0.019 &     - \\ 
   & Cluster & -0.083 & -2.377 &  0.024 & -0.011 & -0.082 & -0.097 &     - \\ [3pt]
 8 & Obs. & -0.094 & -2.242 &  0.065 & -0.013 & -0.128 & -0.008 & -0.183 \\
   & Bernoulli & -0.065 & -2.360 &  0.060 & -0.008 & -0.127 & -0.006 & -0.092 \\ 
   & Block & -0.120 & -2.319 &  0.066 & -0.014 & -0.129 &  0.009 &     - \\ 
   & Cluster & -0.050 & -2.554 &  0.060 & -0.005 & -0.133 & -0.151 &     - \\[3pt]
\multicolumn{5}{l}{Constant external hazard and time-varying internal hazard} \\ [3pt]   
 2 & Obs. &  0.006 & -2.241 &  0.000 &  0.000 & -0.046 &  0.004 & -0.030 \\ 
   & Bernoulli & -0.097 & -2.268 &  0.006 & -0.011 & -0.050 & -0.008 & -0.058 \\ 
   & Block & -0.009 & -2.267 &  0.001 & -0.001 & -0.046 &  0.044 &     - \\ 
   & Cluster & -0.136 & -2.310 &  0.008 & -0.017 & -0.051 & -0.066 &     - \\[3pt] 
 4 & Obs. & -0.114 & -2.357 &  0.040 & -0.014 & -0.119 & -0.015 & -0.079 \\ 
   & Bernoulli & -0.109 & -2.319 &  0.042 & -0.016 & -0.118 & -0.015 & -0.088 \\ 
   & Block & -0.105 & -2.294 &  0.041 & -0.015 & -0.115 &  0.030 &     - \\ 
   & Cluster & -0.115 & -2.284 &  0.042 & -0.017 & -0.117 & -0.136 &     - \\[3pt] 
 8 & Obs. & -0.108 & -2.286 &  0.089 & -0.010 & -0.163 & -0.011 & -0.180 \\ 
   & Bernoulli & -0.074 & -2.260 &  0.082 & -0.007 & -0.160 & -0.009 & -0.126 \\ 
   & Block & -0.114 & -2.412 &  0.090 & -0.013 & -0.171 &  0.021 &     - \\ 
   & Cluster & -0.127 & -2.225 &  0.090 & -0.013 & -0.160 & -0.217 &     - \\[3pt] 
\multicolumn{5}{l}{Time-varying external and internal hazards} \\ [3pt]
 2 & Obs. &  0.006 & -2.241 &  0.000 &  0.000 & -0.046 &  0.004 & -0.030 \\ 
   & Bernoulli & -0.097 & -2.268 &  0.006 & -0.011 & -0.050 & -0.008 & -0.058 \\ 
   & Block & -0.009 & -2.267 &  0.001 & -0.001 & -0.046 &  0.044 &     - \\ 
   & Cluster & -0.136 & -2.310 &  0.008 & -0.017 & -0.051 & -0.066 &     - \\[3pt] 
 4 & Obs. & -0.114 & -2.357 &  0.040 & -0.014 & -0.119 & -0.015 & -0.079 \\ 
   & Bernoulli & -0.109 & -2.319 &  0.042 & -0.016 & -0.118 & -0.015 & -0.088 \\ 
   & Block & -0.105 & -2.294 &  0.041 & -0.015 & -0.115 &  0.030 &     - \\ 
   & Cluster & -0.115 & -2.284 &  0.042 & -0.017 & -0.117 & -0.136 &     - \\ [3pt]
 8 & Obs. & -0.108 & -2.286 &  0.089 & -0.010 & -0.163 & -0.011 & -0.180 \\ 
   & Bernoulli & -0.074 & -2.260 &  0.082 & -0.007 & -0.160 & -0.009 & -0.126 \\ 
   & Block & -0.114 & -2.412 &  0.090 & -0.013 & -0.171 &  0.021 &     - \\ 
   & Cluster & -0.127 & -2.225 &  0.090 & -0.013 & -0.160 & -0.217 &     - \\ 
\bottomrule
\end{tabular}
\caption{Simulation results showing estimates of hazard coefficients, exposure-marginalized causal estimands, and direct and indirect effect under $e^{\beta_1}=0.9$ and $e^{\beta_2}=0.1$. Estimands are evaluated under three transmission scenarios -- constant external and internal hazards of $\alpha(t)=0.3$ and $\gamma(t)=3$, constant external hazard and time-varying internal hazard of $\alpha(t)=0.3$ and $\gamma(t)=6 e^{-0.5(t-t_j)}$, and time-varying internal and external hazard of $\alpha(t)=0.2 (1+\sin (2 \pi t+\frac{\pi}{2}))$. The individual covariates are correlated with $\rho=0.1$ and have coefficients $e^{\theta_1}=e^{\theta_2}=0.9$. Clusters of $2$, $4$, and $8$ are observed at $0.4$ year, $0.3$ year, and $0.2$ year, respectively. }
\label{tab:estimation4}
\end{table}

\begin{figure}
    \centering 
\begin{subfigure}{0.49\textwidth}
  \includegraphics[width=\linewidth]{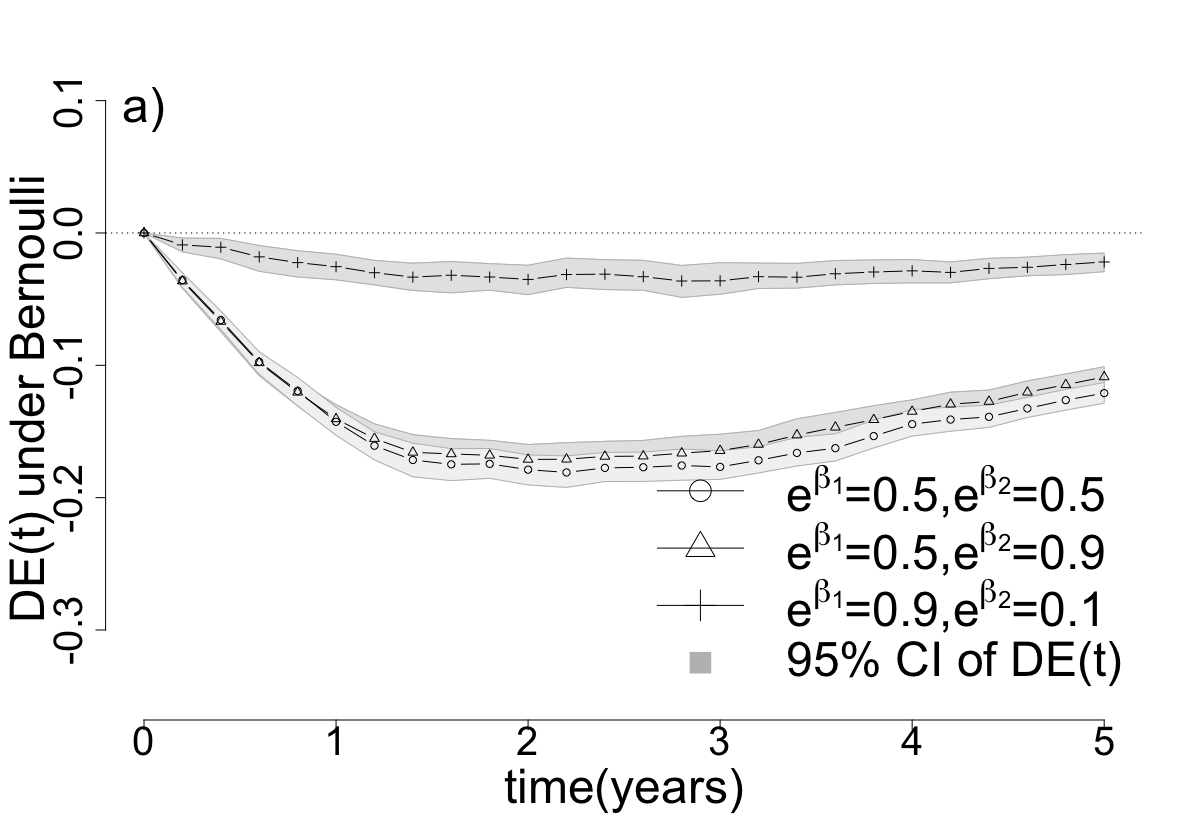}
\end{subfigure}\hfil 
\begin{subfigure}{0.49\textwidth}
  \includegraphics[width=\linewidth]{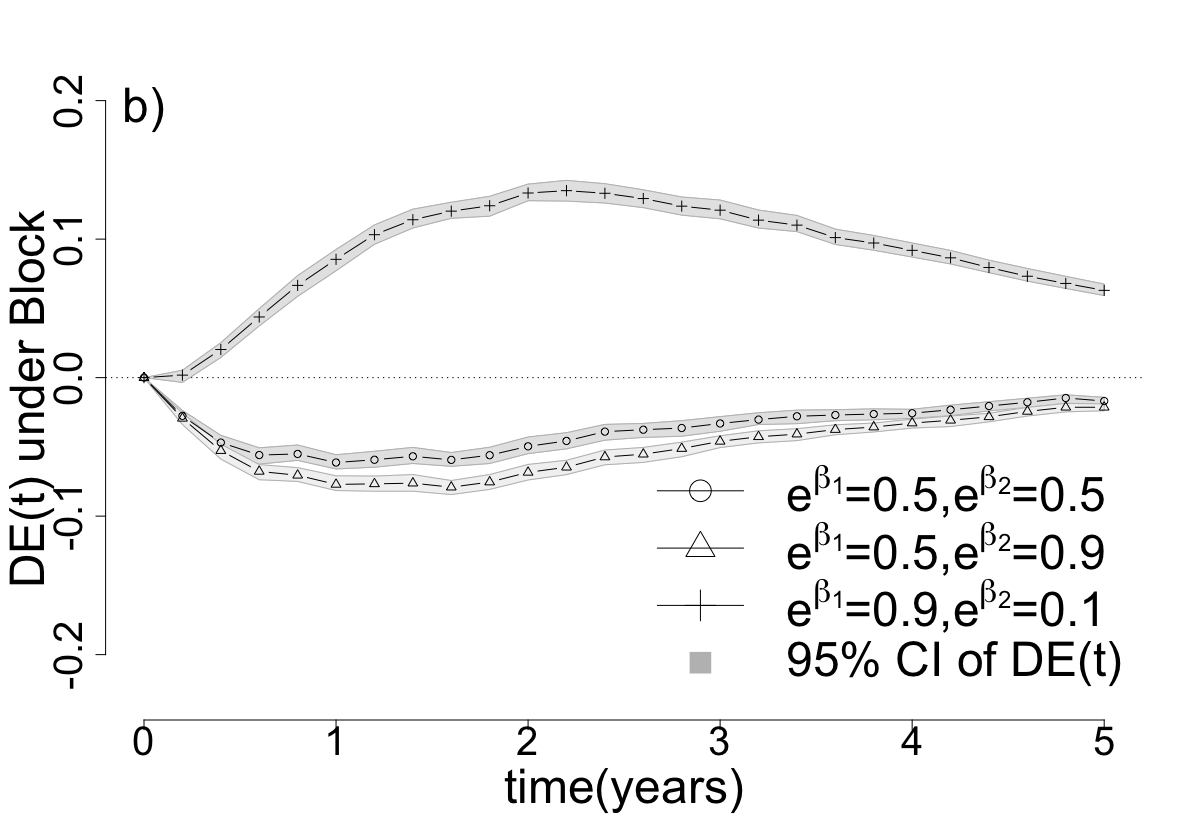}
\end{subfigure}\hfil 
\begin{subfigure}{0.49\textwidth}
  \includegraphics[width=\linewidth]{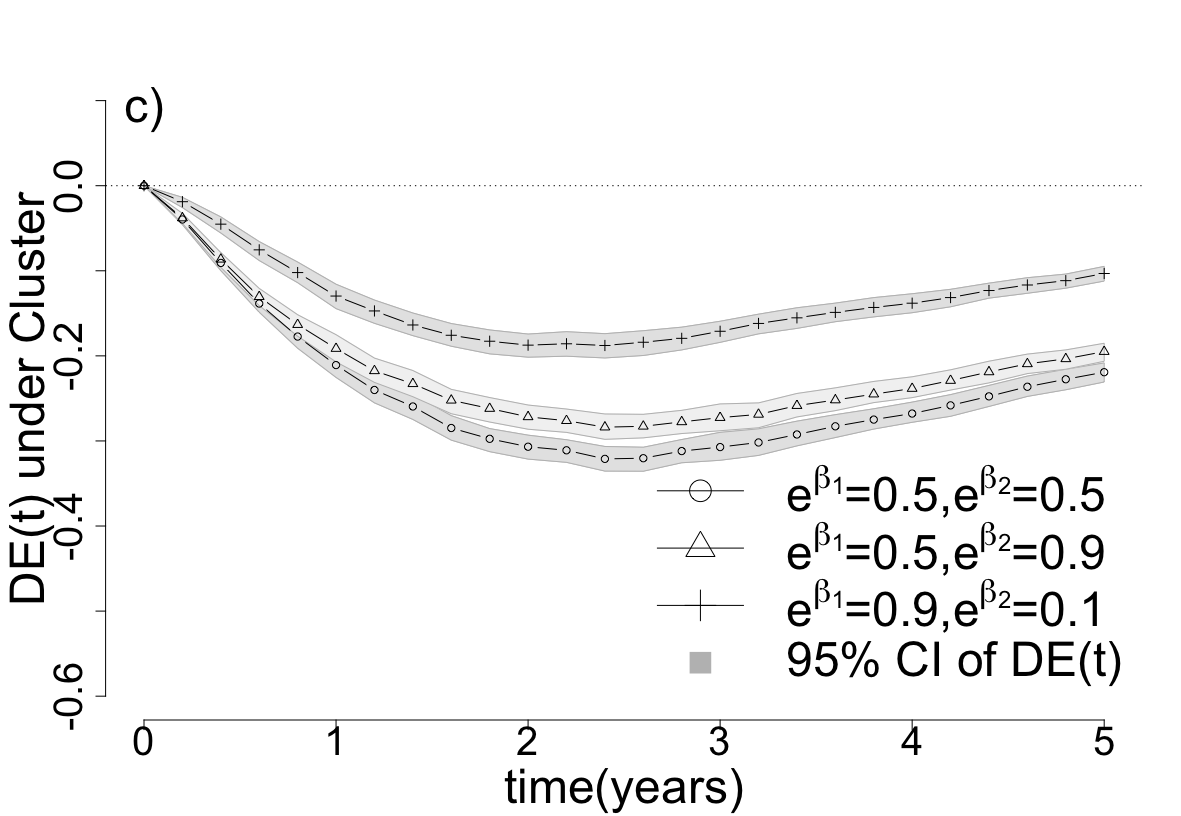}
\end{subfigure}\hfil 
\begin{subfigure}{0.49\textwidth}
  \includegraphics[width=\linewidth]{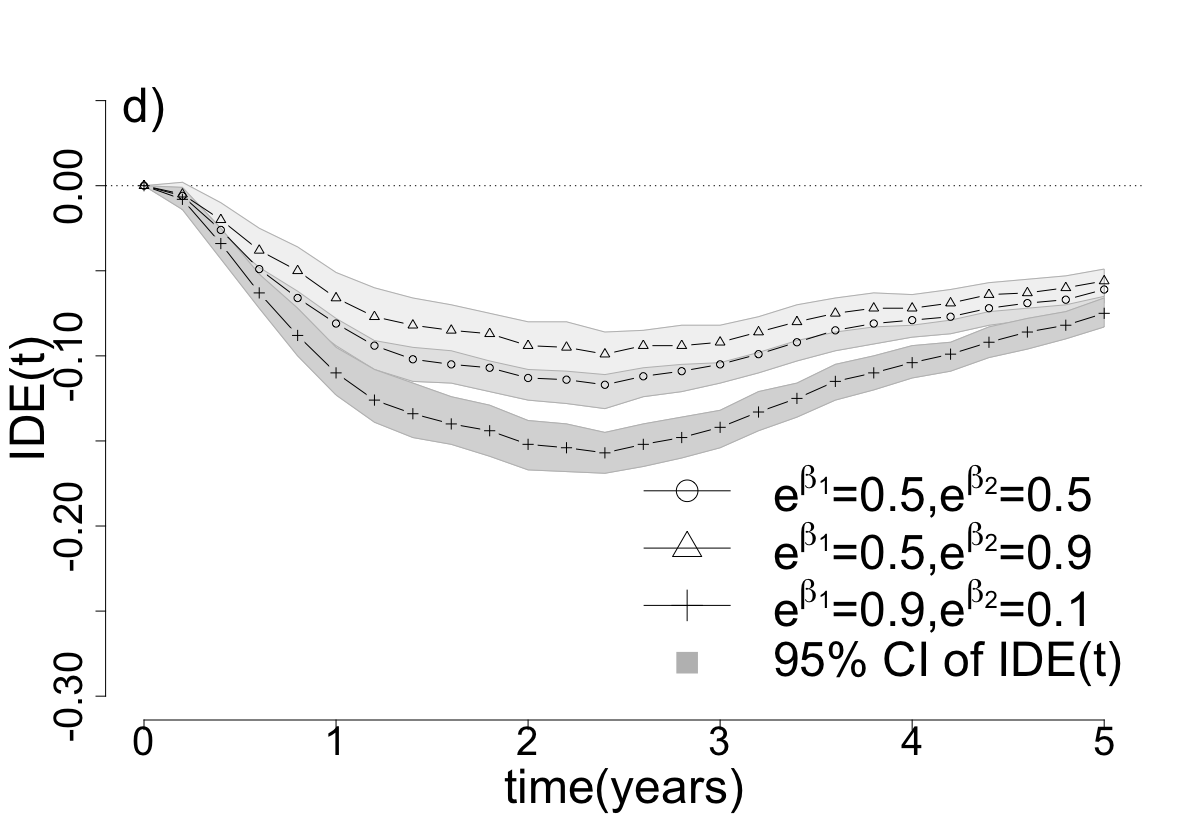}
\end{subfigure}\hfil 
\caption{Illustration of $\DE(t)$ under three randomization schemes and $\IDE(t)$, when the infectiousness effect is relatively strong, moderate and weak compared to the susceptibility effect. All graphs are generated under constant hazard scenario of $\alpha=0.3$, $\gamma=3$ in clusters of 2 individuals. Figures a), b) and c) show the $\DE(t)$ over time when individuals in clusters are under Bernoulli, Block, and Cluster randomization, respectively. Figure d) shows the $\IDE(t)$ over time in the two-stage randomization scheme. $\DE(t)$ performs poorly under block randomization, giving an effect with sign opposite that of the true susceptibility effect under one parameter setting. Related results are shown in Figures \ref{fig:shiftinpeak} and \ref{fig:correlatedX}.} 
\label{fig:compareDEIDE}
\end{figure}

\begin{figure}
  \centering 
\begin{subfigure}{0.49\textwidth}
  \includegraphics[width=\linewidth]{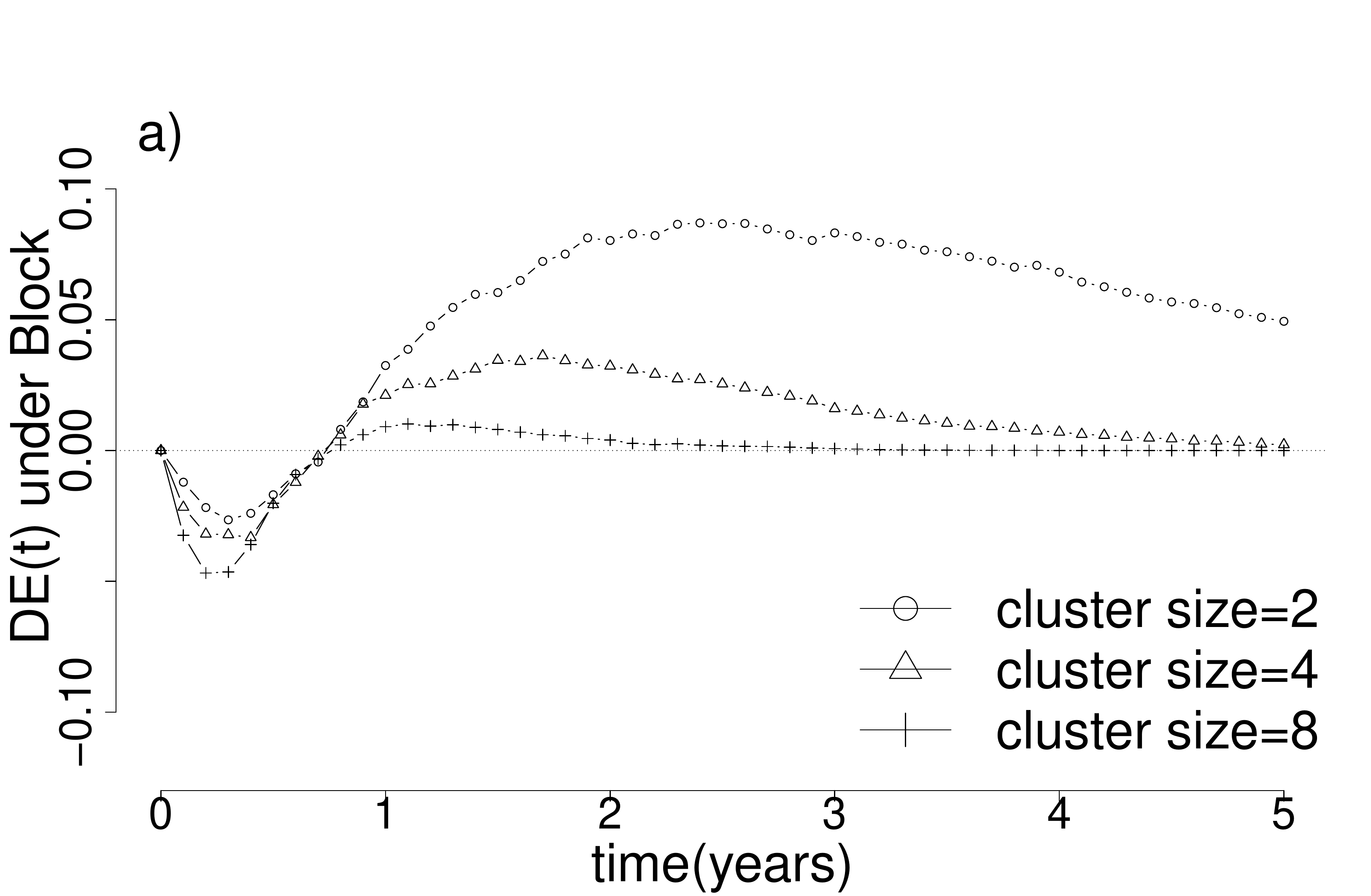}
\end{subfigure}\hfil 
\begin{subfigure}{0.49\textwidth}
  \includegraphics[width=\linewidth]{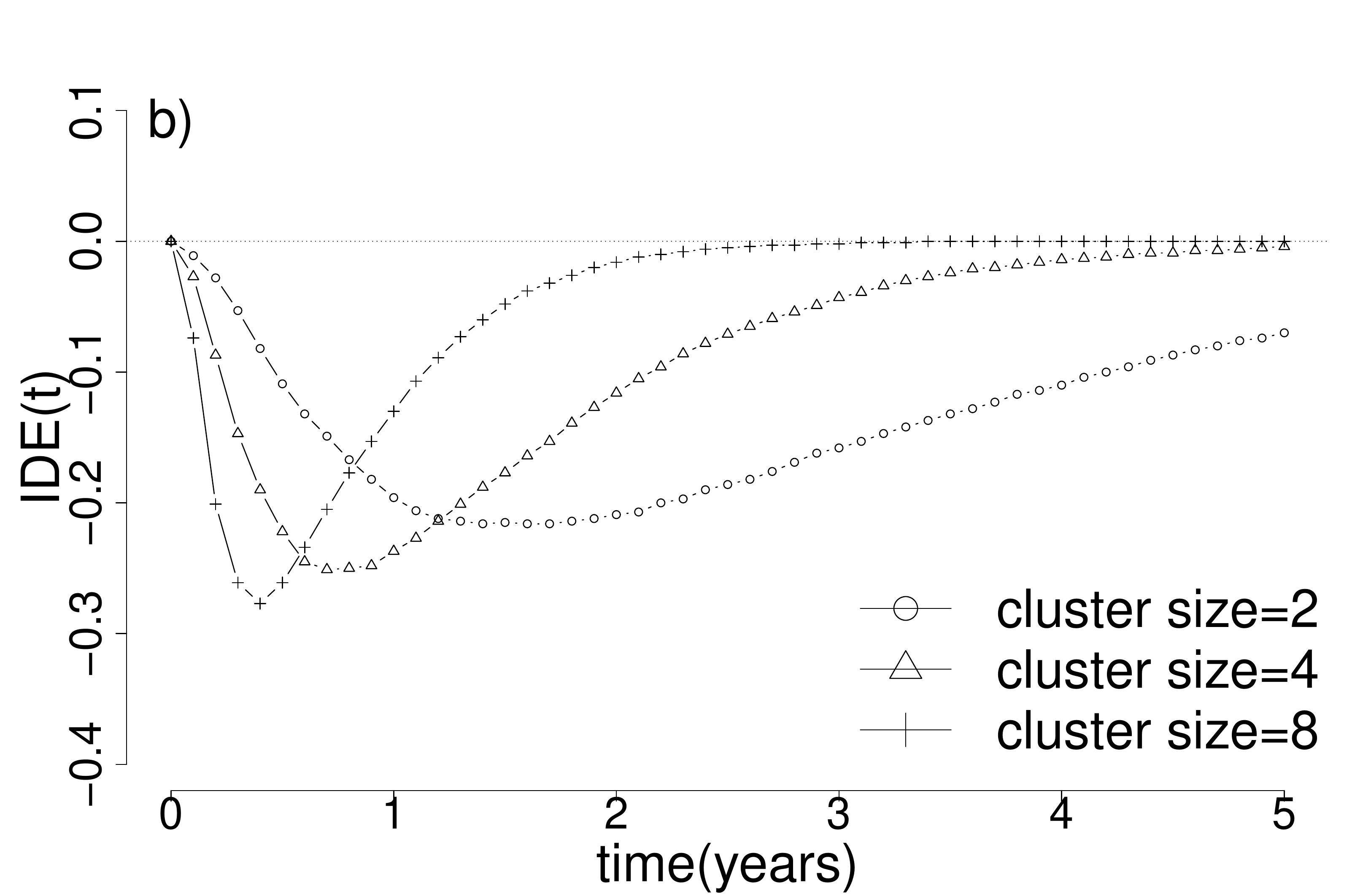}
\end{subfigure}\hfil 
\caption{Illustration of $\DE(t)$ under block randomization and $\IDE(t)$ under two-stage randomization in clusters of $2$, $4$, and $8$. Simulations are generated under $\alpha(t)=0.3$, $\gamma(t)=3$, $e^{\beta_2}=0.1$ and $e^{\beta_1}=0.9$, so that the infectiousness effect is stronger than susceptibility effect of the vaccine.  Figure a) shows $\DE(t)$ may have an opposite sign to $\beta_1<0$ under block randomization, and how the peak of this directional bias shifts with cluster size. As cluster size increases, each individual is generally exposed to greater number of infected neighbors and thus the infection process is accelerated. Figure b) shows the $\IDE(t)$ has the same sign as $\beta_2<0$, and how its peak shifts with cluster size.  } 
\label{fig:shiftinpeak}
\end{figure}


\subsection{Comparison to other estimands}

Statisticians and epidemiologists have proposed a wide variety of estimands summarizing the effect of interventions for infectious disease outcomes. 
In this section, we evaluate the performance of two commonly used estimands -- direct effect ($DE$) and indirect effect ($IDE$) defined in the two-stage cluster randomization design \cite{halloran1997study,halloran1999design,hudgens2008toward,halloran2010design,perez2014assessing,halloran2016dependent} -- and compare them to our estimands under the same simulation scenarios. We also employ GEE regression models to adjust for covariates in the estimation of $\DE$ and $\IDE$ in the observational study setting.

The ``attack risk'' of an infectious disease for individuals with treatment $x$ is defined as $\AR_x(t) = \E[Y_i(t)|X_i=x]$. A ratio of attack risks under different treatments (sometimes called ``relative cumulative incidence'') has been used to measure the vaccine effect on susceptibility, defined as $\VE_{\AR}(t) = 1- \AR_1(t)/\AR_0(t)$ \citep{greenwood1915statistics,francis1955evaluation}. A related estimand defined on the difference scale is referred to as the ``direct effect'' \citep{hudgens2008toward}, $\DE(t) = \AR_1(t) - \AR_0(t)$. Put simply, $DE(t)$ is the contrast of expectations of infection outcomes between treated individuals versus untreated individuals. In the Bernoulli setting, attack risks $\AR_x(t)$ conditioning only on treatment $X_i$ implicitly marginalize over the treatments assigned to other individuals within the same cluster. In the block randomization setting, we compare outcomes of treated individuals to those of untreated individuals in clusters where exactly half of individuals are treated and untreated. In the cluster randomization setting, we compare the outcome of individuals in treated clusters to those in control clusters. In the observational setting, $DE(t)$ is estimated with GEE regression model under the logit link function and exchangeable correlation structure to adjust for covariates.
In Tables \ref{tab:estimation1}--\ref{tab:estimation4}, we show the estimated $\DE(t)$ under three randomization strategies and under an observational setting. Figure \ref{fig:compareDEIDE} (a)--(c) show $\DE(t)$ over time when the infectiousness effect is relatively strong, moderate and weak compared to the susceptibility effect. 
Figure \ref{fig:shiftinpeak} (a) shows how $\DE(t)$ changes with cluster size under block randomization when the infectiousness effect is much stronger than the susceptibility effect.

\citet{halloran1991direct,halloran1994exposure,halloran1997study} and \citet{rhodes1996counting} warned that $\VE_{\AR}(t)$ (or equivalently, $\DE(t)$) may be a biased approximation to the susceptibility effect due to differential exposure to infection between treated and untreated individuals in clusters. This differential exposure to infection is due to different distributions of neighbors' infection times among treated and untreated, which cannot always be controlled by randomization at baseline. \citet{morozova2018risk}, \citet{eck2019randomization} and \citet{cai2021identification} proved similar results. 

Our simulation results also demonstrate the potential bias in  $\DE(t)$. In particular, we show that under block randomization, $DE(t)$ can have a different sign from that of $\hat{\beta_1}$ and $\hat{\SE}(t,\mathbf{0})$ in scenarios in Table \ref{tab:estimation4}  where the infectiousness effect is much stronger than the susceptibility effect. Figures \ref{fig:compareDEIDE} (b) and \ref{fig:shiftinpeak} (a) show how this bias changes with time and cluster size. We see that this bias can still be substantial in clusters of larger size. 
To better illustrate the reasons behind for misleading performance of $\DE(t)$, we simulate outcomes under an intervention whose true susceptibility effect is null with $\beta_1=0$ and whose true infectiousness effect is negative (protective) with $\beta_2<0$, and simulate an infectious disease with $\alpha(t)=0.3$ and $\gamma(t)=3$. In Figure~\ref{fig:correlatedX}, we show the calculated $\DE(t)$ over time for different cluster sizes when treatment assignment for $X_i$ and $X_j$ ($i \neq j$) is positively correlated and when treatment assignment for $X_i$ and $X_j$ ($i \neq j$) is negatively correlated, achieved by making $\rho$ positive and negative respectively in \eqref{eq:rho}. We see that for positively correlated treatment assignment, $\DE(t)>0$ for all $t>0$ and all cluster sizes, and similarly for negatively correlated treatment assignment, $\DE(t)<0$ for all $t>0$ and all cluster sizes.  Theorem 2 of \citet{cai2021identification} provides a nonparametric proof of this apparent paradox in the two-person partnership scenario.

\begin{figure}
  \centering 
\begin{subfigure}{0.49\textwidth}
  \includegraphics[width=\linewidth]{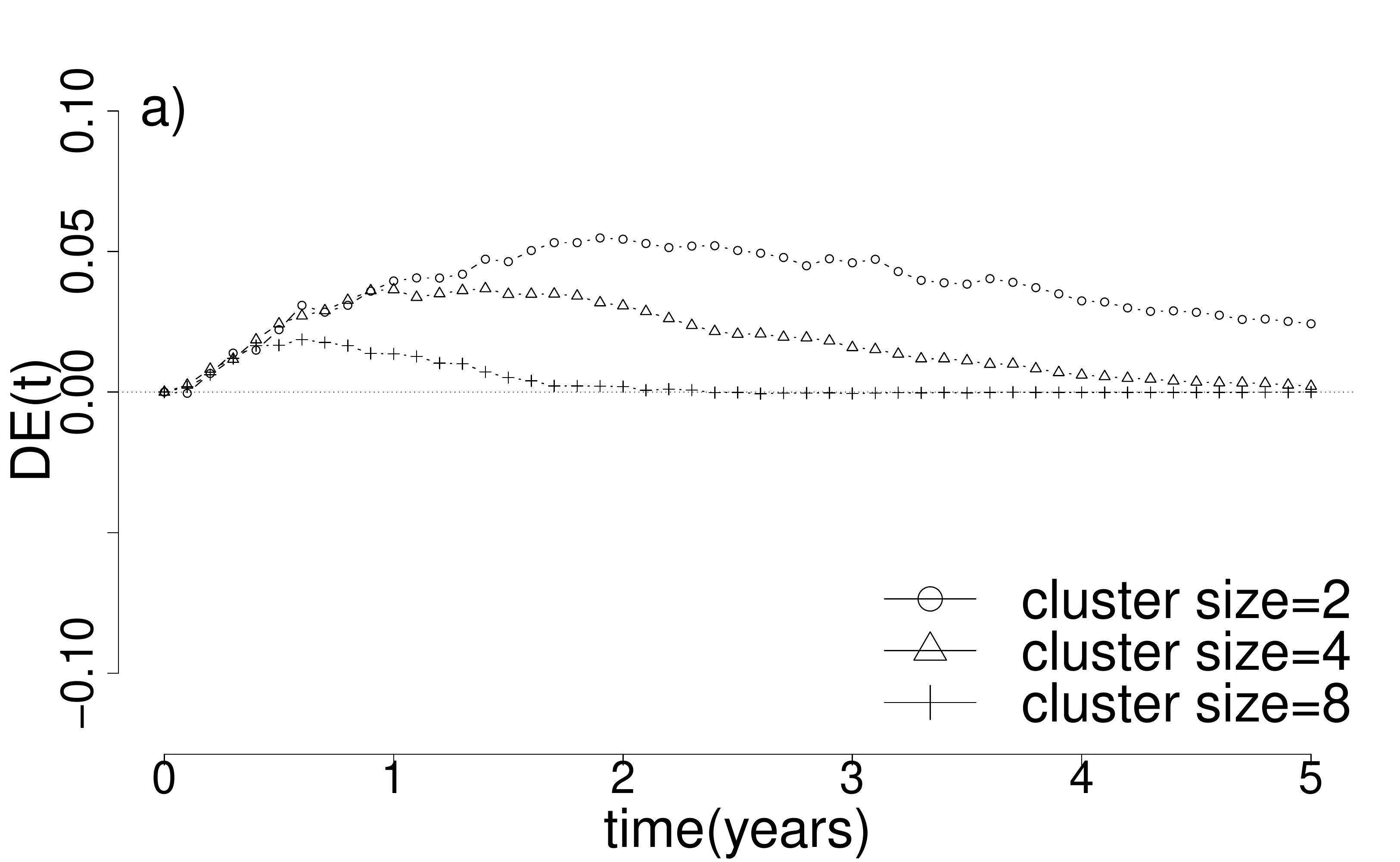}
\end{subfigure}\hfil 
\begin{subfigure}{0.49\textwidth}
  \includegraphics[width=\linewidth]{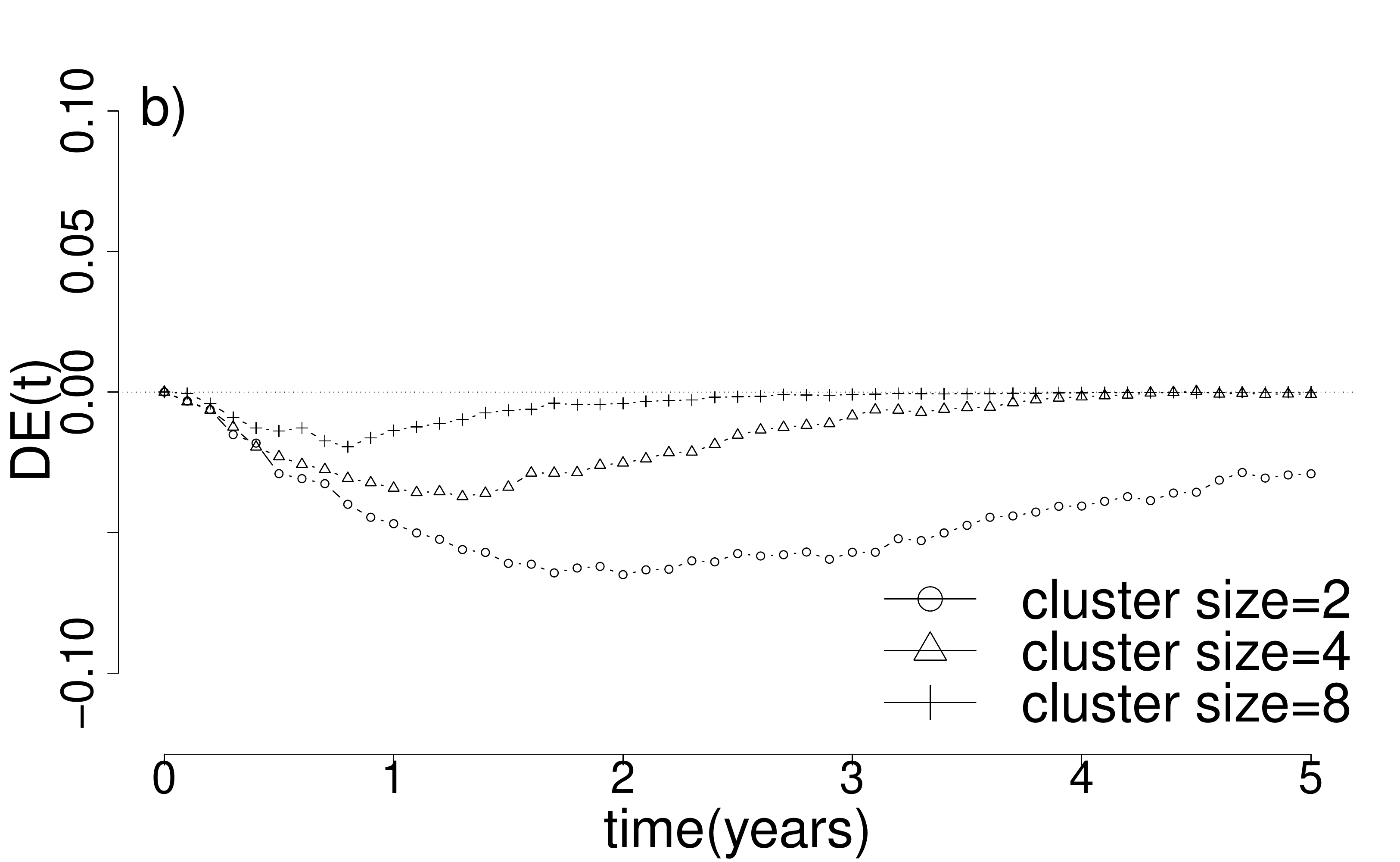}
\end{subfigure}\hfil 
\caption{Illustration of $\DE(t)$ when treatment assignment is correlated randomized in clusters of $2$, $4$, and $8$. Simulations are generated under $\alpha(t)=0.3$, $\gamma(t)=3$, $e^{\beta_1}=1$ and $e^{\beta_2}=0.1$, so that the susceptibility effect of the vaccine is null and the infectiousness effect decreases the transmission ability by $90\%$.  Figure a) shows $\DE(t)$ is strictly positive, contrary to $\beta_1=0$ when treatment assignment is negatively correlated with $\rho=-0.9$ for cluster size of 2, $\rho=-0.3$ for cluster size of 4, and $\rho=-0.1$ for cluster size of 8. Figure b) shows the $\DE(t)$ is strictly negative, contrary to $\beta_1=0$ when treatment assignment is positively correlated with $\rho=0.9$ for cluster size of 2, $\rho=0.3$ for cluster size of 4, and $\rho=0.1$ for cluster size of 8.} 
\label{fig:correlatedX}
\end{figure}

To capture the indirect effect of others' treatments on individual $i$, researchers have defined indirect effects under a two-stage randomization design \citep{struchiner1990behaviour,halloran1991direct,halloran1991study,halloran1995causal,hudgens2006causal,hudgens2008toward}. In the two-stage randomization design, clusters are first randomized to receive either high or low coverage of vaccines, then individuals in each cluster are randomized to receive vaccines according to the coverage rate chosen for this cluster.  In Tables \ref{tab:estimation1}--\ref{tab:estimation4}, we estimate 
\[ \IDE(t)=\sum_{ |\xminusi |=\frac{n}{2}} \E[Y_i(t)|X_i=0,\Xminusi=\xminusi] p(\xminusi)-\sum_{ | \xminusi |=0} \E[Y_i(t)| X_i=0,\X_{(i)}=\xminusi] p(\xminusi) \]
to compare the infection rates between untreated individuals in clusters of $50\%$ coverage rate and untreated individuals in clusters of $0\%$ coverage rate. In the observational setting, we apply the GEE model under the logit link function and exchangeable correlation structure to estimate $\IDE(t)$.  The $\IDE(t)$ estimand preserves the sign of $\beta_2$, as shown in Tables \ref{tab:estimation1}-\ref{tab:estimation4}. In Figure \ref{fig:compareDEIDE} (d), we show $\IDE(t)$ over time under the circumstances when the infectiousness effect is relatively strong, moderate and weak compared to the susceptibility effect. In Figure \ref{fig:shiftinpeak} (b), we show how $\IDE(t)$ changes over study time across different cluster sizes. 

 
\section{Discussion}

Causal dependence between infection outcomes is a major obstacle to identification of meaningful treatment effects in infectious disease epidemiology \citep{struchiner1990behaviour,halloran1991study,halloran1991direct,halloran1995causal,eisenberg2003bias,hudgens2008toward,vanderweele2011bounding,o2014estimating}.  Inferential approaches that neglect this dependence can result in misleading estimates of vaccine effects, even in randomized trials \citep{morozova2018risk,eck2019randomization,cai2021identification}.  In this paper, we have presented a framework for estimating meaningful treatment effects under contagion by clarifying the temporal relationship between infection outcomes in clusters. We describe a new strategy for identifying average potential infection outcomes. The approach proceeds by first identifying the average infection outcome under a deterministic infection exposure history, then integrating this history with respect to a distribution over others' infections, in the situation where the focal individual is not present.  This work is a direct extension of prior work by \citet{cai2021identification} and \citet{ogburn2017vaccines} to the general cluster setting, with no restrictions on the direction of transmission or treatment assignment. 

We define effects that contrast potential infection outcomes when the infection history is either held fixed at a deterministic value (controlled) or integrated (marginalized).  One advantage of this approach is that it permits definition of causal estimands in which the infection history of others (or its distribution) is held constant in a contrast of individual treatments. In this way, the susceptibility effect of Definition \ref{defn:susceptibility} avoids the problem of differential exposure to infection between treated and untreated individuals which arises in the definition of the ``direct effect'' proposed by \citet{hudgens2008toward}. Simulation results show that while crude contrasts like $DE(t)$ may sometimes have the same sign as the true effect when the susceptibility effect is different from zero, they may be biased under certain randomization designs or when the infectiousness effect is strong. 

In this work, we have drawn a connection between nonparametric causal estimands \citep[e.g.][]{ogburn2017vaccines} for infectious disease intervention effects and statistical models for learning about the dynamics of tranmission \citep{kenah2011contact,kenah2015semiparametric,kenah2013non,kenah2008generation}. Hazard models provide a natural and well studied class of structural models for time-to-event data.  The decomposition \eqref{eq:Tdecomposition} and corresponding independence assumptions provide a natural way of turning a complex sequence of infection times into conditionally independent waiting times.  These assumptions can endow certain hazard model parameters, corresponding to covariate effects, with causal interpretations. The hazard model respects the logic of infection disease transmission: no transmission can occur before infection, and multiple waiting times to transmission can elapse concurrently.  Beyond the causal interpretation of hazard ratios described in Section \ref{sec:hazardestimands}, the fitted hazard models provide a convenient statistical framework in which to learn about both exposure-controlled and exposure-marginalized estimands via Theorems \ref{mainidentification} and \ref{thm:exposure}. 

Despite the generality of the proposed framework, the approach has several limitations.  First, it relies on exact observation of infection times. In reality, binary infection status outcomes may be measured either at regular intervals or at follow-up. Statistical tools for interval censoring and/or tied responses could be employed to fit hazard models of the kind proposed here, but identifiability of model parameters could become more challenging.  However, some statistical estimands do not require observation of exact infection times.  In particular, Corollary \ref{cor:obs} shows that under some circumstances, meaningful causal contrasts can be computed by conditioning on the full vector of treatments within the cluster, even when infection times are unobserved.  This suggest that randomized trials that identify average outcomes under fully specified treatment allocations within clusters may reveal causal effects of the kind studied here.  Another limitation of this work is that practical statistical estimation requires specification of a hazard model with parametric or semiparametric structure. Lack of knowledge about the true structure of the distribution of pairwise transmission waiting times could result in bias.  A potential mitigating factor is that statisticians have access to a wide variety of models for time-to-event data.  Almost any of these could be usefully employed in \eqref{eq:haz} if the requisite conditional independence assumptions are met.



\textbf{Acknowledgements}: This work was supported by National Institutes of Health (NIH) grants 1DP2HD09179 and R01 AI116770. 
We are grateful to 
Peter M. Aronow, 
Olga Morozova,
Elizabeth Ogburn, 
and
Virginia E. Pitzer.

\singlespacing
\bibliographystyle{unsrtnat}

\bibliography{contagion2}

\appendix

\section{Proofs}

\begin{proof}[Proof of Lemma \ref{Iidentification}]
By Assumptions \ref{as:exclusion} and \ref{as:positivity}, we can define $f_{I_i^0}(s|x_i,\bl_i)$ be the density of $I_i^0(x_i)$ when $\bL_i=\bl_i$. Let $F_{I_i^0}(s|x_i,\bl_i)$ be the corresponding cumulative distribution function, and by Assumption \ref{as:positivity}, we have $0<F_{I_i^0}(s|x_i,\bl_i)<1$ for all $s>0$, $x_i$, and $\bl_i$. Thus, we can write,
  \[ \frac{f_{I_i^0}(s|x_i,\bl_i)}{1-F_{I_i^0}(s|x_i,\bl_i)} = -\frac{d}{ds} \log(1-F_{I_i^0}(s|x_i,\bl_i)) . \]
Then rearranging, we have: 
\begin{equation*}
\begin{split}
F_{I_i^0}(s|x_i,\bl_i) &= 1 - \exp\left[ -\int_0^s \frac{f_{I_i^0}(u|x_i,\bl_i)}{1-F_{I_i^0}(u|x_i,\bl_i)} du \right] \\
               &= 1 - \exp\left[ -\int_0^s \frac{f_{I_i^0}(u|x_i,\bl_i)  \prod_{j \neq i} \left( 1-F_{I_j^0}(u|x_j,\bl_j) \right) }{\big{(} 1-F_{I_i^0}(u|x_i,\bl_i)\big{)}\prod_{j \neq i} \left( 1-F_j(u|x_j,\bl_j) \right) } du \right] \\
               &= 1 - \exp\left[ -\int_0^s \frac{p(I_i^0(x_i)=u,I_j^0(x_j)>u \text{ for all } j \neq i|\bL=\bl)}{\Pr(I_i^0(x_i)>u,I_j^0(x_j)>u \text{ for all } j \neq i|\bL=\bl)} du \right] \\ 
\intertext{by Assumption \ref{as:independence}}                &= 1 - \exp\left[ -\int_0^s \frac{p(I_i^0(x_i)=u,I_j^0(x_j)>u \text{ for all } j \neq i|\X=\x,\bL=\bl)}{\Pr(I_i^0(x_i)>u,I_j^0(x_j)>u \text{ for all } j \neq i|\X=\x,\bL=\bl)} du \right] \\ \intertext{Assumption \ref{as:txexchangeability}}
               &= 1 - \exp\left[ -\int_0^s \frac{p(I_i^0=u,I_j^0>u \text{ for all } j \neq i|\X=\x,\bL=\bl)}{\Pr(I_i^0>u,I_j^0>u \text{ for all } j \neq i|\X=\x,\bL=\bl)} du \right] \\ \intertext{Assumption \ref{as:consistency}} 
               &= 1 -\exp \Big{[} -\int_0^s \frac{p(T_i=u,T_{\varphi_i^1}>u,\ldots,T_{\varphi_i^{n-1}}>u|\X=\x,\bL=\bl)}{\Pr(T_i>u,T_{\varphi_i^1}>u,\ldots,T_{\varphi_i^{n-1}}>u|\X=\x,\bL=\bl)} du \Big{]}\\ \intertext{by Definition \eqref{eq:Tdecomposition}}
\end{split}
\end{equation*}

For $k=2,\ldots,n-2$, by Assumptions \ref{as:exclusion} and \ref{as:positivity}, we can define $f_{I_i^k}(s|\x,\bhi(t_{(i)}^k),\bl)$ as the density of $I_i^k \big{(} \x,\bhi(t_{(i)}^k) \big{)}$ when $\bL=\bl$. Let $F_{I_i^k}(s|\x,\bhi(t_{(i)}^k),\bl)$ be the corresponding cumulative distribution function,  and by Assumption \ref{as:positivity}, we have $0<F_{I_i^k}(s|\x,\bhi(t_{(i)}^k),\bl)<1$ for all $s>0$, $\x$, $\bhi(t_{(i)}^k)$, and $\bl$. Thus, we can write: 
\[ \frac{f_{I_i^k}(s|\x,\bhi(t_{(i)}^k),\bl)}{1-F_{I_i^k}(s|\x,\bhi(t_{(i)}^k),\bl)} = -\frac{d}{ds} \log(1-F_{I_i^k}(s|\x,\bhi(t_{(i)}^k,\bl)) \]
Then rearranging, we have:
\begin{equation}
\begin{split}
& F_{I_i^k}(s|\x,\bhi(t_{(i)}^k),\bl) = 1 - \exp\left[ -\int_0^s \frac{ \scriptstyle f_{I_i^k}(u|\x,\bhi(t_{(i)}^k),\bl)}{\scriptstyle 1-F_{I_i^k}(u|\x,\bhi(t_{(i)}^k),\bl)} du\right] \\
&=  1 - \exp\left[ -\int_0^s \frac{\scriptstyle f_{I_i^k}(t|\x,\bhi(t^k_{(i)}),\bl)\prod_{j=\varphi_i^{k+1}}^{\varphi_i^{n-1}} \left( 1-F_{I_j^k}(t|\x,\bhi(t_{(i)}^k),\bl) \right)}{\scriptstyle \left(1-F_{I_i^k}(t|\x,\bhi(t_{(i)}^k),\bl \right) \prod_{j=\varphi_i^{k+1}}^{\varphi_i^{n-1}} \left( 1-F_{I_j^k}(t|\x,\bhi(t_{(i)}^k),\bl) \right)
} du \right] \\
&=  1 - \exp\left[ -\int_0^s \frac{\scriptstyle p \left(I_i^k \big{(} \x,\bhi(t^k_{(i)})\big{)}=u,I_j^k \big{(} \x,\bhi(t_{(i)}^k)\big{)} >u \text{ for }j=\varphi_i^{k+1},\ldots,\varphi_i^{n-1}|\bL=\bl \right)}
								  {\scriptstyle \Pr \left(I_i^k \big{(}\x,\bhi(t_{(i)}^k)\big{)}>u,I_j^k \big{(}\x,\bhi(t_{(i)}^k)\big{)}
    >u \text{ for }j=\varphi_i^{k+1},\ldots,\varphi_i^{n-1}|\bL=\bl \right)} du \right] \\ 
\intertext{by Assumption \ref{as:independence}} 
&= 1 - \exp \Bigg{[} \\
& -\int_0^s \frac{\scriptstyle p \left(I_i^k(\x,\bhi(t^k_{(i)}))=u,I_j^k(\x,\bhi(t^k_{(i)})) >u \text{ for }j=\varphi^{k+1},\ldots,\varphi^{n-1}|\bHi(t^k_{(i)})=\bhi(t_{(i)}^k), \bL=\bl \right)}{ \scriptstyle\Pr \left(I_i^k(\x,\bhi(t^k_{(i)}))>u,I_j^k(\x,\bhi(t_{(i)}^k)) >u \text{ for }j=\varphi^{k+1},\ldots,\varphi^{n-1}|\bHi(t_{(i)}^k)=\bhi(t_{(i)}^k),\bL=\bl \right)} du \Bigg{]} \\
\intertext{by Assumption \ref{as:prevexchangeability}}
&= 1 - \exp \Bigg{[} -\int_0^s \frac{\scriptstyle p \left(I_i^k(\x,\bhi(t^k_{(i)}))=u,I_j^k(\x,\bhi(t_{(i)}^k)) >u \text{ for }j=\varphi^{k+1},\ldots,\varphi^{n-1}|\bHi(t_{(i)}^k)=\bhi(t_{(i)}^k),\X=\x, \bL=\bl \right)}{\scriptstyle \Pr \left(I_i^k(\x,\bhi(t^k_{(i)}))>u,I_j^k(\x,\bhi(t_{(i)}^k)) >u \text{ for }j=\varphi^{k+1},\ldots,\varphi^{n-1}|\bHi(t_{(i)}^k)=\bhi(t_{(i)}^k),\X=\x,\bL=\bl \right)} du \Bigg{]}\\
\intertext{by Assumption \ref{as:txexchangeability}}
&= 1 - \exp\left[ -\int_0^s \frac{ \scriptstyle p \left(I_i^k=u,I_j^k >u \text{ for }j=\varphi^{k+1},\ldots,\varphi^{n-1}|\bHi(t_{(i)}^k)=\bhi(t_{(i)}^k),\X=\x, \bL=\bl \right)}{\scriptstyle \Pr \left(I_i^k>u,I_j^k>u \text{ for }j=\varphi^{k+1},\ldots,\varphi^{n-1}|\bHi(t_{(i)}^k)=\bhi(t_{(i)}^k),\X=\x,\bL=\bl \right)} du \right] \\
\intertext{by Assumption \ref{as:consistency}}
&=1-\exp \Big{[} -\int_{t_{(i)}^k}^{t_{(i)}^k+s} \frac{\scriptstyle p(T_i=u,T_{\varphi_i^{k+1}}>u,\ldots,T_{\varphi_i^{n-1}}>u|T_{\varphi_i^1}=t_{(i)}^1,\ldots,T_{\varphi_i^k}=t_{(i)}^k,\X=\x,\bL=\bl)}{\scriptstyle \Pr(T_i>u,T_{\varphi_i^{k+1}}>u,\ldots,T_{\varphi_i^{n-1}}>u|T_{\varphi_i^1}=t_{(i)}^1,\ldots,T_{\varphi_i^k}=t_{(i)}^k,\X=\x,\bL=\bl)} du \Big{]} \\
\intertext{by Definition \eqref{eq:Tdecomposition} and the definition of $\bHi(t)$.}
\end{split}
\label{lemma2_main}
\end{equation}

There is no competing risk problem for $I_i^{n-1}(s|\x,\bhi(t_{(i)}^{n-1}),\bl)$. Let $F_{I_i^{n-1}}(s|\x,\bhi(t_{(i)}^{n-1}),\bl)$ be the cumulative distribution function for $I_i^{n-1}(\x,\bhi(t_{(i)}^{n-1}))$ when $\bL=\bL$, then we have:
\begin{equation*}
\begin{split}
& F_{I_i^{n-1}}(s|\x,\bhi(t_{(i)}^{n-1}),\bl) = \Pr(I_i^{n-1}(\x,\bhi(t_{(i)}^{n-1}))<s|\bL=\bl) \\
&= \Pr(I_i^{n-1}(\x,\bhi(t_{(i)}^{n-1}))<s|\bHi(t_{(i)}^{n-1})=\bhi(t_{(i)}^{n-1}),\bL=\bl) \text{ by Assumption \ref{as:prevexchangeability}} \\
&= \Pr(I_i^{n-1}(\x,\bhi(t_{(i)}^{n-1}))<s|\bHi(t_{(i)}^{n-1})=\bhi(t_{(i)}^{n-1}),\X=\x, \bL=\bl) \text{ by Assumption \ref{as:txexchangeability}} \\
&= \Pr(I_i^{n-1}<s|\bHi(t_{(i)}^{n-1})=\bhi(t_{(i)}^{n-1}),\X=\x, \bL=\bl) \text{ by Assumption \ref{as:consistency}} \\
&= \Pr(T_i<s+t_i^{(n-1)}|\X=\x,T_{\varphi_i^1}=t_{(i)}^1,\ldots,T_{\varphi_i^{n-1}}=t_{(i)}^{n-1},T_i>t_{(i)}^{n-1}, \bL=\bl) \\
\intertext{by Definition \eqref{eq:Tdecomposition} and the definition of $\bHi(t)$.}
\end{split}
\end{equation*}
\end{proof}

\begin{proof}[Proof of Theorem ~\ref{mainidentification}]
$T_i(\x,\bhi)$ is composed of $I_i^k(\x,\bhi(t_{(i)}^k)$ by Definition \eqref{eq:Tdecomposition}. Since $I_i^k(\x,\bhi(t_{(i)}^k)) \indep \bHi^*(t;\xminusi) |\bL$ and $I_i^k(\x,\bhi(t_{(i)}^k)) \indep \X |\bL$ by Assumptions \ref{as:prevexchangeability} and \ref{as:txexchangeability}, we have $T_i(\x,\bhi)\indep \bHi^*(t;\xminusi) |\bL$ and $T_i(\x,\bhi)\indep \X |\bL$. Thus, 
\begin{equation}
\E[Y_i(t;\x,\bhi)|\bL=\bl] = \E[Y_i(t;\x,\bhi)|\bHi^*=\bhi,\X=\x,\bL=\bl]
\label{eq:exchangeabilityforT}
\end{equation}
Therefore, 
\begin{equation*}
\begin{split}
& \E[Y_i(t;\x,\bhi)|\bL=\bl] = \E[Y_i(t;\x,\bhi)|\bHi^*=\bhi,\X=\x,\bL=\bl] \\ \intertext{ by Equation \eqref{eq:exchangeabilityforT}}
& = \Pr \left( T_i(\x,\bhi)<t |\bHi^*=\bhi,\X=\x,\bL=\bl \right)  \\ \intertext{ by the definition of $T_i(\x,\bhi)$}
& = \sum_{j=0}^{n-1}  \Pr \left( \scriptstyle I_i^k \ge t_{(i)}^{k+1}-t_{(i)}^{k} \text{ for k=1,\ldots,j-1 },\bHi^*=\bhi,\X=\x,\bL=\bl \right)  \\
& \quad \quad \quad \cdot \Pr \left( \scriptstyle I_i^j< t_{(i)}^{j+1}-t_{(i)}^j, I_i^k \ge t_{(i)}^{k+1}-t_{(i)}^{k} \text{ for k=1,\ldots,j-1 }|\bHi^*=\bhi,\X=\x,\bL=\bl\right) \\ \intertext{by the Law of Total Probability with extra conditioning}
& = \sum_{j=0}^{n-1}  \Pr \left( \scriptstyle I_i^j(\x,\bhi) <t-t_{(i)}^j |I_i^j< t_{(i)}^{j+1}-t_{(i)}^j, I_i^k \ge t_{(i)}^{k+1}-t_{(i)}^{k}  \text{ for k=1,\ldots,j-1 },\bHi^*=\bhi,\X=\x, \bL=\bl \right)  \\
& \quad \quad \quad \cdot \Pr \left( \scriptstyle I_i^j< t_{(i)}^{j+1}-t_{(i)}^j, I_i^k \ge t_{(i)}^{k+1}-t_{(i)}^{k} \text{ for k=1,\ldots,j-1 }|\bHi^*=\bhi,\X=\x,\bL=\bl\right) \\ \intertext{ by Definition \eqref{eq:Tdecomposition}} 
& = \sum_{j=0}^{n-1}  \Pr \Big{(} \scriptstyle I_i^j(\x,\bhi) <t-t_{(i)}^j |I_i^j< t_{(i)}^{j+1}-t_{(i)}^j, I_i^k \ge t_{(i)}^{k+1}-t_{(i)}^{k} \text{ for k=1,\ldots,j-1 },\bHi^*=\bhi,\X=\x, \bL=\bl \Big{)} \\
\end{split}
\end{equation*}
\begin{equation*}
\begin{split}
& \quad \quad \quad \cdot \Pr \left( \scriptstyle I_i^j< t_{(i)}^{j+1}-t_{(i)}^j|I_i^k \ge t_{(i)}^{k+1}-t_{(i)}^{k}  \text{ for k=1,\ldots,j-1 }|\bHi^*=\bhi,\X=\x,\bL=\bl\right)  \\
& \quad \quad \quad \cdot  \prod_{k=1}^{j-1}\Pr \left( \scriptstyle I_i^k \ge t_{(i)}^{k+1}-t_{(i)}^{k} |I_i^m \ge t_{(i)}^{m+1}-t_{(i)}^{m} \text{ for all $m=1,\ldots,k-1$}|\bHi^*=\bhi,\X=\x,\bL=\bl \right) \\ \intertext{by the Chain Rule}
& = \sum_{j=0}^{n-1}  \Pr \Big{(} \scriptstyle I_i^j(\x,\bhi(t_{(i)}^j)) <t-t_{(i)}^j |I_i^j< t_{(i)}^{j+1}-t_{(i)}^j, I_i^k\ge t_i^{(k+1)}-t_i^{(k)} \text{ for k=1,\ldots,j-1 } , \bHi(t_{(i)}^j)=\bhi(t_{(i)}^j),\X=\x, \bL=\bl \Big{)}  \\
& \quad \quad \quad \cdot \Pr \left( \scriptstyle I_i^j< t_{(i)}^{j+1}-t_{(i)}^j|I_i^k\ge t_{(i)}^{k+1}-t_{(i)}^{k}  \text{ for k=1,\ldots,j-1 }|\bHi(t_{(i)}^j)=\bhi(t_{(i)}^j),\X=\x,\bL=\bl \right)  \\
& \quad \quad \quad \cdot  \prod_{k=1}^{j-1}\Pr \left( \scriptstyle I_i^k \ge t_{(i)}^{k+1}-t_{(i)}^{k}|I_i^m \ge t_{(i)}^{m+1}-t_{(i)}^{m} \text{ for $m=1,\ldots,k-1$}|\bHi(t_{(i)}^k)=\bhi(t_{(i)}^k),\X=\x,\bL=\bl \right) \\
\intertext{by Assumption \ref{as:exclusion} and the equivalence of $\bHi(t_{(i)}^j)$ to $\bHi^*(t_{(i)}^j)$ in distribution given $T_i>t_{(i)}^j$}
& = \sum_{j=0}^{n-1}  \Pr \left( I_i^j(\x,\bhi(t_{(i)}^j)) <t-t_{(i)}^j |\bHi(t_{(i)}^j)=\bhi(t_{(i)}^j),\X=\x, \bL=\bl \right) \\ 
& \quad \quad \quad \cdot \Pr \left( I_i^j< t_{(i)}^{j+1}-t_{(i)}^j|\bHi(t_{(i)}^j)=\bhi(t_{(i)}^j),\X=\x,\bL=\bl\right) \\
& \quad \quad \quad \cdot  \prod_{k=1}^{j-1}\Pr \left(I_i^k \ge t_{(i)}^{k+1}-t_{(i)}^{k}  |\bHi(t_{(i)}^k)=\bhi(t_{(i)}^k),\X=\x,\bL=\bl \right) \\
\intertext{ by Assumption \ref{as:independence}}
& = \mathlarger{\sum}_{j=0}^{n-1}  \Pr \left(I_i^j <t-t_{(i)}^j |\bHi(t_{(i)}^j)=\bhi(t_{(i)}^j),\X=\x, \bL=\bl \right) \\ 
& \quad \quad \quad \cdot \Pr \left( I_i^j< t_{(i)}^{j+1}-t_{(i)}^j|\bHi(t_{(i)}^j)=\bhi(t_{(i)}^j),\X=\x,\bL=\bl\right)  \\
& \quad \quad \quad \cdot  \prod_{k=1}^{j-1}\Pr \left(I_i^k \ge t_{(i)}^{k+1}-t_{(i)}^{k}|\bHi(t_{(i)}^k)=\bhi(t_{(i)}^k),\X=\x,\bL=\bl \right) \\ 
\intertext{ by Assumption \ref{as:consistency}}
& = \mathlarger{\sum}_{j=0}^{n-1}  \Pr \left(I_i^j < \min \{ t-t_{(i)}^j, t_{(i)}^{j+1}-t_{(i)}^j \} |\bHi(t_{(i)}^j)=\bhi(t_{(i)}^j),\X=\x, \bL=\bl \right) \\ 
& \quad \quad \quad \cdot  \prod_{k=1}^{j-1}\Pr \left(I_i^k \ge t_{(i)}^{k+1}-t_{(i)}^{k}|\bHi(t_{(i)}^k)=\bhi(t_{(i)}^k),\X=\x,\bL=\bl \right) \\ 
& = \mathlarger{\mathlarger{\sum}}_{j=0}^{n-1} \left[ F_{I_i^j} \left( \min \left\{ t,t_{(i)}^{j+1} \right\} -t_{(i)}^{j}|\x,\bhi(t_{(i)}^j),\bl \right) \prod_{k=0}^{j-1} \left(1- F_{I_i^k}(t_{(i)}^{k+1}-t_{(i)}^{k}|x,\bhi(t_{(i)}^k),\bl) \right) \right] 
\end{split}
\end{equation*}
	
\end{proof}

\begin{proof}[Proof of Corollary \ref{cor:twoperson}] 
When $n=2$, Theorem \ref{mainidentification} for individual $i=1$ becomes,
\begin{equation*}
\begin{split}
& \E[Y_1(t;\bh_{(1)},\x)|\bL=\bl] = \E[Y_1(t;t_2,\x)|\bL=\bl] \\
&= F_{I_i^0}( \min \{ t,t_2 \}\mid \x,t_2,\bl )+  F_{I_i^1}( t-t_2\mid \x,t_2,\bl ) \left(1- F_{I_i^0}(t_2 \mid \x,t_2,\bl ) \right) \text{by Lemma \ref{Iidentification}} \\
&= F_{I_i^0}( \min \{ t,t_2 \}\mid x_1,\bl_1 )+  F_{I_i^1}( t-t_2\mid \x,t_2,\bl ) \left(1- F_{I_i^0}(t_2 \mid x_1,\bl_1 ) \right) \text{by Assumption \ref{as:exclusion}} \\
&= F_{I_i^0}( \min \{ t,t_2 \}\mid x_1,\bl_1 )+  \Pr[T_1 <t|T_2=t_2,T_1>t_2,\X=\x,\bL=\bl] \left(1- F_{I_i^0}(t_2 \mid x_1,\bl_1 ) \right) \\
\intertext{by Lemma \ref{Iidentification}}
&= F_1( \min \{ t,t_2 \} | x_1,\bl_1) + \Pr[T_1 <t|T_2=t_2,T_1>t_2,\X=\x,\bL=\bl] \left( 1- F_1(t_2|x_1,\bl_1)  \right) \\
\end{split}
\end{equation*}
Therefore, when $t > t_2$,
\[ \E[Y_1(t;\bh_{(1)},\x)|\bL=\bl] = F_1(t_2| x_1,\bl_1) + \Pr[T_1 <t|T_2=t_2,T_1>t_2,\X=\x,\bL=\bl] \big{(} 1- F_1(t_2|x_1,\bl_1) \big{)} \]
and when $t \le t_2$,
\[ \E[Y_i(t;\bhi,\x)|\bL=\bl] =F_i(t|x_i,\bl_i) \]
\end{proof}


\begin{proof}[Proof of Theorem \ref{thm:exposure}] 
\begin{equation*}
\begin{split}
& dG^*_{(i)}(\bhi|\xminusi,\bl_{(i)}) = p(\bHi^*(\xminusi)=\bhi|\bLi=\bli)\\
& \quad = p(T_{\varphi_i^1}(\bh^i_{(\varphi_i^1)},\x)=t^1_{(i)},\ldots,T_{\varphi_i^{n-1}}(\bh^i_{(\varphi_i^{n-1})},\x)=t^{n-1}_{(i)}|\bL=\bl) \\
\intertext{by the definitions of $\bHi^*(\xminusi)$ and $\bh^i_{(j)}$, and Assumption \ref{as:exclusion}}
& \quad = p(I^0_{\varphi_i^1}(\bh^i_{(\varphi_i^1)},\x)=t^1_{(i)}, I^0_{\varphi_i^2}(\bh^i_{(\varphi_i^1)},\x)>t^1_{(i)},\ldots, I^0_{\varphi_i^{n-1}}(\bh^i_{(\varphi_i^1)},\x)>t^1_{(i)},\\
& \quad \quad \quad I^1_{\varphi_i^2}(\bh^i_{(\varphi_i^2)},\x)=t^2_{(i)}-t^1_{(i)}, I^1_{\varphi_i^3}(\bh^i_{(\varphi_i^3)},\x)> t^2_{(i)}-t^1_{(i)}, \ldots, I^1_{\varphi_i^{n-1}}(\bh^i_{(\varphi_i^{n-1})},\x)> t^2_{(i)}-t^1_{(i)}, \\
& \quad \quad \quad , \ldots, \\
& \quad \quad \quad I^{n-2}_{\varphi_i^{n-1}}(\bh^i_{(\varphi_i^{n-1})},\x)=t^{n-1}_{(i)}-t^{n-2}_{(i)} \mid \bLi=\bli)  \\
\intertext{by Definition \eqref{eq:Tdecomposition}}
& \quad =\prod_{j=1}^{n-1} \left[ f_{I_{\varphi_i^j}^{j-1}}\big{(} t_{(i)}^j-t_{(i)}^{j-1} \mid \x,\bh^i_{(\varphi_i^j)},\bl \big{)} \prod_{k=j+1}^{n-1} S_{I_{\varphi_i^k}^{j-1}}\big{(} t_{(i)}^j-t_{(i)}^{j-1}\mid \x,\bh^j_{(\varphi_i^k)},\bl \big{)}\right] \\
\intertext{by Assumption \ref{as:independence}}
\end{split}
\end{equation*}
\end{proof}

\begin{proof}[Proof of Corollary \ref{cor:obs}]
\begin{equation*}
\begin{split}
\E \big[Y_i\big{(}t;x_i,\xminusi,\bHi^*(\xminusi)\big{)} \,\big|\, \bL=\bl\big] & =  \int \E[Y_i(t;\x,\bhi)|\bL=\bl] \text{ d}G^*_{(i)}(\bhi|\x,\bl)   \\ \intertext{by Definition \ref{eq:yitx}}
& = \int \E[Y_i(t;\x,\bhi)|\bHi^*=\bhi,\X=\x,\bL=\bl] \text{ d}G^*_{(i)}(\bhi|\x,\bl)   \\ 
\intertext{by Equation \eqref{eq:exchangeabilityforT} under Assumptions \ref{as:prevexchangeability} -- \ref{as:txexchangeability}}
& = \int \E[Y_i(t)\,|\, \bHi^*=\bhi, \X=\x, \bL=\bl] \text{ d}G^*_{(i)}(\bhi|\x,\bl)  \\ 
\intertext{by Assumption \ref{as:consistency}}
& = \E[Y_i(t) | \X=\x, \bL=\bl]  \intertext{as $G^*_{(i)}(\bhi|\x,\bl)$ does not depend on $Y_i(t)$} 
\end{split}	
\end{equation*}

\end{proof}

\begin{proof}[Proof of Proposition \ref{prop:overexposuretx}]
\begin{align*}
    \overline{Y}_i^p(t; x_i \,|\, \bL = \bl) 
    &= \sum_{x_{(i)}\in\mathcal{X}^{n-1}} \E\big[Y_i\big(t;x_i,\x_{(i)},\bHi^*(\x_{(i)})\big)\mid \bL=\bl\big] p(\x_{(i)}\,|\, X_i=x_i,\bL = \bl) \\
\intertext{by Assumption \ref{as:sufficientL}}
    & = \sum_{x_{(i)}\in\mathcal{X}^{n-1}} \E\big[Y_i\big(t;x_i,\x_{(i)},\bHi^*(\x_{(i)})\big)\mid \bL=\bl\big] p(\x_{(i)}\,|\,\bL = \bl)  \\
    \overline{Y}_i^p(t; x_i, \x_{(i)}' \,|\, \bL = \bl) 
    &= \sum_{x_{(i)}\in\mathcal{X}^{n-1}} \E\big[Y_i\big(t;x_i, \x_{(i)}, \bHi^*(\x_{(i)}')\big)\mid \bL=\bl\big] p(\x_{(i)}\,|\, X_i=x_i,\bL = \bl)\\
    \intertext{by Assumption \ref{as:sufficientL}}
    &= \sum_{x_{(i)}\in\mathcal{X}^{n-1}} \E\big[Y_i\big(t;x_i, \x_{(i)}, \bHi^*(\x_{(i)}')\big)\mid \bL=\bl\big] p(\x_{(i)}\,|\,\bL = \bl)\\    
  \end{align*}
\end{proof}

\begin{proof}[Proof of Theorem \ref{thm:haz_mainresult}]

\begin{equation*}
\begin{split}
& \E[Y_i(t;\bhi,\x)|\bL=\bl] \\
&  = \mathlarger{\mathlarger{\sum}}_{j=0}^{n-1} \left[ F_{I_i^j} \left( \min \left\{ t,t_{(i)}^{j+1} \right\} -t_{(i)}^{j}|\x,\bhi,\bl \right) \prod_{k=0}^{j-1} \left(1- F_{I_i^k}(t_{(i)}^{k+1}-t_{(i)}^{k})|\x,\bhi,\bl \right) \right] 
 \\ \intertext{by Theorem \ref{thm:exposure}}
& = \Pr \big( I_i^0(x_i) \le \min \{t,t_{(i)}^{1} \}|\bL=\bl \big) \\
& \quad + \Pr \big( I_i^1(x_i,\xminusi,\bhi(t_{(i)}^1)) \le \min \{t,t_{(i)}^{2}\}-t_{(i)}^{1} |\bL=\bl \big) \cdot \Pr \big( I_i^0(x_i) > t_{(i)}^{1} |\bL=\bl \big)  \\
& \quad + \Pr \big( I_i^2(x_i,\xminusi,\bhi(t_{(i)}^2)) \le \min \{t,t_{(i)}^{3}\}-t_{(i)}^{2} |\bL=\bl \big) \cdot \Pr \big( I_i^0(x_i) > t_i^{(1)} |\bL=\bl \big) \\
& \quad \quad \quad \cdot \Pr \big( I_i^1(x_i,\xminusi,\bhi(t_{(i)}^1)) >t_{(i)}^{2}-t_{(i)}^{1} |\bL=\bl \big) \\
& \quad + \ldots \\
& \quad + \Pr \big( I_i^{n-1} (x_i,\xminusi,\bhi(t_{(i)}^{n-1})) \le t- t_{(i)}^{n-1} |\bL=\bl \big) \\
& \quad \quad \quad \cdot \prod_{k=0}^{n-2} \Pr \big( I_i^k(x_i,\xminusi,\bhi(t^k_{(i)})) > t_i^{k+1}-t_i^{k} |\bL=\bl \big) \\
& = 1-e^{-\int_0^{\min \{t,t_{(i)}^{1}\}} \lambda_{0i}(s|x_i,\bl_i)ds} \\
& \quad + \Big( 1- e^{-\int_{t_{(i)}^{1}}^{\min \{t,t_{(i)}^{2}\}} \lambda_{0i}(s|x_i,\bl_i)+\lambda_{\varphi_i^1 i }(s-t_{(i)}^{1}|x_i,x_{(i)}^{1},\bl_i,\bl_{\varphi_i^1 }) ds} \Big) \cdot e^{-\int_0^{t_{(i)}^{1}} \lambda_{i0}(s|x_i,\bl_i)ds} \\ 
& \quad + \Big( 1- e^{-\int_{t_{(i)}^{2}}^{\min \{t,t_{(i)}^{3}\}} \lambda_{i0}(s|x_i,\bl_i)+\lambda_{\varphi_i^1 i}(s-t_{(i)}^{1}|x_i,x_{(i)}^{1},\bl_i,\bl_{\varphi_i^1})+\lambda_{\varphi_i^2 i}(s-t^2_{(i)}|x_i,x_{(i)}^{2},\bl_i,\bl_{\varphi_i^2}) ds} \Big) \\
& \quad \quad \quad \quad  \cdot e^{-\int_0^{t_{(i)}^{1}} \lambda_{i0}(s|x_i,\bl_i)ds} \cdot e^{-\int_{t_{(i)}^{1}}^{t_{(i)}^{2}} \lambda_{0i}(s|x_i,\bl_i)+\lambda_{\varphi_i^1 i }(s-t_{(i)}^{1}|x_i,x_{(i)}^{1},\bl_i,\bl_{\varphi_i^1 }) ds} \\ 
& \quad + \ldots \\
& \quad + \Big( 1- e^{-\int_{t_{(i)}^{n-1}}^t \lambda_{i0}(s|x_i,\bl_i)+\lambda_{\varphi_i^1 i}(s-t_{(i)}^1|x_i,x_{(i)}^{1},\bl_i,\bl_{\varphi_i^1})+ \ldots +\lambda_{\varphi_{i}^{n-1} i}(s-t{(i)}^{n-1}|x_i,x_{(i)}^{n-1},\bl_i,\bl_{\varphi_i^{n-1}}) ds} \Big) \\
& \quad \quad \quad \cdot e^{-\int_0^{t_{(i)}^{1}} \lambda_{i0}(s|x_i,\bl_i)ds} \cdot e^{-\int_{t_{(i)}^{1}}^{t_{(i)}^{2}} \lambda_{0i}(s|x_i,\bl_i)+\lambda_{\varphi_i^1 i }(s-t_{(i)}^{1}|x_i,x_{(i)}^{1},\bl_i,\bl_{\varphi_i^1 }) ds}\\
& \quad \quad \quad \cdot \ldots \cdot e^{-\int_{t_{(i)}^{n-2}}^{t_{(i)}^{n-1}} \lambda_{i0}(s|x_i,\bl_i)+\lambda_{\varphi_i^1 i}(s-t_{(i)}^1|x_i,x_{(i)}^{1},\bl_i,\bl_{\varphi_i^1}) + \ldots + \lambda_{\varphi_i^{n-2}i}(s-t_{(i)}^{n-2}|x_i,x_{(i)}^{(n-2)},\bl_i,\bl_{\varphi_i^{n-1}}) ds} \\ \intertext{by plugging in Equation \eqref{eq:sumhazard_forI}}
& = 1 - \exp ^{- \int_0^t  \lambda_{0i}(s|x_i,\bl_i) \text{d}s- \sum_{j=1}^{n-1} \int_{t_{(i)}^j}^{t} \lambda_{\varphi_i^ji}(s-t_{(i)}^j|x_i,x_{\varphi_i^j},\bl_i,\bl_{\varphi_i^j})  \text{d}s}\\
& = 1 - \exp ^{ - \int_0^t \Big{(} \lambda_{0i}(s|x_i,\bl_i) + \sum_{j \neq i} y_j(s) \lambda_{ji}(s-t_j|x_i,x_j,\bl_i,\bl_j) \Big{)} \text{d}s} \\ 
\end{split}
\end{equation*}
\end{proof}

\begin{proof}[Proof of Theorem \ref{thm:haz_Hresult}]
\begin{equation*}
\begin{split}
& \text{d}G_{(i)}^*(\bhi \,|\, \x_{(i)}, \bli) = \prod_{j=1}^{n-1} \bigg[f_{I_{\varphi_i^j}^{j-1}}\big( t_{(i)}^j-t_{(i)}^{j-1} \,\big|\, \x,\bh^i_{(\varphi_i^j)},\bl \big) \prod_{k=j+1}^{n-1} S_{I_{\varphi_i^k}^{j-1}}\big( t_{(i)}^j-t_{(i)}^{j-1} \,\big|\, \x,\bh^j_{(\varphi_i^k)},\bl\big)\bigg] \\
\intertext{by Theorem \ref{thm:exposure}}
& \quad = f_{I_{\varphi_i^1}^{0}}  \big( t_{(i)}^1 \,\big|\, \x,\bh^i_{(\varphi_i^1)},\bl \big)  \prod_{k=2}^{n-1}  S_{I_{\varphi_i^k}^{0}}\big( t_{(i)}^1 \,\big|\, \x,\bh^j_{(\varphi_i^k)},\bl\big)  \\
& \quad \quad \cdot f_{I_{\varphi_i^2}^{1}}  \big( t_{(i)}^2-t_{(i)}^1 \,\big|\, \x,\bh^i_{(\varphi_i^2)},\bl \big)  \prod_{k=3}^{n-1}  S_{I_{\varphi_i^k}^{0}}\big( t_{(i)}^2-t_{(i)}^1 \,\big|\, \x,\bh^j_{(\varphi_i^k)},\bl\big)  \\
& \quad \quad \cdot \ldots \cdot \\
& \quad \quad \cdot f_{I_{\varphi_i^{n-1}}^{n-2}}  \big( t_{(i)}^{n-1}-t_{(i)}^{n-2} \,\big|\, \x,\bh^i_{(\varphi_i^2)},\bl \big)  \\
& \quad = \lambda_{\varphi_i^1}(t_{(i)}^1|\x,\bh^i_{(\varphi_i^1)},\bl) \exp \left(-\int_0^{t_{(i)}^1}  \lambda_{\varphi_i^1}(s|\x,\bh^i_{(\varphi_i^1)},\bl) \text{d}s\right) \cdot  \prod_{k=2}^{n-1}  \exp \left( -\int_0^{t_{(i)}^1}  \lambda_{\varphi_i^k}(s|\x,\bh^i_{(\varphi_i^k)},\bl) \text{d}s\right) \\
& \quad \quad \cdot \lambda_{\varphi_i^2}(t_{(i)}^2|\x,\bh^i_{(\varphi_i^2)},\bl) \exp \left(-\int_{t_{(i)}^1}^{t_{(i)}^2}  \lambda_{\varphi_i^2}(s|\x,\bh^i_{(\varphi_i^2)},\bl) \text{d}s\right) \cdot  \prod_{k=3}^{n-1}  \exp \left( -\int_{t_{(i)}^1}^{t_{(i)}^2} \lambda_{\varphi_i^k}(s|\x,\bh^i_{(\varphi_i^k)},\bl) \text{d}s\right) \\
& \quad \quad \cdot \ldots \cdot \\
& \quad \quad \cdot \lambda_{\varphi_i^{n-1}}(t_{(i)}^{n-1}|\x,\bh^i_{(\varphi_i^1)},\bl) \exp \left(-\int_{t_{(i)}^{n-2}}^{t_{(i)}^{n-1}}  \lambda_{\varphi_i^{n-1}}(s|\x,\bh^i_{(\varphi_i^{n-1})},\bl) \text{d}s\right)  \\
& \quad = \prod_{j=1}^{n-1} \bigg[\lambda_{\varphi_i^j}(t_{(i)}^j \,|\, \x, \bh^i_{(\varphi_i^j)},  \bl) \prod_{k = j}^{n-1} \exp\bigg(-\int_{t_{(i)}^{j-1}}^{t_{(i)}^j} \lambda_{\varphi_i^k}(s \,|\, \x, \bh^i_{(\varphi_i^k)}, \bl) \,\mathrm{d}s\bigg)\bigg] \\
& \quad = \prod_{j=1}^{n-1} \lambda_{\varphi_i^j}(t_{(i)}^j \,|\, \x, \bh^i_{(\varphi_i^j)},  \bl)  \times \prod_{k = 1}^{n-1} \exp\bigg(-\int_0^{t_{(i)}^k} \lambda_{\varphi_i^k}(s \,|\, \x, \bh^i_{(\varphi_i^k)}, \bl) \,\mathrm{d}s\bigg) \\
\intertext{by grouping terms by $\varphi_i^k$}
\end{split}
\end{equation*}
\end{proof}


\begin{cor}
Under hazard model \eqref{eq:haz}, the exposure-controlled contagion effect $\CE_i(t,\x,\bhi,\bhi')>0$ for $\bhi \prec \bhi'$ as long as $\gamma(t)>0$ for any $t>0$.
\label{cor:contagionsign}	
\end{cor}

\begin{proof}[Proof of Corollary \ref{cor:contagionsign}]
\begin{equation}
\begin{split}
\CE_i(t,\x,\bhi,\bhi')  & = \E[Y_i(t;\x, \bhi) - Y_i(t;\x, \bhi')] \\
& = 1 - \exp \big[- \int_0^t \Big{(} \lambda_{0i}(s|x_i,\bl_i) + \sum_{j \neq i} y_j(s) \lambda_{ji}(s-t_j|x_i,x_j,\bl_i,\bl_j) \Big{)} \mathrm{d}s \big] \\ 
& \quad - \Big( 1 - \exp \big[- \int_0^t \Big{(} \lambda_{0i}(s|x_i,\bl_i) + \sum_{j \neq i} y_j(s) \lambda_{ji}(s-t'_j|x_i,x_j,\bl_i,\bl_j) \Big{)} \mathrm{d}s \big] \Big) \\
& =  \exp \big[- \int_0^t \Big{(} \lambda_{0i}(s|x_i,\bl_i) \mathrm{d}s \big]\\
&  \quad \times  \Big( e^{ - \sum_{j \neq i} \int_0^t y_j(s) \lambda_{ji}(s-t'_j|x_i,x_j,\bl_i,\bl_j) \mathrm{d}s } - e^{ - \sum_{j \neq i} \int_0^t y_j(s) \lambda_{ji}(s-t_j|x_i,x_j,\bl_i,\bl_j) \mathrm{d}s } \Big)
\end{split}
\label{proof:contagion_main}	 
\end{equation}
For the term inside the round brackets, plugging in model \eqref{eq:haz}, we have
\begin{equation}
\begin{split}
& \exp \Big[ - \mathlarger{\sum}_{j \neq i} \int_0^t y_j(s) \lambda_{ji}(s-t'_j|x_i,x_j,\bl_i,\bl_j) \mathrm{d}s \Big] \\
& =\exp \Big[  - \mathlarger{\sum}_{j \neq i} \int_0^t y_j(s)   \gamma(s-t'_j) \exp(\beta_1 x_i + \beta_2 x_j + \theta_1 \bl_i + \theta_2 \bl_j) \mathrm{d}s \Big] \text{ by \eqref{eq:haz}} \\
& =\exp \Big[ - \mathlarger{\sum}_{j \neq i}     \exp(\beta_1 x_i + \beta_2 x_j + \theta_1 \bl_i + \theta_2 \bl_j)  \cdot \int_0^{t-t'_j} \gamma(s) \mathrm{d}s \Big]  \\
& > \exp \Big[ - \mathlarger{\sum}_{j \neq i}     \exp(\beta_1 x_i + \beta_2 x_j + \theta_1 \bl_i + \theta_2 \bl_j)   \cdot \int_0^{t-t_j} \gamma(s) \mathrm{d}s \Big] \text{ since $t_j < t'_j$ and $\gamma(t)>0$} \\
& = \exp \Big[ - \mathlarger{\sum}_{j \neq i} \int_0^t y_j(s) \lambda_{ji}(s-t_j|x_i,x_j,\bl_i,\bl_j) \mathrm{d}s \Big] \\
\end{split}
\label{proof:contagion_1}	 
\end{equation}
Plugging the result of \eqref{proof:contagion_1} in \eqref{proof:contagion_main}, we have $\CE_i(t,\x,\bhi,\bhi')>0$.

Similar conclusion applies for the exposure-marginalized contagion effect, and the proof is omitted here.
\end{proof}

\begin{cor}
Under hazard model \eqref{eq:haz}, the exposure-controlled susceptibility effect $\SE_i(t,\x_{(i)}, \bhi)$ has the same sign as $e^{\beta_1}$.
\label{cor:susceptibilitysign}	
\end{cor}
\begin{proof}[Proof of Corollary \ref{cor:susceptibilitysign}]

\begin{equation*}
\begin{split}
& \SE_i(t,\x_{(i)}, \bhi) = \E[Y_i(t;1,\x_{(i)},\bhi) - Y_i(t;0,\x_{(i)},\bhi)]	 \\
& = 1 - \exp \big[- \int_0^t \Big{(} \lambda_{0i}(s|1,\bl_i) + \sum_{j \neq i} y_j(s) \lambda_{ji}(s-t_j|1,x_j,\bl_i,\bl_j) \Big{)} \mathrm{d}s \big] \\ 
& \quad - \Big( 1 - \exp \big[- \int_0^t \Big{(} \lambda_{0i}(s|0,\bl_i) + \sum_{j \neq i} y_j(s) \lambda_{ji}(s-t'_j|0,x_j,\bl_i,\bl_j) \Big{)} \mathrm{d}s \big] \Big) \\
& =  \exp \big[- \int_0^t \alpha(t) \exp( \theta_1 \bl_i) + \sum_{j \neq i} y_j(s)   \gamma(t-t_j) \exp(\beta_2 x_j + \theta_1 \bl_i + \theta_2 \bl_j) \mathrm{d}s \big]\\
&  \quad - \exp \big[- e^{\beta_1} \int_0^t \alpha(t) \exp(\theta_1 \bl_i) + \sum_{j \neq i} y_j(s)   \gamma(t-t_j) \exp(\beta_2 x_j + \theta_1 \bl_i + \theta_2 \bl_j) \mathrm{d}s \big]\\
\end{split}	
\end{equation*}
Therefore, if $\beta_1<0$ thereby $e^{\beta_1}<1$, then $\SE_i(t,\x_{(i)}, \bhi)<0$; if $\beta_1=0$ thereby $e^{\beta_1}=1$, then $\SE_i(t,\x_{(i)}, \bhi)=0$; if $\beta_1>0$ thereby $e^{\beta_1}>1$, then $\SE_i(t,\x_{(i)}, \bhi)>0$.

Similar conclusions apply for the exposure-marginalized susceptibility effect and the exposure-and-treatment-marginalized susceptibility effect, and the proofs are omitted here.
\end{proof}

\begin{cor}
Under hazard model \eqref{eq:haz}, the exposure-controlled infectiousness effect $\IE_{i}(t,x_i,\bhi)$ has the same sign as $e^{\beta_2}$.
\label{cor:infectiousnesssign}	
\end{cor}
\begin{proof}[Proof of Corollary \ref{cor:infectiousnesssign}]
	
\begin{equation*}
\begin{split}
&  \IE_{i}(t,x_i,\bhi) = \E[Y_i(t;x_i,\mathbf{1},\bhi) - Y_i(t;x_i,\mathbf{0}, \bhi)] \\
& = 1 - \exp \big[- \int_0^t \Big{(} \lambda_{0i}(s|x_i,\bl_i) + \sum_{j \neq i} y_j(s) \lambda_{ji}(s-t_j|x_i,1,\bl_i,\bl_j) \Big{)} \mathrm{d}s \big] \\ 
& \quad - \Big( 1 - \exp \big[- \int_0^t \Big{(} \lambda_{0i}(s|x_i,\bl_i) + \sum_{j \neq i} y_j(s) \lambda_{ji}(s-t'_j|x_i,0,\bl_i,\bl_j) \Big{)} \mathrm{d}s \big] \Big) \\
& =  \exp \big[- \int_0^t \alpha(t) \exp(\beta_1 x_i+ \theta_1 \bl_i) + \sum_{j \neq i} y_j(s)   \gamma(t-t_j) \exp(\beta_1 x_i + \theta_1 \bl_i + \theta_2 \bl_j) \mathrm{d}s \big]\\
&  \quad - \exp \big[-  \int_0^t \alpha(t) \exp(\beta_1 x_i+\theta_1 \bl_i) + e^{\beta_2} \sum_{j \neq i} y_j(s)   \gamma(t-t_j) \exp(\beta_1 x_i + \theta_1 \bl_i + \theta_2 \bl_j) \mathrm{d}s \big]\\
\end{split}	
\end{equation*}
Therefore, if $\beta_2<0$ thereby $e^{\beta_2}<1$, then $\IE_{i}(t,x_i,\bhi)<0$; if $\beta_2=0$ thereby $e^{\beta_2}=1$, then $\IE_{i}(t,x_i,\bhi)=0$; if $\beta_2>0$ thereby $e^{\beta_2}>1$, then $\IE_{i}(t,x_i,\bhi)>0$.

Similar conclusions apply for the exposure-marginalized infectiousness effect and the exposure-and-treatment-marginalized infectiousness effect, and the proofs are omitted here.
\end{proof}

\begin{proof}[Proof of Equations \eqref{eq:controlledSE}--\eqref{eq:contagion2}]
First, we prove the result in Equation \eqref{eq:controlledSE}.
\begin{equation*}
\begin{split}
HSE^C(t,\x_{(i)},\bhi,\bl) & = \frac{\lambda_{i}(t|x_i=1,\x_{(i)},\bhi,\bl)}{\lambda_{i}(t|x_i=0,\x_{(i)},\bhi,\bl)}\\
&= \frac{ e^{\beta_1 +\theta_1 \bl_i} \left( \alpha(t) + \sum_{j  \neq i}   y_j(t) \gamma(t-t_j) e^{\beta_2 X_j + \theta_2 \bl_j} \right)  }{e^{\theta_1 \bl_i} \left( \alpha(t) + \sum_{j \neq i}   y_j(t) \gamma(t-t_j) e^{\beta_2 X_j + \theta_2 \bl_j} \right)} \\
& =e^{\beta_1} \\
\end{split}
\end{equation*}

Second, we prove the result in Equation (\ref{eq:ie2}).
For $\bh_{(i,j)}$ with $\bt_{(i,j)}$, then 
\begin{equation*}
\begin{split}
& \lambda_{i}(t|x_{j}=1, \x_{(j)},h'_j, \bh_{(i,j)},\bl) - \lambda_{i}(t|x_{j}=1, \x_{(j)},h_j, \bh_{(i,j)},\bl) \\ 
& \quad =  \lambda_{i}(t|x_{j}=1, \x_{(j)},h'_j, \bh_{(i,j)},\bl) - \lambda_{i}(t|x_{j}=1, \x_{(j)},h_j, \bh_{(i,j)},\bl)  \\
& \quad= e^{\beta_1 X_i +\theta_1 \bl_i} \Big( \alpha(t) + \sum_{k \neq i,j}  y_k(t) \gamma(t-t_k) e^{\beta_2 X_k + \theta_2\bl_k} +  y_j(t)\gamma(t-t'_j)e^{\beta_2 + \theta_2\bl_j} \Big)   \\
& \quad \quad- e^{\beta_1 X_i +\theta_1 \bl_i} \Big( \alpha(t) + \sum_{k \neq i,j}  y_k(t) \gamma(t-t_k) e^{\beta_2 X_k + \theta_2\bl_k} \Big) \\
& \quad =  e^{\beta_1 X_i +\theta_1 \bl_i}  y_j(t)\gamma(t-t'_j)e^{\beta_2 + \theta_2\bl_j} \\
\intertext{Similarly,}
& \lambda_{i}(t|x_{j}=0,\x_{(j)}, h'_j, \bh_{(i,j)},\bl)- \lambda_{i}(t|x_{j}=0, \x_{(j)},h_j, \bh_{(i,j)},\bl) = e^{\beta_1 X_i +\theta_1 \bl_i} y_j(t) \gamma(t-t'_j)e^{\theta_2\bl_j}  
\end{split}
\end{equation*}
Thus,
\[
HIE^C(t;h_j,h'_j,\x_{(j)}, \bh_{(i,j)}) = \frac{ e^{\beta_1 X_i +\theta_1 \bl_i} \gamma(t-t'_j)e^{\beta_2 + \theta_2\bl_j}}{ e^{\beta_1 X_i +\theta_1 \bl_i} \gamma(t-t'_j)e^{\theta_2\bl_j}} = e^{\beta_2} 
\]

Third, we prove the result in Equation \eqref{eq:contagion2}.
\begin{equation*}
\begin{split}
& \int_0^t \big{[} \lambda_i(u| \mathbf{x}=\mathbf{0},h'_j,\bh_{(i,j)},\bl)-\lambda_i(u; \mathbf{x}=\mathbf{0},h_j,\bh_{(i,j)},\bl) \big{]} du \\
& = \int_0^t \Big{\{} e^{\theta_1 \bl_i} \big{[}\alpha(u) + \sum_{k \neq i,j}  y_k(u) \gamma(u-t_k) e^{\rho_k \bl_k}+  y_j(t)  \gamma(u-t'_j) e^{\theta_2\bl_j}] \\
& \quad \quad \quad \quad - e^{\theta_1 \bl_i} \big{[}\alpha(u) + \sum_{k \neq i,j}  y_k(u) \gamma(u-t_k) e^{\rho_k \bl_k}]  \Big{\}} du  \\
&=  \int_0^t e^{\theta_1 \bl_i}   y_j(u)\gamma(u-t'_j) e^{\theta_2\bl_j} du =e^{\theta_1 \bl_i+\theta_2\bl_j}  \cdot \int_0^t   y_j(u) \gamma(u-t'_j)du  \\
\intertext{Similarly,}
& \int_0^t \big{[} \lambda_i(u; \mathbf{x}=\mathbf{0},h''_j,\bh_{(i,j)})-\lambda_i(u; \mathbf{x}=\mathbf{0},h_j,\bh_{(i,j)}) \big{]} du =  e^{\theta_1 \bl_i+\theta_2\bl_j}  \int_0^t   y_j(u) \gamma(u-t''_j) du 
\end{split}
\end{equation*}
Thus,
\[
HCE^C(t;h^{''}_j(t),h'_j(t),\mathbf{h}_{(i,j)}(t))  = \frac{\int_{t'_j}^t \gamma(u-t'_j) du}{\int_{t''_j}^t \gamma(u-t''_j)  du} 
\]

\end{proof}


\end{document}